%% file: article.tex
\documentclass[12pt,letter]{article}

\usepackage{graphicx}
\graphicspath{{./}{Courbes/}}

\usepackage{amsmath}

\usepackage[english]{babel}
\usepackage[latin1]{inputenc}
\usepackage{url}
\usepackage{latexsym}

\newtheorem{example}{Example}[section]
\newtheorem{theoreme}[example]{Theorem}
\newtheorem{corollaire}[example]{Corollary}
\newtheorem{proposition}[example]{Proposition}
\newtheorem{lemme}[example]{Lemma}

\newcommand{\qed}{\hfill$\Box$ \vspace{0.5 cm}}
\newenvironment{proof}{{\it Proof~: }}{\qed}

\newcommand{\ie}{i.e.}
\newcommand{\moy}[1]{\ensuremath{\langle #1 \rangle}}
\newcommand{\somme}[2]{\ensuremath{\sum_{#1}^{#2}}}
\newcommand{\pk}[2][k]{\ensuremath{p_{#1}(#2)}}

\newcommand{\harmo}[2]{\ensuremath{H_{#1}^{(#2)}}}
\newcommand{\defharmo}{%
where $\harmo{K}{\alpha} = \somme{k=1}{K} k^{-\alpha}$
is the $K$-th harmonic number for $\alpha$
}

\newcommand{\forPoissonFinis}{For large Poisson networks with $N$ nodes and average degree $z$}
\newcommand{\forPoisson}{For Poisson networks with size tending towards infinity and average degree $z$}
\newcommand{\forVsfFinis}{For large discrete power-law networks with $N$ nodes and exponent $\alpha$}
\newcommand{\forVsf}{For discrete power-law networks with size tending towards infinity and
exponent $\alpha$}
\newcommand{\forSfFinis}{For large continuous power-law networks with $N$ nodes,
exponent $\alpha$ and
minimal degree $m$}
\newcommand{\forSf}{For continuous power-law networks with size tending towards infinity,
exponent $\alpha$
and minimal degree $m$}

\newcommand{\refplotsseuil}{
For technical details on our plots, on the computation of thresholds,
and for discussions on the origins of differences
between experiments and predictions, see Section~\ref{sec_plots}.
}

\newcommand{\refplotscc}{
For technical details on our plots
see Section~\ref{sec_plots}.
}

\newcommand{\lefttoright}{}

\newcommand{\degc}[4]{
\begin{table}[!h]
\centering
\begin{tabular}{c||cc|cc||cc|cc}
 &
\multicolumn{2}{c|}{continuous} &
\multicolumn{2}{c||}{Poisson} &
\multicolumn{2}{c|}{discrete} &
\multicolumn{2}{c}{Poisson} \\
 &
\multicolumn{2}{c|}{power-law} &
\multicolumn{2}{c||}{} &
\multicolumn{2}{c|}{power-law} &
\multicolumn{2}{c}{} \\
$\alpha$ &
prev. & exp. &
prev. & exp. &
prev. & exp. &
prev. & exp.\\
\hline
2.5 & #3\\
  3 & #4
\end{tabular}
\caption{Values of the threshold for #2 on discrete and continuous power-law networks of exponents $2.5$ and $3$, and on Poisson networks having the same average degree (see~Table~\ref{tab_cases}). The values are the analytic previsions at the infinite limit and the ones obtained for experiments with networks of $N=100\ 000$ nodes.}
\label{tab_#1}
\end{table}}

\newcommand{\degcn}[4]{
\begin{table}[!h]
\centering
\begin{tabular}{c||cc|cc||cc|cc}
 &
\multicolumn{2}{c|}{continuous} &
\multicolumn{2}{c||}{Poisson} &
\multicolumn{2}{c|}{discrete} &
\multicolumn{2}{c}{Poisson} \\
 &
\multicolumn{2}{c|}{power-law} &
\multicolumn{2}{c||}{} &
\multicolumn{2}{c|}{power-law} &
\multicolumn{2}{c}{} \\
$\alpha$ &
bound & exp. &
bound & exp. &
bound & exp. &
bound & exp.\\
\hline
2.5 & #3\\
  3 & #4
\end{tabular}
\caption{Values of the threshold for #2 on discrete and continuous power-law networks of exponents $2.5$ and $3$, and on Poisson networks having the same average degree (see~Table~\ref{tab_cases}). The values  are the ones obtained for experiments with networks of $N=100\ 000$ nodes, and the theoretical upper bounds.}
\label{tab_#1}
\end{table}}

\newcommand{\scaletroisres}{0.43}

\newcommand{\includetroisres}[1]{
\mbox{}\hspace{-0.4cm}\includegraphics[scale=\scaletroisres]{er_#1}\hspace{-0.7cm}
\includegraphics[scale=\scaletroisres]{sf_#1}\hspace{-0.7cm}
\includegraphics[scale=\scaletroisres]{vsf_#1}
}

\newcommand{\figreal}[2]{
\begin{figure}[h!]
\centering
\includegraphics[scale=0.4]{#1.distrib.eps}
\hspace*{0.1cm}
\includegraphics[scale=0.4]{#1.sommets.eps}
\hspace*{0.1cm}
\includegraphics[scale=0.4]{#1.aretes.eps}
\caption{{\em #2} network.
From left to right: degree distribution, node removals and link removals.}
\label{fig_#1}
\end{figure}
}

\begin{document}

\graphicspath {{PLOT/}}


\input{title}

\input{abstract}
\vskip 0.2cm

\input{intro}

\input{preliminaires}

\input{random}

\input{attack}
\input{real_graphs}

\input{conclu}

\input{ack}

\bibliographystyle{plain}
\bibliography{xbib}

\end{document}

%% file: title.tex
\begin{center}
{\LARGE \bf
Impact of Random Failures and Attacks on\\
\vskip 0.3cm
Poisson and Power-Law Random Networks
}
\vskip 0.5cm
{\large
Cl\'emence Magnien\,\footnote{\textsc{lip6} -- {\sc cnrs}
-- Universit\'e Pierre et Marie Curie -- 104 avenue du pr\'esident Kennedy,
750016 Paris, France,},
Matthieu Latapy\,\footnotemark[1]
and
Jean-Loup Guillaume\,\footnotemark[1]
}
\vskip 0.2cm
\small
Contact: clemence.magnien@lip6.fr
\vskip 0.2cm
\end{center}


%% file: abstract.tex
\begin{abstract}
It appeared recently that the underlying degree distribution of
networks may play a crucial role concerning their robustness.
Empiric and
analytic results have been obtained, based on asymptotic and
mean-field approximations. Previous work insisted on the fact
that power-law degree distributions induce high resilience to
random failure but high sensitivity to attack strategies,
while Poisson degree distributions are quite sensitive in both
cases. Then, much work has been done to extend these results.

We aim here at
studying in depth these results,
their origin, and 
limitations.
We review in detail previous contributions and give
full proofs in a unified framework, and identify the approximations on
which these results rely. We then present
new results aimed at enlightening some important
aspects. We also provide extensive rigorous experiments which help
 evaluate the relevance of the analytic results.

We reach the conclusion that, even if  the basic
results of the field are clearly true and important, they are in practice much
less striking than generally thought. The differences between random
failures and attacks are not so huge and can be explained with simple
facts. Likewise, the differences in the behaviors induced by power-law
and Poisson distributions are not as striking as often claimed.

\medskip
\noindent
Categories and Subject Descriptors:
A.1 [{\bf Introductory and Survey}];
C.2.1 [{\bf Com\-pu\-ter-Communication Networks}]: Network Architecture and Design -- {\em Network topology};
G.2.2 [{\bf Discrete Mathematics}]: Graph Theory -- {\em Network Problems};

\medskip
\noindent
General Terms: Experimentation, Reliability, Security

\end{abstract}


%% file: intro.tex
\subsection*{Introduction.}

It has been shown recently, see for instance
\cite{albert02statistical,dorogovtsev02evolution,newman03structure,strogatz01exploring,watts98collective},
that most real-world complex networks have non-trivial properties in
common.
In particular, the degree distribution
(probability $p_k$ that a randomly chosen node has $k$ links, for each
$k$) of most real-world complex networks is heterogeneous
and well fitted by a power-law, \ie{} $p_k \sim k^{-\alpha}$, with
an exponent $\alpha$ between $2$ and $3$ in general.
This property has been observed in many cases, including internet and web
graphs \cite{faloutsos99powerlaw,govindan00heuristics,crit_int3,IPconn,pastor01dynamical,magoni01analysis,chang01inferring,pastor01dynamical,crit_int2,crit_int1,chen02origin,barabasi02modeling,bu02distinguishing,albert99diameter,Adamic2000Web,kumar99trawling,broder00graph,leonardi03},
social networks
\cite{Liljeros2001SexualContacts,newman01scientific1,newman01scientific2,ebel02,jin01structure},
and biological networks
\cite{kohn99cell,uetz2000protein,jeong00largescale,farkas03yeast}.

In most of these cases, the existence of a path in the network from
most nodes to most others, called {\em connectivity}, is a crucial
feature.
For instance, in the case of the internet, it means that computers can
communicate; in the case of the web if means that one may reach most
pages from most others by following hyperlinks; and in the case of
social networks it conditions the ability of information and diseases
to spread.
Note that
connectivity may be a desirable feature (in the case of the internet for instance),
or an unwanted one
(in the case of virus propagation, for instance), depending on the
application under focus.

Networks are subject to damages (either accidental or not) which
may affect connectivity. For instance, failures may occur on
computers on the internet, causing removal of nodes in internet and web
graphs.  Likewise, in social networks, people can die from a disease,
or people deemed likely to propagate the disease can be vaccinated,
which corresponds to node removals.
For the study of these phenomena, accidental failures may
be modeled by removals of random nodes and/or links in the considered network, while attacks may
be modeled by removals following a given strategy.

\medskip

Networks of different natures may behave differently when one removes
nodes and/or links.
The choice of the removed nodes or links may also influence
significantly the obtained behavior. It has been confirmed that this
is indeed the case in the famous paper \cite{albert00error}, in which
the authors consider networks with Poisson and power-law degree
distributions\,\footnote{%
More specifically, they considered Erd\"os-R\'enyi random graphs and graphs obtained
with the preferential attachment model~\cite{barabasi99emergence}, which are not representative of
all networks with Poisson or power-law degree distributions.
In particular networks obtained with the preferential attachment model are not equivalent
to the random power-law networks we consider in this paper, see Section~\ref{sec_net}.},
and then remove nodes either randomly (failures) or in decreasing
order of their degree (attacks). They measure the size of the largest connected
component (\ie\ the largest set of nodes such that there is a path in the network
from any node to any other node of the set) as a function of the fraction of removed
nodes.

The authors of~\cite{broder00graph} had pursued the same kind of idea
earlier. They tried to establish whether the connectivity of the web
is mainly due to the (very popular) pages with a very large number of
incoming links by studying the connectivity of the web graph from
which the links going to these pages have been removed.

The authors of~\cite{albert00error} obtained the results presented in
Figure~\ref{fig_ps_as_cc},  which shows two things:
the removal strategies play an important role,
and the two kinds of networks behave significantly differently:
in particular, it seems that networks with power-law degree
distributions are
 very resilient to random failures, but very
sensitive to attacks. This particularity is now referred as the {\em
Achille's heel of the internet}~\cite{nature00cover,Barabasi2003Handbook}.
In the case of a social network on which one wants to design
vaccination strategies, it means that one may expect better
efficiency by vaccinating people with the highest number of acquaintances than with random
vaccination~\cite{Dezso2002virusScaleFree,Vespignani2002Immuniztion,Cohen2003Immunization,holme04efficient}.

\begin{figure}[!h]
\begin{center}
\includetroisres{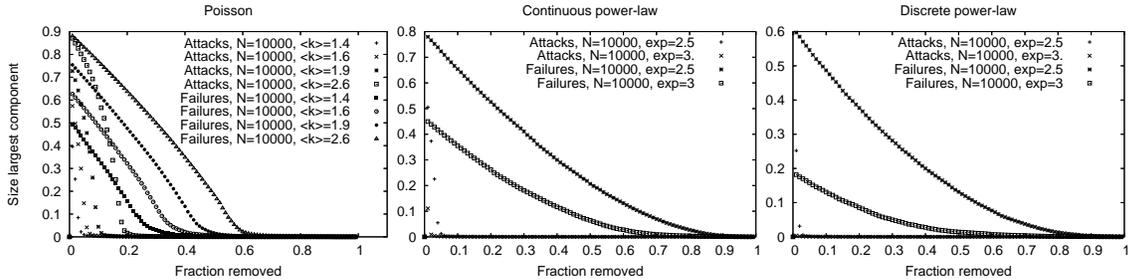}
\caption{Size of the largest connected component as a function
of the fraction of randomly removed nodes (failures) and nodes removed
in decreasing order of node degrees (classical attacks).
\lefttoright
We will define properly these different kinds of networks in
Section~\ref{sec_net}.
\refplotscc
\label{fig_ps_as_cc}
}
\end{center}
\end{figure}

Since then, much work has been done to extend this initial result.
Other kinds of failures and attacks, in particular cascade ones, as
well as other kinds of topologies, have been studied. See for
instance~\cite{lee05robustness,crucitti04error,crucitti04model,newth04evolving,Broido2002Resilience,park03static,Zhou2004RobustnessAS,Holme2002Vulnerability,Flajolet2002Robustness,Motter2002RangeAttacks,motter02cascade,Motter2004Cascade,Pertet2005Cascading,Zhao2004Cascading}.
Some studies introduced other criteria for measuring the state of the
network, see for
instance~\cite{Latora2001Efficient,Crucitti2003EfficiencyScale-Free,Broido2002Resilience,park03static}.
Cases where the underlying networks have non-trivial degree
correlations have also been studied, see for
instance~\cite{Boguna2003EpidemiCorrelations,Vazquez2003ResilienceCorrelations}.
Recent studies focused on the identification and design of robust
topologies, or repair strategies, see for
instance~\cite{valente04peak,paul04optimizarion,paul04optimizarionbis,tanizawa05optimization,costa04reinforcing,chi05stability,holme04efficient,Gallos2004Tolerance,gallos05stability,rezaei05disaster,Raidl2002Augmentation,Beygelzimer2004Improving,Beygelzimer2005EdgeModification}. More details are given in section~\ref{sec_prelim_modeling}.

\medskip

Important efforts have been made to give analytic
results based on mean-field
approximations completing the experimental ones, see in
particular~\cite{Cohen2000RandomBreakdown,Cohen2001Attack,Cohen2003Handbook,Newman2000Robustness,Newman2003Handbook,albert02statistical}.
Our aim is to present these results in detail, to deepen them with new results, and to discuss
their implications.

A significant part of this paper is therefore devoted to detailing
the existing proofs of previously known results.
Indeed, these proofs rely strongly
on mean-field approximations which are classical in statistical
mechanics, but quite unusual in computer science. They are therefore
only sketched in the original papers, and many approximations are
made implicitly.
Therefore we deem in important  to give proofs with full
details and explicit approximations.

The original papers moreover focus on specific cases of interest.  We
give here a unified and complete view of the questions
under concern,  including some new results,
which makes it possible to deepen significantly our
understanding of the field. In particular, we give results
concerning  failures and attacks on links, as well as results on finite
cases, which have received little attention. We also compare the
two different approaches for the study of power-law networks proposed
in~\cite{Cohen2001Attack,Cohen2000RandomBreakdown,Cohen2003Handbook}
and in~\cite{Newman2000Robustness,Newman2003Handbook}.

This paper may therefore be considered as an in-depth and didactic
survey of the main current results of the field, with the aim of
unifying the different approaches and questions that have been studied,
which leads to the introduction of some
 new results.

\medskip

This paper is organized as follows. Section~\ref{sec_prelim} is devoted to
 preliminaries, which consist of definitions and models, of methodological
discussions, and of some preliminary results.
In particular, the approximations made in the proofs in
this paper are presented and explained.  Sections~\ref{sec_p}
and~\ref{sec_a} deal with failures and attacks,
respectively. We present classical results of the field,
as well as several new results which aim at improving our understanding of the
phenomena under concern. We discuss the behavior of several
real-world complex networks in Section~\ref{sec_real}, and compare
them to expected behaviors from theory. Finally, we give an in-depth
discussion and synthesis of our understanding of the field in
Section~\ref{sec_conclu}.


%% file: preliminaires.tex
\section{Preliminaries.}
\label{sec_prelim}

Before entering the core of the paper, we need some important
preliminaries. They consist of preliminary definitions and results, mainly concerning probability distributions and the models we will
consider, but also of methodological aspects. This section should be
read carefully before the rest of the paper since important notions are
introduced and discussed here.

Let us insist on the fact that
most results in this paper are obtained using {\em approximations},
aimed at simplifying the computation. These
approximations are valid in the limit of networks with large sizes.
 They typically consist of neglecting the difference between $N$ and $N \pm n$ when $n$ is
small compared to $N$, or of supposing that random values are equal
to the average. More subtle approximations are also done,
belonging to the {\em mean-field} approximation framework, classical
in statistical mechanics and widely used in the context
of complex networks~\cite{Vespignani2005TCS,Moore2005Traceroute,Newman2001RandomGraphs}.

It is important to understand that the proofs we provide are
valid only in this framework,
as we have no formal guarantee that all approximations are valid.
This is why we will
always explicitly point out the approximations we make, and we
will always compare analytic results to experiments.
Moreover, we believe that efforts should be made to obtain exact results
and proofs in this context: most results presented here are currently
beyond the areas to which exact methods have been applied with success.
Presenting exact results is however out of the scope of this paper.

\subsection{Poisson and power-law distributions.}

A probability distribution is given by the probability $p_k$, for all $k$,
that the considered value is equal to $k$.
The sum of all $p_k$ must be equal to 1. 
A Poisson distribution is characterized by $p_k = e^{-z}\frac{z^k}{k!}$,
where $z$ is the average value of the distribution. The probability
of occurrence of a  value $x$ in such a distribution therefore decays
exponentially with its difference $|z-x|$ to the average, which
means that, in practice, all the values are close to this
average.

A power-law distribution with exponent $\alpha$ is such that
$p_k$ is proportional $k^{-\alpha}$. 
In such distributions, the probability of
occurrence of a value $x$ decays only polynomially with $x$. This
implies that, though most values are small, one may obtain very large
values.
In the whole paper, we will generally consider exponents between $2$
and $4$, which are the relevant cases in our context (see
Section~\ref{sec_prelim_model}), but we will also state some results
valid out of this range.

We will consider here two types of power-law distributions,
which are the most widely used in the literature: {\em discrete}
and {\em continuous} power-law distributions. They are both
defined by their exponent $\alpha$ and their minimal value
$m$.

The corresponding continuous power-law is a Pareto distribution,
such that $\int_m^{\infty} C x^{-\alpha} \mathrm{d}x=1$.
$C$ is a normalization constant that we can compute:
$\int_m^{\infty} C\ x^{-\alpha} \mathrm{d}x
= C\ \frac{m^{-\alpha+1}}{\alpha-1} = 1$. We then obtain
$C = m^{\alpha-1}(\alpha-1)$.
To obtain discrete values, we then take $p_k$ equal to
$\int_k^{k+1} C x^{-\alpha} \mathrm{d}x$,  which is proportional
to $k^{-\alpha}$ in the limit where $k$ is large\,\footnote{
  One can also define $p_k$ to be proportional to
  $\int_{k-1/2}^{k+1/2} x^{-\alpha}\mathrm{d}x$,
  see~\cite{Cohen2001Reply}.
  This has little impact on the obtained results.
}.
And finally, $p_k = m^{\alpha-1}(\alpha-1)\int_k^{k+1} x^{-\alpha} \mathrm{d}x = m^{\alpha-1}(k^{-\alpha+1} - (k+1)^{-\alpha+1})$. We will
mainly use this form in the sequel but at some points we will switch
back to the continuous form.


The corresponding discrete power-law distribution is
$p_k = \frac{1}{C}\ k^{-\alpha},\ k \ge m$,
where $C=\sum_{k=m}^\infty k^{-\alpha}$ is the normalization constant
necessary to ensure that each $p_k$ is between $0$ and $1$ and that
their sum is $1$. In such a distribution, therefore, $p_k$ is exactly
proportional to $k^{-\alpha}$ for all $k \ge m$. In order to simplify the
computation, we will always take $m=1$ for discrete power-law
distributions in this paper. This implies that $C=\zeta(\alpha)$,
where $\zeta$ is the Riemann zeta function defined for $\alpha>1$ by
$\zeta(\alpha) = \sum_{k=1}^\infty k^{-\alpha}$.
Then,
$p_k = \frac{1}{\zeta(\alpha)}\ k^{-\alpha}$.

Discrete and continuous power-law distributions
each have their own advantages and drawbacks. For instance, continuous
power-laws are easier to use in experiments than discrete ones, which
themselves are more rigorous than continuous ones for small values.
For a more complete discussion on the advantages and drawbacks of
discrete and continuous distributions, see for
instance~\cite{Dorogovtsev2001Comment,Cohen2001Reply}.
We will use both of them in the sequel, and discuss their differences.

\subsubsection*{Bounded distributions}
\label{secK}


Given a distribution $p_k$ as defined above, one may sample a finite
number $N$ of values from it.
In  such a sample, there is a maximal value $K$.
Therefore, the actual distribution of the values in this sample, \ie\
the fraction $p_k(N)$ of values equal to $k$ for each $k$, is slightly
different from the original distribution $p_k$.
In particular, for all $k>K$, $p_k(N)=0$ (while in general $p_k > 0$).
We will therefore call these distributions {\em bounded distributions}.
The difference between bounded and unbounded distributions goes
to zero when $N$ tends towards infinity, but for any finite value of
$N$ it may play a role in our observations.

We detail below important properties of bounded distributions, starting with
their expected maximal value $K$.

The maximal value $K$ among a sample of $N$ values from a given distribution $p_k$ is
a random variable, and it is possible to give the exact formula for its expected value:
let $X_1, \ldots, X_N$ be the values sampled from the distribution,
and let $Y = \max_{i=1..N} X_i$.
$Y$ then has the following distribution:
$$P(Y=K) = (\sum_{k=0}^K p_k)^N - (\sum_{k=0}^{K-1} p_k)^N.$$
It is the probability that all values are lesser than or equal to $K$,
minus the probability that all values are lesser than $K$,
\ie{} the probability that all values are lesser than or equal to $K$ and
at least one is equal to $K$,
and its expected value is given by:
$$E[Y] = \sum_{k=0}^\infty k P(Y=k).$$
However, deriving numerical values from this formula is too intricate.
We will therefore use an approximation:

\begin{lemme} \cite{Cohen2000RandomBreakdown}
\label{lem_K_general}
For a given distribution $p_k$ such that \mbox{$p_k>0$} for all $k$,
the expected maximal value $K$ among a sample of $N$ values can be approximated by
$$\somme{0}{K-1} p_k = 1 - \frac{1}{N}.$$
\end{lemme}
\begin{proof}{
The claim is equivalent to $\somme{K}{\infty} p_k = \frac{1}{N}$,
which means that $K$ is such that there is only one value larger
than $K$ in the sample.
Moreover this value must be exactly equal to $K$, otherwise there
would be only one value larger than $K+1$ and we would have
$\sum_{K+1}^\infty p_k = \frac{1}{N}$, which is impossible since
$p_k>0$ for all $k$.
}\end{proof}

We can apply this result to the three cases of interest:
\begin{lemme}
\label{lem_K_er}
For a Poisson distribution with average value $z$, the expected
maximal value $K$ among a sample of $N$ values can be approximated by
$$
\somme{0}{K-1} \frac{z^k}{k!} = e^z\left(1-\frac{1}{N}\right).
$$
\end{lemme}
\begin{proof}{
Direct application of Lemma~\ref{lem_K_general} with
$p_k = e^{-z}\frac{z^k}{k!}$.
}\end{proof}




\begin{lemme} \cite{Cohen2000RandomBreakdown}
\label{lem_K_sf}
For a continuous power-law with exponent $\alpha$ and minimal value $m$,
the expected maximal value $K$ among a sample of $N$ values can be approximated by
$K = m N^\frac{1}{\alpha-1}.$
\end{lemme}
\begin{proof}{
From Lemma~\ref{lem_K_general}, $K$ satisfies
$\sum_1^{K-1} p_k = 1 - \frac{1}{N}$.
Therefore, $\frac{1}{N} = \somme{K}{\infty} p_k$. We have that
$p_k 
     = (\alpha-1)m^{\alpha-1}\int_k^{k+1} x^{-\alpha}\mathrm{d}x$.
Therefore,
$\frac{1}{N} = (\alpha-1)m^{\alpha-1} \int_K^\infty x^{-\alpha}\mathrm{d}x
             = m^{\alpha-1} K^{-\alpha+1}$.
The result follows directly.
}\end{proof}

\begin{lemme}
\label{lem_K_vsf}
For a discrete power-law with exponent $\alpha$, the expected
maximal value $K$ among a sample of $N$ values can be approximated by
$\zeta(\alpha)(1-\frac{1}{N}) = \harmo{K-1}{\alpha}$,
\defharmo{}.
\end{lemme}
\begin{proof}{
Direct application of Lemma~\ref{lem_K_general} with
$p_k = \frac{k^{-\alpha}}{\zeta(\alpha)}$.
}\end{proof}

These results may be used in practice to approximate the expected maximal
value $K$ among a sample of $N$ values. For Lemmas~\ref{lem_K_er} and~\ref{lem_K_vsf},
it is obtained iteratively by setting $K=0$ and increasing it until
$\sum_{k=0}^K p_k \ge 1-\frac{1}{N}$.

Figure~\ref{fig_K_100000} plots the estimates of the maximal
value for samples of size $N=100\,000$, for the three types of distributions we consider,
 obtained
from the results above, together with experimental values obtained by
computing the average of the maximum values of $1\,000$ sets of $N$ random values.

In the case of power-law distributions, our approximations
underestimate slightly the experimental values.
For Poisson distributions, the evaluation fits experiments exactly, but for
some precise values of the average only. This is due to the fact that $K$ can only take
integer values in the evaluation: it is actually the first integer such that
$\somme{k=0}{K} p_k > \frac{N-1}{N}$.
Therefore, this sum may sometimes be significantly
larger than $\frac{N}{N-1}$, and then the evaluation of $K$ is poor.
To evaluate this bias, we have plotted the relative error
$\left(\somme{k=0}{K} p_k\right) \frac{N-1}{N}$ in
Figure~\ref{fig_K_100000}~(left).

\begin{figure}[!h]
\begin{center}
\includegraphics[scale=0.47]{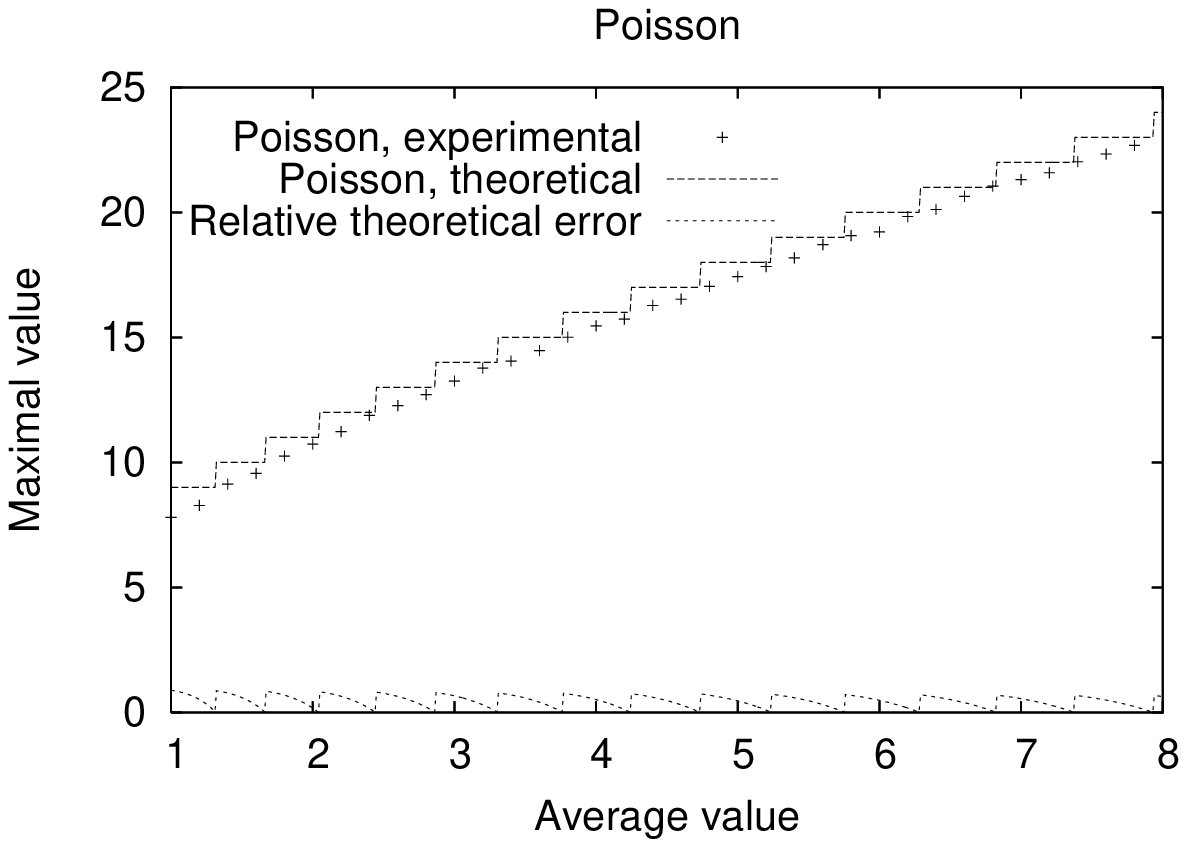}
\ \ \ \ \
\includegraphics[scale=0.47]{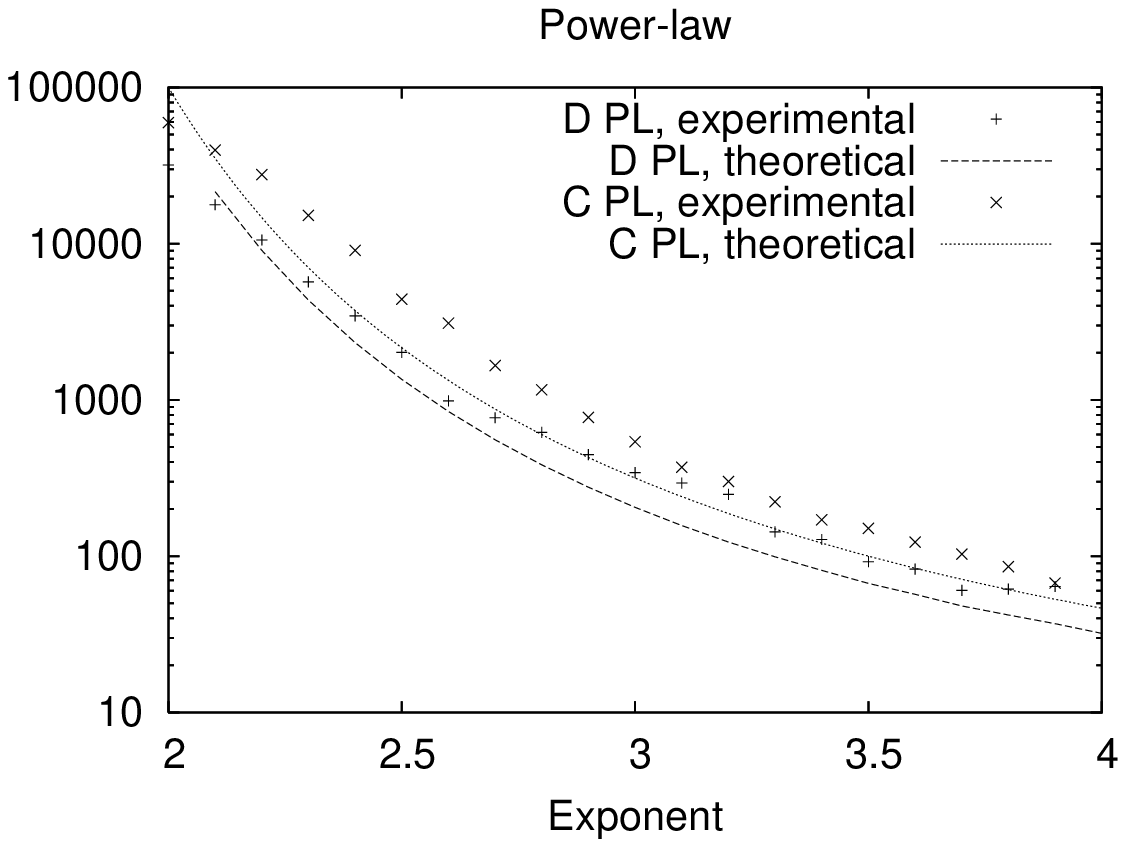}
\caption{\label{fig_K_100000}
Analytic and experimental estimates of the expected maximal among
$N=100\,000$ sampled values.
Left: for Poisson distributions, as a function of the average value;
right: for discrete and continuous power-law distributions, as a function
of the exponent.}
\end{center}
\end{figure}

All in all, our approximations for the expected maximal value are quite accurate,
and we will use them in the rest of the paper.

An important point here is to notice that, for all the distributions we
consider here, the expected maximal among $N$ sampled values grows
sublinearly with $N$.
For Poisson distributions this is obvious.
Lemma~\ref{lem_K_sf} explicitly states it for
continuous power-laws, and one may also check
the discrete power-law case.
As we will see, this is important for some approximations we will make in
the following, and for some results in Section~\ref{sec_na}.
\label{pag_K}

\bigskip

Until now, we discussed the fact that sampling $N$ values from a distribution induces an
expected maximal value. But it also induces an expected distribution of
the $N$ values, denoted by $p_k(N)$, which is different from the
original distribution $p_k$. We now study more precisely this expected
distribution.

\begin{lemme}
\label{lem_distdeg_finis_general}
The expected distribution $p_k(N)$ of $N$ values sampled from a given
distribution $p_k$ can be approximated, for all $k\le K$, by
$$
p_k(N) = \frac{N}{N-1}\ p_k
$$
where $K$ is the expected maximal value, related to $p_k$ and $N$ by
Lemma~\ref{lem_K_general}.
\end{lemme}
\begin{proof}{
Since we sample values from a truncated distribution, we must have that $p_k(N)$ is proportional to $p_k$ for all $k\le K$:
$p_k(N) = C\ p_k$, and that the sum of all $p_k(N)$ is $1$:
$\sum_{k=0}^{\infty} p_k(N)= 1$. Moreover, we know from
Lemma~\ref{lem_K_general} that $\somme{K}{\infty} p_k = \frac{1}{N}$.
We obtain $1 = \somme{k=0}{\infty} p_k(N) = \somme{k=0}{K} p_k(N) =
C \somme{k=0}{K} p_k = C (1-\frac{1}{N})$, where we neglected the difference
between $\somme{K}{\infty} p_k$ and $\somme{K-1}{\infty} p_k$. The
claim follows.
}\end{proof}

\begin{lemme}
\label{lem_distdeg_finis_er}
For a Poisson distribution with average value $z$, the expected
distribution $p_k(N)$ of a sample of $N$ values can be approximated, for all $k\le K$, by
$$
p_k(N) = \frac{N}{N-1}\ \frac{e^{-z}z^k}{k!},
$$
where $K$ is the expected maximal value, related to $p_k$ and $N$ by
Lemma~\ref{lem_K_er}.
\end{lemme}
\begin{proof}{
Direct application of Lemma~\ref{lem_distdeg_finis_general}
with $p_k = e^{-z}\frac{z^k}{k!}$.
}\end{proof}

\begin{lemme}
\label{lem_distdeg_finis_sf}
For a continuous power-law distribution with exponent $\alpha$ and minimal
value $m$, the expected distribution $p_k(N)$ of a sample of $N$
values can be approximated, for all $m\le k\le K$, by
$$
p_k(N) = \frac{N}{N-1} m^{\alpha-1}\left(k^{-\alpha+1} - (k+1)^{-\alpha+1}\right),
$$
where $K$ is the expected maximal value, related to $p_k$ and $N$ by
Lemma~\ref{lem_K_sf}.
\end{lemme}
\begin{proof}{
Direct application of Lemma~\ref{lem_distdeg_finis_general}
with $p_k = m^{\alpha-1}(k^{-\alpha+1} - (k+1)^{-\alpha+1})$.
}\end{proof}

\begin{lemme}
\label{lem_distdeg_finis_vsf}
For a discrete power-law distribution with exponent $\alpha$,
the expected distribution $p_k(N)$ of a sample of $N$ values can be approximated, for
all $k\le K$, by
$$
p_k(N) = \frac{N}{N-1}\ \frac{k^{-\alpha}}{\zeta(\alpha)}
=\frac{k^{-\alpha}}{\harmo{K-1}{\alpha}},
$$
\defharmo{},
and $K$ is the expected maximal value, related to $p_k$ and $N$ by
Lemma~\ref{lem_K_vsf}.
\end{lemme}
\begin{proof}{
Direct application of Lemma~\ref{lem_distdeg_finis_general},
with $p_k = \frac{k^{-\alpha}}{\zeta(\alpha)}$.
}\end{proof}

The results above give a precise description of what one may
expect from finite samples from Poisson  and power-law distributions.
They will be useful  when dealing with finite networks below.

\subsubsection*{First moments of a distribution}

The average $\moy{k} = \somme{k=0}{\infty} k p_k$ of a distribution
$p_k$ is also called its {\em first moment}, the $i$-th moment being
defined as $\moy{k^i} = \somme{k=0}{\infty} k^i\ p_k$. In the
continuous case, the $i$-th moment is similarly
defined as $\moy{k^i} = \int_{k=m}^{\infty} k^i\ p_k$.
 In the whole paper,
the first and second moments will play a central role. We
present here the results we will need about them. Namely, we give
formul\ae{} for the Poisson and power-law cases, both in the infinite
case and in the case of a sample of finite size $N$.

\begin{lemme}
\label{lem_moments_finis_er}
For a Poisson distribution with average value $z$, the first
two moments of the expected distribution of a sample of $N$ values can be approximated by
$$\moy{k} = \left(\frac{N}{N-1} \right)
\somme{k=0}{K}\frac{e^{-z}z^k}{(k-1)!}
\mbox{\ \ \ and \ \ \ }
\moy{k^2} = \left( \frac{N}{N-1} \right)
\somme{k=0}{K}\frac{ke^{-z}z^k}{(k-1)!},
$$
where $K$ is the expected maximal value, related to
$p_k$ and $N$ by Lemma~\ref{lem_K_er}.
\end{lemme}
\begin{proof}{
Direct application of Lemma~\ref{lem_distdeg_finis_er}.
}\end{proof}

\begin{lemme}
\label{lem_moments_er}
For a Poisson distribution with average value $z$, the first
two moments are
$$
\moy{k}=z \mbox{ and } \moy{k^2} = z^2 + z.
$$
\end{lemme}
\begin{proof}{
Direct computation, with $p_k = e^{-z}\frac{z^k}{k!}$.
}\end{proof}

\begin{lemme}\cite{Cohen2000RandomBreakdown}
\label{lem_moments_finis_sf}
For a continuous power-law distribution with exponent $\alpha$ and
minimal value $m$, the first two moments of the expected distribution
of a sample of $N$ values can be approximated by
$$
\begin{array}{lclclclcl}
\moy{k} & = & m^{\alpha-1}K^{-\alpha+2}\ \frac{\alpha-1}{-\alpha+2} &
\mbox{ and }  &
\moy{k^2} & = & m^{\alpha-1}K^{-\alpha+3}\frac{\alpha-1}{-\alpha+3} &
\mbox{ if } & 1<\alpha<2,\\
\moy{k} & = & m\ \frac{\alpha-1}{\alpha-2} &
\mbox{ and }  &
\moy{k^2} & = & m^{\alpha-1}K^{-\alpha+3}\ \frac{\alpha-1}{-\alpha+3} &
\mbox{ if } & 2<\alpha<3,\\
\moy{k} & = & m\ \frac{\alpha-1}{\alpha-2} &
\mbox{ and }  &
\moy{k^2} & = & m^2\ \frac{\alpha-1}{\alpha-3} &
\mbox{ if } & \alpha>3,
\end{array}
$$
where $K$ is related to $p_k$ and $N$ by Lemma~\ref{lem_K_sf}.
\end{lemme}
\begin{proof}{
If we approximate $\frac{N-1}{N}$ by $1$, we obtain
directly in the continuous case:
$\moy{k} = m^{\alpha-1}(\alpha-1)\int_m^K x x^{-\alpha} \mathrm{d}x
         = \frac{(\alpha-1)m^{\alpha-1}}{-\alpha+2} (K^{-\alpha+2} - m^{-\alpha+2})$
and
$
\moy{k^2} = m^{\alpha-1}(\alpha-1)\int_m^K x^2 x^{-\alpha} \mathrm{d}x
          = \frac{(\alpha-1)m^{\alpha-1}}{-\alpha+3} (K^{-\alpha+3} - m^{-\alpha+3})
$.\\
Moreover, when $N$ is large, we have $K\gg m$ and we can approximate
$K^{-\alpha+2} - m^{-\alpha+2}$ by $K^{-\alpha+2}$
and
$K^{-\alpha+3} - m^{-\alpha+3}$ by $K^{-\alpha+3}$
if $1 < \alpha < 2$;
$K^{-\alpha+2} - m^{-\alpha+2}$ by $-m^{-\alpha+2}$
and
$K^{-\alpha+3} - m^{-\alpha+3}$ by $K^{-\alpha+3}$
if
$2 < \alpha < 3$; and
$K^{-\alpha+2} - m^{-\alpha+2}$ by $-m^{-\alpha+2}$
and
$K^{-\alpha+3} - m^{-\alpha+3}$ by $-m^{-\alpha+3}$
if $\alpha > 3$.
The results follow.
}\end{proof}

\begin{lemme} \cite{Cohen2000RandomBreakdown}
\label{lem_moments_sf}
For a continuous power-law distribution with exponent $\alpha$ and
minimal value $m$, the first two moments are
$$
\moy{k} = m\ \frac{\alpha-1}{\alpha-2} \mbox{\ \  if\ \  } \alpha > 2
\mbox{\ \ \ \ \ and\ \ \ \ \ }
\moy{k^2} = m^2\ \frac{\alpha-1}{\alpha-3} \mbox{\ \  if\ \  } \alpha > 3,
$$
and they diverge in all the other cases.
\end{lemme}
\begin{proof}{
Direct application of Lemma~\ref{lem_moments_finis_sf} with $K$
tending towards infinity.
}\end{proof}

\begin{lemme}
\label{lem_moments_finis_vsf}
For a discrete power-law distribution with exponent $\alpha$, the first
two moments of the expected distribution of a sample of $N$ values can be approximated by
$$
\moy{k} = \frac{\harmo{K}{\alpha-1}}{\harmo{K-1}{\alpha}}
\mbox{\ \ \  and\ \ \  }
\moy{k^2} = \frac{\harmo{K}{\alpha-2}}{\harmo{K-1}{\alpha}},
$$
\defharmo{},
where $K$ is the expected maximal value, related to
$p_k$ and $N$ by Lemma~\ref{lem_K_vsf}.
\end{lemme}
\begin{proof}{
Direct application of Lemma~\ref{lem_distdeg_finis_vsf}.
}\end{proof}

\begin{lemme} \cite{Newman2003Handbook}
\label{lem_moments_vsf}
For a discrete power-law distribution with exponent $\alpha$, the first
two moments are
$$\moy{k} = \frac{\zeta(\alpha-1)}{\zeta(\alpha)}
\mbox{\ \  and\ \  }
\moy{k^2} = \frac{\zeta(\alpha-2)}{\zeta(\alpha)}.$$
\end{lemme}
\begin{proof}{
Direct computation, with $p_k = \frac{k^{-\alpha}}{\zeta(\alpha)}$.
}\end{proof}

We would like to discuss here the differences between the moments
of bounded and unbounded distributions.
Although the difference between the distributions themselves
is small, both for the Poisson and the power-law case
(the ratio between a bounded and an unbounded distribution is approximately $N/(N-1)$),
this is not the case for the moments of these distributions.
In practice, we can notice that for Poisson distributions,
the values of the first and second moments are almost identical
for bounded and unbounded distributions,
while there is a noticeable difference for power-law distributions.
This can be understood as follows:
these differences are strongly related to the quantities
$\somme{k=K+1}{\infty} kp_k$ and $\somme{k=K+1}{\infty} k^2p_k$.
In both cases, values of $p_k$ for $k>K$ are quite small:
$\somme{k=K}{\infty} p_k = 1/N$.
For Poisson networks, this $1/N$ is distributed among $p_k$ which decrease exponentially,
and $K$ is small.
Therefore the values of $k p_k$ and $k^2 p_k$ for $k>K$ are small.
For power-law networks on the other hand, $K$ is large, and
the probabilities $p_k$ decrease only polynomially,
 therefore values of $k p_k$ and $k^2 p_k$ for $k>K$ are much larger.


These observations will explain in the following why in some
cases theoretical predictions for the finite case and for the
infinite limit
are almost identical for Poisson networks and quite different for
power-law networks.

\bigskip

\noindent
We finally have all the preliminary results we need on distributions;
we can now use them in the context of complex networks.

\subsection{Modeling issues.}
\label{sec_prelim_model}

In this section we detail the models of networks we will consider, then discuss the modeling of failures and attacks we will use. We finally present results concerning the connectivity of random networks, which will play a key role in the sequel.

\subsubsection*{Random networks}
\label{sec_net}

Given an integer $N$ and a distribution $p_k$ one can easily generate
a network taken uniformly at random among the ones having $N$ nodes
and degree distribution $p_k$. Indeed, it is sufficient to sample the
degree of each of the $N$ nodes with respect to $p_k$, then to attach
to each node as many {\em stubs} as its degree, and finally to construct
links by choosing random pairs of stubs.  If the sum of  degrees is odd,
then one just has to sample again the degree of a random node until the sum
becomes even.
This model is known as the
{\em configuration model}~\cite{Bender1978Configuration}
and is widely used in the literature,
see for instance~\cite{bollobas85random,molloy95critical,molloy98size,aiello00random}.
We will
call {\em random
networks} all networks obtained using it\,\footnote{These networks may contain loops (links from one
node to itself) and multiple
links (more than one link between two given nodes) in small quantities,
which we will neglect in our reasoning as explained in
Section~\ref{sec_mfa}.}.

If one chooses a Poisson distribution of average $z$ then one obtains
an equivalent of the Erd\"os-R\'enyi model~\cite{erdos59random}
in which  the network
is constructed from $N$ initially disconnected  nodes by adding
$M = \frac{z\ N}{2}$ links between randomly chosen pairs of nodes. One then
obtains a network taken uniformly at random among the ones having $N$
nodes and $M$ links.

As already discussed in the introduction, and as we will  see all
along this contribution, the degree distribution of a network may
be seen as responsible for some of its most important features (like
robustness). Studying random networks with prescribed degree
distributions is therefore a key issue. Much
work has already been done to this regard, see for
instance~\cite{Newman2003Handbook,aiello00random,PV01,Cohen2003Handbook}.
These networks are particularly well suited for formal analysis, and most
formal results obtained on complex networks in the literature,
including the ones  on robustness, rely on this modeling,
this is why we use it here.
For a more detailed discussion on network modeling and other approaches,
see Section~\ref{sec_prelim_modeling}.

We will focus on three classes of networks, namely the ones with
Poisson, continuous power-law and  discrete power-law degree
distributions, which we will call {\em Poisson networks},
{\em continuous power-law networks} and {\em discrete power-law
networks} respectively.

In our experiments, we will consider Poisson networks with average
degree $z$ between $1$ and $8$, because for $z<1$ the networks do not
have a giant component
(see p.~\pageref{sec_largest_conn} and Lemma~\ref{lem_poisson_conn}),
 and we have observed that the behaviors for
\mbox{$z \ge 8$} are very similar to and easily predictable from the ones
observed for $z=8$. Concerning power-law networks, we will always take
the minimal degree $m$ equal to $1$, which fits most real-world
cases. We will consider exponents between $2$ and $4$ because below
$2$ the average degree is infinite (see Lemmas~\ref{lem_moments_sf}
and~\ref{lem_moments_vsf}) and above $4$ the network has only small
connected components, as we will see below (see Lemmas~\ref{lem_cpl_conn} and~\ref{lem_dpl_conn}).
Moreover, most real-world cases fit in these ranges.

Let us insist finally on the fact that real-world complex networks
may have other properties that influence their robustness, like for
instance correlations between degrees, clustering (local density),
and others.
Capturing these properties in formal models however
remains a challenge.
See section~\ref{sec_prelim_modeling} for more details.


We will however discuss them informally when observing the behavior of real-world networks in
Section~\ref{sec_real}.

\subsubsection*{Failures and attacks}

There are many ways to model various kinds of failures and attacks.
We will focus here on removals of nodes and/or links. We will
suppose that failures are random, in contrast to attacks, which
follow strategies.

{\em Random node failures} are then series of removal of nodes
chosen at random. Equivalently, one may choose a fraction of the
nodes at random and then remove them all. Likewise, {\em random
link failures} consist of series  of removal of links chosen at
random.

Attacks on the other hand follow  a {\em strategy} for removing nodes or/and links
which has to be defined. We then say that we observe an
{\em attack following this strategy}. For instance, we presented
in the introduction the most famous strategy, which consists of
removing nodes  in decreasing order of their degrees. We will
call this the {\em classical attack}, and we will define other
strategies in Section~\ref{sec_a}.

Notice moreover that, when one removes a node, one also removes
all the links attached to it. This leads to  the {\em link point of
view} of node failures and attacks,
 which consists of observing the fraction
of {\em links} removed during {\em node} failures or attacks.

In the sequel we will consider all these situations: random node
or link failures, attacks following various strategies, and link
point of view of node failures and attacks.
%
%
In these various cases we want to observe the resilience of networks, which requires to use a criterion to capture the impact of failures or attacks on a network.
%
We will here consider the size of its largest connected component, or more precisely the fraction of nodes in this component as a function of the number of nodes or links removed. This captures the ability of nodes to communicate, which is central in our context:
the smaller this fraction, the greater the impact of the removals. Notice however that one may use other criteria to measure the impact of failures or attacks, see Section~\ref{sec_prelim_modeling} for more details.


\subsubsection*{Largest connected component}
\label{sec_largest_conn}

In many cases, the largest connected component of a random network
contains most nodes of the network. More precisely, depending on
the underlying degree distribution, the size of the largest connected
component may scale linearly with the size of the network.
The network is then said to have a {\em giant component}.

There actually exists a precise and simple criterion on the degree
distribution of a random network to predict if this network will
have a giant connected component or not. Since most of the results
we will discuss later in this contribution rely on an appropriate use
of this criterion, we recall it here.

\begin{theoreme}
\cite{molloy95critical,aiello00random,Cohen2000RandomBreakdown,Newman2003Handbook}
\label{thseuil}
A random network with  size $N$ tending towards infinity and with
degree distribution $p_k$ such that it has maximal value $K < N^{1/4}$
almost surely has a giant component
if and only if:
$$
\langle k^2 \rangle - 2 \langle k \rangle = \sum_{k=0}^{K} k(k-2) p_k > 0.
$$
\end{theoreme}

This theorem has been rigorously proved in
\cite{molloy95critical,aiello00random} and has been proved in the
mean-field approximation framework in
\cite{Cohen2000RandomBreakdown,Newman2003Handbook}. Detailing these
proofs is out of the scope of this paper.

This result can be applied to the three kinds of networks we consider
here (since their maximal degree is sublinear, as explained in
Section~\ref{secK}, page~\pageref{pag_K}),
which gives the following results\,\footnote{Refer
to Section~\ref{sec_mfa} for the conditions under which the previous
theorem is going to be applied.}.

\begin{lemme}
\label{lem_poisson_conn}
A Poisson network with size tending towards infinity
and average degree $z$ almost surely has a giant component
if and only if $z>1$.
\end{lemme}
\begin{proof}{
Direct application of Theorem~\ref{thseuil} and Lemma~\ref{lem_moments_er}.
}\end{proof}

\begin{lemme}
\label{lem_cpl_conn}
A continuous power-law network with size tending towards infinity,
exponent $\alpha$ and minimal degree $m=1$ almost surely has a giant
component if and only if
$\alpha<4$.
\end{lemme}
\begin{proof}{
Direct application of Theorem~\ref{thseuil} and Lemma~\ref{lem_moments_sf}.
}\end{proof}

\begin{lemme}
\label{lem_dpl_conn}
A discrete power-law network with  size tending towards infinity
and exponent $\alpha$  almost surely has a giant component
if and only if $\alpha$ is such that
$\frac{\zeta(\alpha-2)}{\zeta(\alpha-1)} > 2$.
\end{lemme}
\begin{proof}{
Direct application of Theorem~\ref{thseuil} and Lemma~\ref{lem_moments_vsf}.
}\end{proof}

One may compute the numerical value from this last lemma. One then
obtains the condition $\alpha < 3.48$ for
discrete power-law networks.
In summary, the criterion of Theorem~\ref{thseuil} gives very simple
conditions under which the random networks we consider have a giant component.

\subsection{Mean-field framework and generating functions.}
\label{sec_mfa}

As already emphasized at the beginning of Section~\ref{sec_prelim},
most results in this paper are made using {\em approximations},
valid in the mean-field framework. Most of these approximations are
classical and very simple,
like for instance neglecting small values when compared to large ones, but some are specific to random networks and
deserve more attention. We detail them below. We then present the
generating function framework, which makes it possible to embed these
approximations in a powerful formalism. We finally recall some results
on generating functions which  will be useful in the rest of the
paper.

\subsubsection*{Mean-field approximations in random networks}

The fact that stubs are linked fully at random in a random
network is a feature which has important consequences in our
context. In particular, when one removes a link chosen at
random in such a network, this is equivalent to the removal of
two stubs at random, and so the obtained network is still random (with
a different degree distribution in general). Likewise, when one
removes a node, the obtained network is also random. These simple remarks will be essential in the following.

Mean-field approximations are very helpful in the study of random
networks since they allow to neglect some correlations which would
otherwise be very hard to handle.

Consider for instance the neighbors of a given node,
which we will call {\em source} node, in a large random
network. Suppose that the network is sparse (the probability for
two randomly chosen nodes to be linked together is almost $0$)
and that its maximal degree is small compared to its size, which
will always be true in our context. Then the probability that two
of these neighbors are directly linked together is negligible.
Likewise, if we take
all the nodes at distance $2$ of the source node then the probability
of having a link between two of them is very small and may also be
neglected. So does the probability to have a link from a node at
distance $2$ to more than one node at distance $1$, or to the
source. Continuing this reasoning, the network may be considered
locally as a tree: any subnetwork composed of the nodes at a distance
lower than a given finite value is a tree if the size of the network
tends towards infinity.

The approximation above relies on the fact that we neglect very
small probabilities, or equivalently that we consider the limit
where the size of the network tends towards infinity.

In the same manner, it is known that any random network with a
maximal degree lower than $\sqrt{\moy{k}N}$ almost surely has no
loops or multiple links~\cite{Chung2002Connected,Burda2003Uncorrelated}.
Likewise, Theorem~\ref{thseuil} is formally true only for
networks with maximal degree less than $N^{1/4}$. For both cases the conditions might not be true for all the networks under concern, however we will consider that the networks do not possess loops and multiple
links and that Theorem~\ref{thseuil} can be applied.

The mean-field framework allows another important approximation which comes from the fact that there is no distinction between choosing a stub at random and following a link at random from a random starting node.
Indeed, since links are formed by pairs of randomly chosen
stubs, it makes in principle no difference.

One consequence is that we suppose that there is no correlation between
the degree of a node and the degrees of its neighbors, \ie\ that the
random starting node we choose has no impact on the neighboring node we
will reach. This is indeed true when the maximal degree is
below $N^{\min(1/2, 1/(\alpha-1))}$~\cite{Burda2003Uncorrelated},
but not if the maximal degree is larger. We will neglect the possible correlations, which is classical in the mean-field approach, even if the
above condition is not fulfilled.

This approximation may be used to describe the degree distribution
of neighbors of nodes, in other words the degree of a node reached
by starting from a
randomly chosen node and following one of its links chosen at random.
According to the mean-field approximation above, this is equivalent to
choosing a random stub and therefore the probability that a random
stub belongs to a given node is proportional to this node's degree,
\ie{} the probability of reaching a
node of degree $k$ is proportional to $k\ p_k$.
The sum of these probabilities must be equal to $1$,
we therefore obtain the following
probability:
$\frac{k\ p_k}{\somme{j=0}{\infty} jp_j} = \frac{k\ p_k}{\moy{k}}$.

We can derive from this the probability $q_k$ that a neighbor of a node
has $k$ {\em other} neighbors, which will be useful in the sequel.
It is nothing but the probability that a node obtained by following
a link has $k+1$ neighbors, and so:
\begin{equation}
\label{eqqk}
q_k = \frac{(k+1)p_{k+1}}{\moy{k}}.
\end{equation}

\subsubsection*{Basics on generating functions}

Generating functions, also called formal power series, are
powerful formal objects widely used in mathematics, computer science
and physics. They encode series of numbers $(s_k)_{k\ge 0}$ as
functions $f(x) = \somme{k=0}{\infty} s_k x^k$. Operations on the
series of numbers then correspond to operations on the associated
functions, which often are much more powerful.
See~\cite{Wilf1994Functionology} for a general introduction.

The application of generating functions  to the random network context
is presented in details in~\cite{Newman2001RandomGraphs}. Using them
to encode series of  probabilities (like for instance degree
distributions), the authors show how mean-field approximations may
be embedded with benefit in this formalism. Once this is done, it is
possible to manipulate the associated notions efficiently and easily.
We give an overview of  this approach below, and we refer
to~\cite{Newman2001RandomGraphs} for a detailed and didactic
presentation with illustrations.
We follow the notations in this reference, and we will use them all
along the paper.

Let us begin by encoding the degree distribution $p_k$ by the
following generating function:
\begin{equation}
\label{eqg0}
G_0(x) = \somme{k=0}{\infty} p_k x^k.
\end{equation}
This function is an encoding of the whole distribution since one
may obtain $p_k$ by differentiating it $k$ times, then evaluate it
at $x=0$ and divide the result by $k!$:
$p_k=G_0^{(k)}(0)/k!$.
Moreover, we have $G_0(1)= \somme{k=0}{\infty} p_k = 1$
(this is true for all generating functions encoding distributions of probabilities),
and the average is given by
$\moy{k} = \somme{k=1}{\infty} k\ p_k = G_0'(1)$.

Going further, let us consider the generating function $G_1$ for the
number of other neighbors of a node chosen by following one random link
of a randomly chosen node. This number is distributed according
to $q_k$, defined in Equation~\ref{eqqk}. We then have
\begin{equation} \label{eqg1}
G_1(x) = \somme{k=0}{\infty} q_k x^k
= \frac{\somme{k=0}{\infty}(k+1)p_{k+1}x^k}{\moy{k}}
= \frac{\somme{k=1}{\infty} kp_k x^{k-1}}{\moy{k}}
= \frac{G'_0(x)}{\moy{k}}.
\end{equation}

This generating function will be useful in the sequel. For more
details on how  to use generating functions in the context of random
networks, see~\cite{Newman2001RandomGraphs}.

We give now a few results on generating functions which will play an
important role. These results are rewritings of results
in~\cite{Newman2000Robustness,Newman2003Handbook}. They aim at
expressing the existence of a giant component in terms of generating
functions.

Let us consider a random network 
with degree distribution $p_k$ encoded in $G_0$. Let us suppose
that some of its nodes (resp. links, \ie\ pairs of stubs) are
marked. All marked nodes are to be removed, we are therefore interested
in components composed of unmarked nodes, \ie\ sets of unmarked nodes
such that there exists a path composed only of unmarked nodes (resp.
links) between any two of them. We will call such
sets of nodes {\em clusters} and we are interested in the existence of a giant such cluster.

Let us consider a node reached by following a random link, \ie\ a node
obtained by picking a random stub. We will first compute the number of
unmarked nodes that can be reached from this node by following links
between unmarked nodes (resp. unmarked links) only.

Two cases may occur: either the chosen node (resp.
stub) is marked, in which case the cluster is of size $0$, or it is
unmarked. 
Let us denote by $r_k$ the probability that it is unmarked and has $k$ other stubs, 
and by $F_1(x)$ the corresponding generating function:
$F_1(x) = \sum_{k=0}^{\infty} r_k x^k$.
Note that the case
where the chosen node (resp. stub) is marked plays no role in $F_1(x)$.
Note also that $F_1(1)$ is the fraction of unmarked nodes (resp. links)
in the network.

When the size of the network tends towards infinity, the clusters have
a limit distribution of sizes. We will call {\em finite} clusters the
ones with a finite size in this limit distribution, while we call {\em
infinite} clusters the others. We denote by $H_1(x)$ the generating
function for the distribution of the size of such {\em finite}
clusters. Notice that $H_1(x)$ does not take into account infinite
clusters, if they exist.

\begin{lemme}
\label{lemh1prelim}
\cite{Newman2000Robustness,Newman2003Handbook}
The generating function $H_1(x)$ satisfies the following
self-consistency condition:
$$H_1(x) = 1 - F_1(1) + x F_1(H_1(x)).$$
\end{lemme}
\begin{proof}{
The cluster is of size $0$ if the chosen node (resp. stub) is marked,
which happens with probability $1-F_1(1)$ since $F_1(1)$ is the
fraction of unmarked nodes (resp. links).

In the other case, let us denote by $r_k$ the probability that the
initial node has $k$ other links, \ie\
$F_1(x) = \sum_{k=0}^{\infty} r_k x^k$.
Since we consider networks whose size tends towards infinity,
according to the mean-field framework we can neglect cycles (\ie{}
multiple paths between two nodes) in finite clusters. Then, the size
of the cluster is $1$ plus the sum of the sizes of the clusters at the
end of these $k$ links. The distribution for the sum of the sizes of
$k$ independent clusters is given by $H_1^k(x)$, see
\cite{Newman2001RandomGraphs}. Moreover, the distribution of $1$ plus
a value is obtained by multiplying the corresponding generating
function of this value by $x$. We therefore obtain $H_1(x) = 1 -
F_1(1) + x \sum_{k=0}^{\infty} r_k H_1^k(x) = 1 - F_1(1) + xF_1(H_1(x))$.
}\end{proof}

\begin{theoreme}
\cite{Newman2000Robustness,Newman2003Handbook}
\label{th-seuil-newman}
If $\tau$ is the fraction of marked nodes (resp. links)
such that removing all these marked nodes (resp. links)
gives a network with no giant component,
then $\tau$ is such that $F'_1(1)=1$.
\end{theoreme}

Before proving this result, we need a new approximation, made
implicitly
in~\cite{Newman2000Robustness,Newman2003Handbook}. It consists of
assuming that the average size of components in a random network
is finite if and only if there is no giant component. This is an approximation since one can construct graphs such that all the components
are of infinite but sub-linear size (thus there is no giant component),
in which case the average is infinite. Conversely, there may be a giant
component but a finite average size\,\footnote{Computing the distribution
of component sizes is a difficult task~\cite{molloy95critical,Newman2000Robustness}.}.
This approximation is however
necessary for the following proof of Theorem~\ref{th-seuil-newman}.
\medskip\\

\begin{proof}{
Suppose we marked enough nodes (resp. links) to ensure that there is
no giant cluster, or equivalently that there is no giant component in
the network where marked nodes have been removed. According to the
assumption above, the average size of components is finite (which does
not mean that there is no infinite component) and is given by
$H'_1(1)$. From Lemma~\ref{lemh1prelim},
$H'_1(x) = F_1(H_1(x)) + x F'_1(H_1(x)) H'_1(x)$,
and since $H_1(1) = 1$,
we obtain $$H'_1(1) = \frac{F_1(1)}{1-F'_1(1)}.$$

If there is still a giant component in the network the above
calculations do not hold since $H_1(1)$ is no longer equal to $1$.
The calculations are valid only for fractions of removed nodes
(resp. links) in the interval $]\tau,1]$ for a given $\tau$ which
is the fraction of marked (thus removed) nodes (resp. links) above
which there is no giant component anymore.

Notice now that the expression above for $H'_1(1)$ diverges at
the point $F'_1(1) = 1$, which defines $\tau$. If we choose to remove a
fraction of nodes (resp. links) closer and closer to $\tau$, but still
larger than it, the size of remaining components grows. It keeps
growing until the point where the fraction of removed nodes (resp.
links) is not large enough to destroy the giant component. At this
point, the average size of finite components tends towards infinity.
}\end{proof}

The  result we have just described is very  powerful and general. We
will see that it can be applied to many cases and give simple results
with direct proofs: to compute the fraction of nodes (resp. links) to
remove from a network in order to ensure that the resulting network
contains no giant component, it is sufficient to give an expression
for $F_1(x)$ and then to determine the fraction which leads to
$F'_1(1)=1$.

One must however keep in mind that they rely on mean-field
approximations, and that the formalism sometimes makes it
difficult to see exactly  when approximations are performed.

\medskip

We insist on the fact that, in the current state of our knowledge,
the above-mentioned approximations are necessary to derive the results we
seek. It is important however to pursue the development of exact methods
in order to validate these results and deepen our understanding. It is
important, too, to know  exactly the approximations we make and when
we make them. We will carefully point out the uses of these
approximations in the whole paper.

\subsection{Plots and thresholds.}
\label{sec_plots}

In all plots of this paper, {\em Poisson}, {\em C PL} and {\em D PL}
stand for Poisson networks, continuous power-law networks,
and discrete power-law networks, respectively.

The first main kind of plots we will consider in the sequel
represents the fraction of nodes in the largest connected component of
a network as a function of the fraction of removed nodes or links.
Figure~\ref{fig_ps_as_cc} provides an example.
To produce these plots,
we sampled a large number of networks (typically
$1\,000$) on which we repeated the experiment, and then plotted the
average behavior.
In order to be able to compare the various kinds of networks, we
selected two typical exponents for the power-law, namely $2.5$ and $3$,
produced continuous and discrete power-law networks with these exponents,
as well as Poisson networks with the same average degrees.
These values are summarized in
Table~\ref{tab_cases}.
The figures of this kind are Figures~\ref{fig_ps_cc},
\ref{fig_psa_cc}, \ref{fig_pa_cc}, \ref{fig_as_cc}, \ref{fig_asa_cc},
\ref{fig_nas_cc} and \ref{fig_naa_cc}.

\begin{table}[h!]
\centering
\begin{tabular}{|c|c|c|}
\hline
         & \multicolumn{2}{|c|}{average degree}\\
\hline
exponent & continuous power-law & discrete power-law \\
\hline
$2.5$ & $2.6$ & $1.9$\\
\hline
$3$   & $1.6$ & $1.4$\\
\hline
\end{tabular}
\caption{
The exponents we consider in our experiments on power-law networks,
and the average degrees they induce
(obtained in practice with minimal value $m=1$ and $N=100\,000$ nodes;
they are slightly lower than previsions from Lemma~\ref{lem_moments_finis_sf}
for continuous power-laws, but fit very well the previsions from
Lemma~\ref{lem_moments_finis_vsf} for discrete power-laws).}
\label{tab_cases}
\end{table}

In our context, it is usual to witness a {\em threshold} phenomenon
(typical of statistical mechanics and more precisely percolation
theory, see for instance \cite{Stauffer1994}): there exists a critical
value $p_c$ such that, whenever the fraction of removed nodes (or links,
depending on the context) is lower than $p_c$, the network almost surely
still has a giant component, whereas whenever the fraction of removed nodes (or links) is greater than $p_c$ the network almost surely does not have a giant component anymore. In other words, the threshold is reached
when the fraction of nodes in the largest connected component goes to zero (there is no giant component anymore). These thresholds play a central
role in the phenomenon under consideration and will be often studied in
the sequel.

Notice that, for a given finite size network, the notion of threshold
does not make sense: the fraction of nodes in the largest connected
component will never be zero. In this case, there are several ways to
define a threshold. One may notice that, when we reach the threshold,
the slope of the plot of the fraction of nodes in the largest connected
component in function of the number of removed nodes goes to infinity.
In finite-size computation, we may therefore consider that we reach the
threshold when this slope is maximal~\cite{Newman2000Robustness}.
Notice that this does not always make sense: it may happen, like in
Figure~\ref{fig_ps_cc} (right), that the slope is maximal at $0$ (while
the expected value of the threshold is closer to $1$). One can then adopt the convention that in such cases there is no threshold, but this reduces
our ability to discuss practical cases.

Another solution, described in~\cite{Paul2005Resilience},
consists of computing the degree distribution of the network
after each removal of a node or a link, and see if it satisfies the
criterion of Theorem~\ref{thseuil} for it to have a giant component.
The threshold is then the fraction of nodes or links to remove
so that the network does not satisfy this criterion anymore.

%

The solution we choose is to consider that the threshold is reached when
the largest connected component contains less than a given (small)
fraction of all the nodes. We chose a fraction which makes both
definitions quite equivalent in our cases, namely $0.05$. In other words,
we consider that a network does  not have a giant component whenever the
size of its largest connected component is less than $5\,\%$ of the
whole. Notice that changing this value may have a impact on numerical
results. However, similar observations would be made.

This leads us to the second main kind of plots encountered in this
paper.
For a given node or link removal strategy,
these plots represent the threshold as a function
of the main character of each kind of networks: the average degree for
Poisson networks, and the exponent of the power-law for
power-law networks.
We plot experimental results obtained
by averaging results on large number of networks (typically $1\,000$),
for different sizes of networks (typically $1\,000$,
$10\,000$ and $100\,000$).
To help in the comparison between different kinds of networks,
we add on these plots vertical lines at the values quoted in
Table~\ref{tab_cases}.
We also plot the theoretical predictions we obtain,
together with experimental results, to make it
possible to compare them.
The figures of this kind are Figures~\ref{fig_ps_seuil},
\ref{fig_psa_seuil}, \ref{fig_pa_seuil}, \ref{fig_as_seuil},
\ref{fig_asa_seuil}, \ref{fig_nas_seuil} and \ref{fig_naa_seuil}.

We will see that the experimental results do not always fit analytic
predictions very closely. This is influenced in part by the choice to
consider that a giant component must contain at least $5\,\%$ of the
nodes, as explained above.
But other factors impact this. In the case of random
failures, for instance, there is a significant difference between the
infinite limit and the finite case, even for large sizes.
This is why we
present results for both finite cases and the infinite limit when possible.
This makes it possible to observe the error
due to the asymptotic approximation. More generally, the difference
between predictions and numerical values  are due to the approximations
made in the derivations of the analytic results.

For Poisson networks, for instance, we are faced with the same problem
as the one concerning the evaluation of the maximal degree of finite
networks, see Section~\ref{secK} and Figure~\ref{fig_K_100000}:
since some parameters can only take integer values, their analytic
evaluation may lead to values quite different from their true values.
Therefore, we have chosen to use only the analytic values of the
threshold for which the error due to this effect is minimal.

Notice finally that the plot for a particular instance may vary
significantly from the average behavior, in particular for power-law
and/or small networks. We do not enter in these considerations here.


\subsection{Towards a more realistic modeling}
\label{sec_prelim_modeling}

Modeling large networks is a complex task and many parameters have to
be taken into account.
The modeling approach we use in this paper relies on {\em random sampling}:
given a set of properties of a real network that one wants to reproduce,
the goal is to choose with uniform probability a graph among the set of all
graphs having these properties.
This approach has the advantage of allowing to study precisely the impact of
a given property: if a certain behavior is observed on graphs obtained with
such a model, then we can conclude that this behavior is a consequence of
the studied property.
This type of model is also well-suited for exact proofs.
However, it also suffers from limitations:
some properties cannot currently be reproduced by this type of models.
It is for instance currently impossible to sample a graph uniformly at random
among all graphs having a given clustering coefficient,
or even a given number of triangles~\cite{Newman2003Clustering,Guillaume2004BipartiIPL,Guillaume2004BipartiCAAN}.
In summary, this type of model is well-adapted for understanding the impact of some properties 
and for formal
proofs,
but is not currently able to reproduce all properties of real-world networks.

In this paper, we focus on the degree distribution of graphs,
and use models producing random graphs with given degree distributions.
However, the sole degree distribution
cannot reproduce the complexity of networks
(for instance it is  possible to
produce completely different networks having the exact same degree
distribution).
It has also been shown that when the sole
degree distribution is taken into account, high degree nodes tend to be
connected to each other, which might not be
realistic~\cite{Doyle05robust}.

Efforts have therefore been made to study other properties
taking into account the tendency of nodes to be linked to nodes of the same degree or not,
 and incorporate
them into random models.
For instance, this can be captured by
the assortativity
parameter~\cite{Newman2002AssortativeMixing},
degree correlations
which are the probabilities $P(d|d')$ that a node of degree $d$ is linked to
another node of degree $d'$~\cite{Boguna2003EpidemiCorrelations,Vazquez2003ResilienceCorrelations},
or by the $s(g)$ parameter
which is the sum of
$d_id_j$ for all links $(i,j)$, normalized by the larger value
obtainable on a graph with the same degree sequence~\cite{Doyle05robust,Li06towards}.



Some authors have studied network resilience to failures and attacks,
in a similar way as what we present in this paper,
on random networks with degree correlations.
In \cite{Vazquez2003ResilienceCorrelations}, the authors study the
resilience of assortative networks (nodes are linked to similar nodes)
and disassortative networks to failures. The resilience depends on the
second moment of the degree distribution and the correlations:
assortative networks are very resilient even when the second
moment of the degree distribution is finite, while disassortative networks can
be fragile when the  second moment is divergent.

\medskip


Another, and orthogonal, modeling approach consists of
taking into account properties that play a role in the construction
of a network.
For instance, in the case of the internet,
one can consider the way routers work,
or other technological or economic constraints.
One then iterates an evolution process which respects a set of properties and
produces in the end a graph similar to the original network.

This type of approach has the advantage of being able to take into
account many properties that cannot be considered in random modeling,
for instance clustering~\cite{watts98collective,kleinberg99smallworld,dorogovtsev00structure,dorogovtsev02evolution,Klemm2002Clustered,Holme2002GrowingClustering}.
This makes it relevant for producing graphs for simulation purposes.
However, the evolution rules in these models create graph structures
that are hard to characterize, and in the end the properties of the
obtained graphs are not always fully understood.
A simple example of this is the preferential attachment model~\cite{barabasi99emergence},
which produces trees (if each new node creates a single link)
or graphs with no nodes of degree one (if new nodes create more than one link).
These models therefore allow to study the impact of the construction rule
on the observed behaviors,
but do not allow yet to study the impact of some topological properties of the networks.
Also, in the case of optimization models, it is not always easy
to find interesting parameters to optimize. Indeed performance related measures are quite
natural in the context of computer networks but can be harder to find
in social networks for instance.
In summary, this approach is relevant for taking into account complex network
properties, and for simulation purposes, but does not allow to study precisely
the impact of global topological properties.
Therefore, these two types of approaches are complementary.


The HOT framework,  {\em Heuristically Optimal Topology} or {\em Highly Optimized/Organized
Tolerance/Tradeoffs},  belongs to this approach. This framework has been mainly used in the context of the
internet~\cite{Doyle05robust,Li04first}.
The authors of these papers propose to rewire a network with a given
degree distribution in a non random fashion which preserves the degree
sequence. The nonrandomness consists of rewiring with the aim of
optimizing some properties of the networks to mimic real properties,
for instance technological or  economical ones. This kind of optimization can also be
found in the context of biological systems~\cite{Doyle08reverse}.


\medskip
Concerning network robustness, we only consider here the size of the largest
component, which describes the ability of nodes to communicate,
 as an indicator of the state of the network after failures or attacks.
However other
approaches have been introduced.

In~\cite{Latora2001Efficient,Holme2002Vulnerability} the authors use
the efficiency, also called average inverse geodesic length, which is
computed as $1/(N(N-1))\sum_{i,j}1/d_{ij}$, where $d_{ij}$ is the
distance between $i$ and $j$. An high efficiency means that pairs of
nodes are on average close to each other. This is very similar to the
average distance but allows to consider disconnected networks, which
is pertinent in our context.  In a similar fashion, the authors
of~\cite{park03static} introduce the Diameter-Inverse-K (DIK) measure
defined as $d/K$, where $d$ is the average distance between pairs of
connected nodes, and $K$ is the fraction of pairs of nodes which are
connected. It allows to take into account disconnected graphs and for
connected graphs allows to distinguish between short or large average
distance.  In~\cite{Crucitti2003EfficiencyScale-Free}, the authors
consider that using shortest paths can put a higher load on some nodes
which can increase when failures or attacks occur. In
consequence, each node is associated with a given capacity that cannot
be exceeded without a loss of efficiency of the node, which forces the use of longer paths.
This concept
 has mainly been used for cascading failures and attacks
but could be used as a more evolved definition of efficiency.

More specific measures have also been introduced. In particular, in
the case of the internet, nodes and links are used to carry some
demand and the efficiency of the network can be measured as its
capacity at carrying it. In~\cite{Li04first,Doyle05robust}, the
authors define the maximal throughput based on the bandwidth of nodes,
a routing matrix and the traffic demand for all nodes. In the case of
the internet, high degree nodes are at the periphery of the network
and the fragility of the network lies more on the low degree core
nodes than on the high degrees periphery nodes.

Some attack strategies based on
these measures of importance have also been introduced.
For instance, attack strategies that focus on important nodes rather than
on large degree nodes.
In~\cite{Holme2002Vulnerability} for instance, the authors
compare the removal of highest degree nodes with the removal of
highest betweenness nodes. Furthermore, the order of removal can be
either chosen before any removal occurs, or recalculated after each
removal.

Finally, more complex types of failures and attacks,
like cascading failures, have been considered in the literature,
for instance in~\cite{crucitti04error,lee05robustness,crucitti04model,newth04evolving,motter02cascade,Motter2004Cascade,Pertet2005Cascading,Zhao2004Cascading}.

\medskip

We presented here a quick overview of the different types of graph modeling
that exist in the literature, together with a brief discussion on their
advantages and drawbacks.
We also presented other attack or failures strategies, as
well as other definitions of network resilience.
Studying all this in detail is however out of the scope of this paper.



%% file: random.tex
\section{Resilience to random failures.}
\label{sec_p}

The aim of this section is to study the resilience of random networks
to random failures.
Recall that random node (resp. link) failures consist of the
removal of randomly chosen nodes (resp. links).

We will first consider random node failures
(Section~\ref{sec_ps}) on general random networks, and then
apply the obtained results to Poisson and power-law networks.
We will see that the empirical observations cited in the preliminaries
concerning the different behaviors of Poisson and power-law networks
are formally confirmed.
In order to deepen our understanding of random node failures, we will
consider in Section~\ref{sec_psa} these failures from the {\em link}
point of view: what fractions of the {\em links} are removed during
random node failures?
Finally, we will consider random {\em link} failures
(Section~\ref{sec_pa}).

\subsection{Random node failures.}
\label{sec_ps}

In this section, we first present a general result on random node failures,
independent of the type of underlying network, as long as it is a {\em
random} network.
We detail the two main proofs proposed for this result
\cite{Newman2000Robustness,Newman2003Handbook,Cohen2000RandomBreakdown,Cohen2003Handbook}.
We then apply this general result to the special cases under concern:
Poisson and power-law (both discrete and continuous) networks.
Figure~\ref{fig_ps_cc} displays the behaviors observed for these three
types of networks.

\begin{figure}[!h]
\begin{center}
\includetroisres{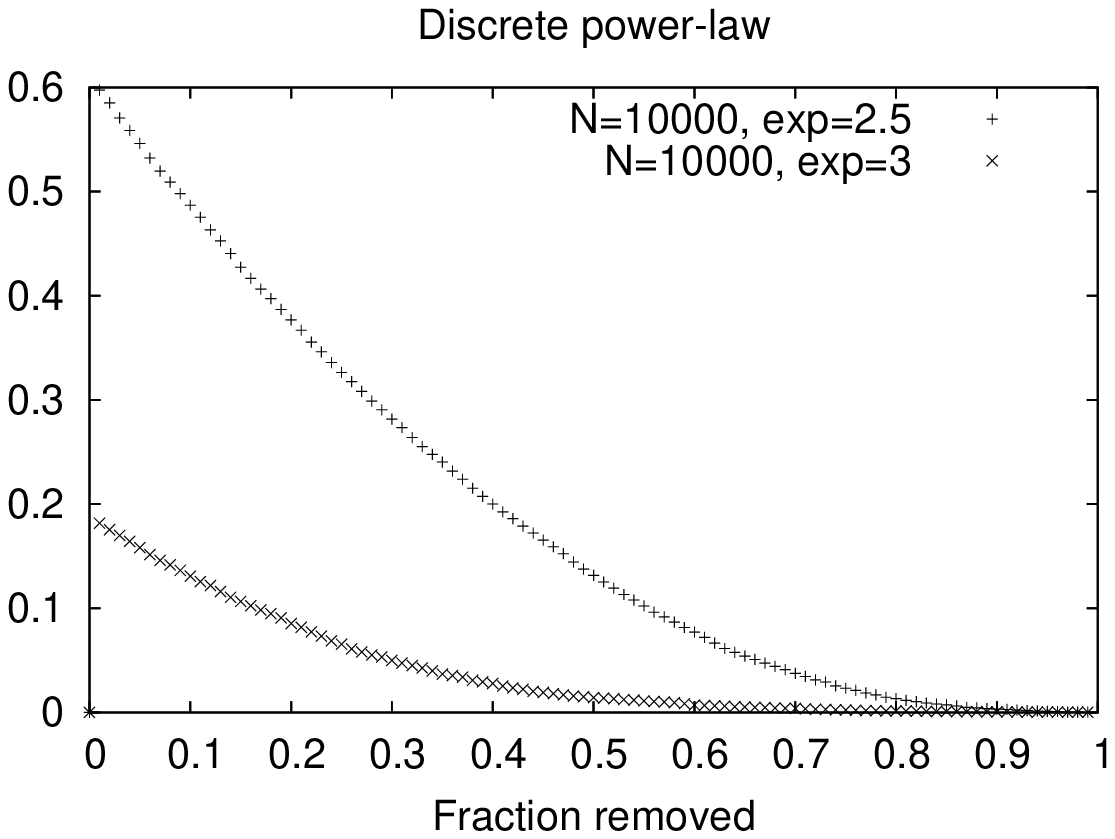}
\end{center}
\caption{
\label{fig_ps_cc}
Size of the largest connected component as a function of the
fraction of randomly removed nodes.
\lefttoright
\refplotscc
}
\end{figure}

As explained in the preliminaries, there is a fundamental difference
between Poisson and power-law networks: in the Poisson case the giant
component is destroyed when a fraction of the nodes significantly lower
than $1$ has been
removed, whereas in the power-law cases one needs to remove almost all
nodes. The aim of this section is to formally confirm this, and give
both formal and intuitive explanations of this phenomenon.

\subsubsection{General results}
\label{sec_ps_gl}

Our aim here is to prove the following general result,
which gives the value of the threshold for random node failures.

\begin{theoreme} \cite{Newman2000Robustness,Newman2003Handbook,Cohen2000RandomBreakdown,Cohen2003Handbook}
\label{th_ps_general}
The threshold $p_c$ for random node failures in large random networks
with degree distribution $p_k$
is given by
$$
p_c = 1 - \frac{\moy{k}}{\moy{k^2} - \moy{k}}.
$$
\end{theoreme}

\bigskip

\noindent
Notice that this theorem states that in some cases $p_c$ might be less
than $0$. But we have:
$$p_c = 1-\frac{\moy{k}}{\moy{k^2} - \moy{k}} \le 0 \iff
\langle k^2 \rangle - 2 \langle k \rangle \le 0.$$
According to Theorem~\ref{thseuil}, this implies  that the network
almost surely has no giant component. In this case, the notion of
threshold therefore has no meaning, and the theorem is irrelevant.

Theorem~\ref{th_ps_general} has been derived in different ways in
the literature. The two main methods were proposed
in~\cite{Cohen2000RandomBreakdown,Cohen2003Handbook} and
in~\cite{Newman2000Robustness,Newman2003Handbook}. We detail both approaches below.

\bigskip

Let us begin with the proof in~\cite{Cohen2000RandomBreakdown,Cohen2003Handbook}. It
relies on the fact that random node failures on a random network lead to a network which may still be considered as random (with a different degree
distribution), as explained in the preliminaries. Therefore, by computing the degree distribution of this network, one can use the criterion in
Theorem~\ref{thseuil} to decide if there is still a giant component or
not.


\begin{lemme} \cite{Cohen2000RandomBreakdown,Cohen2003Handbook}
\label{lem_ps_pk}
In a large random network with degree distribution $p_k$,
after the removal of a
fraction $p$ of the nodes during random node failures
the degree distribution \pk{p}
is given by
$$
\pk[k]{p} = \sum_{k_0=k}^{\infty} p_{k_0} {k_0 \choose k}(1-p)^{k} p^{k_0-k}.
$$
\end{lemme}
\begin{proof}{
If a given node had degree $k_0$ before the removal, then the probability
that it has degree $k'\le k_0$ after the removal is
${k_0\choose  k'} (1-p)^{k'} p^{k_0-k'}$.
Indeed, $k_0 - k'$ of its neighbors have been removed with probability
$p$, and $k'$ of its neighbors have not been removed with
probability $(1-p)$.
}\end{proof}

In order to apply Theorem~\ref{thseuil}, we now have to compute the first
and second moments of the new degree distribution:

\begin{proposition} \cite{Cohen2000RandomBreakdown,Cohen2003Handbook}
\label{prop_ps_moments}
With the notations of Lemma~\ref{lem_ps_pk}, the first and second
moments of the degree distribution \pk{p} are
$$
\moy{k(p)} = (1-p)\moy{k}
\hspace*{1cm}
\mbox{ and }
\hspace*{1cm}
\moy{k^2(p)} = (1-p)^2\moy{k^2} + p(1-p)\moy{k}.
$$
\end{proposition}

\noindent
In order to prove this proposition, we need the following technical lemma which gives the first and second moment of a binomial distribution.

\begin{lemme}
\label{lem_ps_choose}
For any integer $k$ and $k_0$, and any real $p$, we have
$$
\somme{k=0}{k_0} k{k_0 \choose k} (1-p)^k p^{k_0-k} = (1-p)k_0,
$$
and
$$
\somme{k=0}{k_0} k^2{k_0 \choose k}(1-p)^kp^{k_0-k}
= (1-p)^2k_0^2 + p(1-p)k_0.
$$
\end{lemme}
\begin{proof}{
Let us start with:
$$
(x+y)^{k_0} = \somme{k=0}{k_0} {k_0 \choose k} x^k y^{k_0-k}.
$$
If we differentiate this equality with respect to $x$ and then multiply the
resulting equality by $x$, we obtain:
$$
xk_0(x+y)^{k_0-1} = \somme{k=0}{k_0} {k_0 \choose k}k x^{k}y^{k_0-k}.
$$
We obtain the first claim by setting $x=1-p$ and $y=p$ in this equation.

\noindent
If again we differentiate the last equation with respect to $x$ and
multiply the resulting equality by $x$, we obtain:
$$
x(k_0(x+y)^{k_0-1} + xk_0(k_0-1)(x+y)^{k_0-2})
= \somme{k=0}{k_0} k^2 {k_0 \choose k} x^{k}y^{k_0-k}.
$$
By setting $x=1-p$ and $y=p$ we obtain:
$$
\begin{array}[t]{rcl}
\somme{k=0}{k_0} k^2 {k_0 \choose k} (1-p)^{k}p^{k_0-k}
&=& (1-p) k_0 + (1-p)^2k_0(k_0-1)\\
&=& (1-p)^2k^2_0 + p(1-p)k_0,
\end{array}
$$
which ends the proof.
}\end{proof}

\noindent
We can now prove Proposition~\ref{prop_ps_moments}:\\
\begin{proof}{
The claims follow from the following series of equations.\\
$$
\begin{array}[t]{rcl}
\moy{k(p)} &=& \somme{k=0}{\infty} k\pk[k]{p}\\
&=& \somme{k=0}{\infty}
    k \somme{k_0=k}{\infty} p_{k_0}{k_0 \choose k}(1-p)^kp^{k_0-k}\\
&=& \somme{k_0=0}{\infty}
    \somme{k=0}{k_0} k p_{k_0}{k_0 \choose k}(1-p)^kp^{k_0-k}\\
&=& \somme{k_0=0}{\infty} p_{k_0}
    \somme{k=0}{k_0} k{k_0 \choose k}(1-p)^kp^{k_0-k}\\
&=& \somme{k_0=0}{\infty} p_{k_0} (1-p) k_0\\
&=& (1-p)\moy{k}
\end{array}
$$
where we made an inversion of the sum between the second and third line
and used Lemma~\ref{lem_ps_choose} between lines four and five.
$$
\begin{array}[t]{rcl}
\moy{k^2(p)} &=& \somme{k=0}{\infty} k^2\pk[k]{p}\\
&=& \somme{k=0}{\infty}k^2
    \somme{k_0=k}{\infty} p_{k_0}{k_0 \choose k}(1-p)^kp^{k_0-k}\\
&=& \somme{k_0=0}{\infty}
    \somme{k=0}{k_0} k^2 p_{k_0}{k_0 \choose k}(1-p)^kp^{k_0-k}\\
&=& \somme{k_0=0}{\infty} p_{k_0}
    \somme{k=0}{k_0} k^2{k_0 \choose k}(1-p)^kp^{k_0-k}\\
&=& \somme{k_0=0}{\infty} p_{k_0} [(1-p)^2 k_0^2 + p(1-p)k_0]\\
&=& (1-p)^2\moy{k^2} + p(1-p)\moy{k}
\end{array}
$$
using the same tricks as for the first series of equations.
}\end{proof}

\noindent
Finally, this yields the following proof for Theorem~\ref{th_ps_general}:\\
\begin{proof}{
The threshold $p_c$ is reached when the network does not have a giant
component anymore. From Theorem~\ref{thseuil}, this happens when
$\moy{k^2(p_c)} - 2\moy{k(p_c)} = 0$.
From Proposition~\ref{prop_ps_moments}, this is equivalent to
$(1-p_c) [ (1-p_c)\moy{k^2} - (2-p_c)\moy{k}] =0$, which gives the result.
}\end{proof}

\bigskip

Let us now describe the method developed in
\cite{Newman2000Robustness,Newman2003Handbook},
to obtain Theorem~\ref{th_ps_general}. It relies on the use of generating
functions (see Section~\ref{sec_mfa}), each node being marked as {\em absent} with probability $p$ and as {\em present} with probability $1-p$.

Recall that $F_1(x)$ is the generating function for the probability
of finding an unmarked (\ie{} present) node with $k$ (marked or unmarked)
other neighbors at the end of a randomly chosen link. In our case,
$F_1(x)$ therefore is
$$
F_1(x) = \somme{k=0}{\infty} (1-p)q_k x^k = (1-p)G_1(x),
\label{F1_ps}
$$
where $G_1(x) = \somme{k=0}{\infty} q_k x^k$ is the
generating function for the probability of finding a node with $k$
others neighbors at the end of a randomly chosen link, defined in
Section~\ref{sec_mfa}.
\noindent
We can then prove Theorem~\ref{th_ps_general} as a direct consequence
of Theorem~\ref{th-seuil-newman}:\\
\begin{proof}{
\label{preuve2_ps}
From Theorem~\ref{th-seuil-newman}, the threshold $p_c$ is reached
when $F'_1(1)=1$, which is equivalent here to $(1-p_c)G'_1(1)=1$.
Therefore $p_c$ satisfies
$$
p_c = 1-\frac{1}{G'_1(1)}.
$$
We know that $G_1(x)=\somme{k=0}{\infty} q_k x^k
= \somme{k=1}{\infty}kp_kx^{k-1}/\moy{k}$.
Therefore $G'_1(x)=\somme{k=2}{\infty} k(k-1)p_kx^{k-2} / \moy{k}
=\somme{k=0}{\infty} k(k-1)p_kx^{k-2} / \moy{k}
$,
and $G'_1(1)= \frac{\moy{k^2}-\moy{k}}{\moy{k}}$. This ends the proof.
}\end{proof}

\bigskip

The two proofs have different advantages and drawbacks. The
first one is self contained and relies only on classical probabilistic
notions, but it is quite long and technical. The second one is very
concise and simple, but it relies on the generating function formalism,
which has to be first introduced and understood. These differences do
not only have an impact on the aspect of the proofs: they also imply
that one has to think carefully about each approximation in the first
approach, while they are hidden in the generating function formalism
in the second one. As a counterpart, the first approach makes it easier
to tune and locate approximations precisely.

\subsubsection{The cases of Poisson and power-law networks}

Theorem~\ref{th_ps_general} is valid for any random network, whatever its degree
distribution. To study the behavior of Poisson and power-law networks in
case of random node failures, we therefore only have to apply it to these
cases. More precisely, we will consider Poisson, continuous power-law
and discrete power-law networks, and, for each of these classes, both
finite networks with $N$ nodes and finite networks with size tending
towards infinity. Comparison with simulations will be provided at the
end of the subsection.


Note that we will derive all the results for finite size networks as
corollaries of results presented in previous sections. The results for
networks with size tending towards infinity can then be derived either
from results of the previous sections, or as
limits of the corresponding finite cases.

\begin{corollaire}
\label{cor_ps_ER_finis}
\forPoissonFinis{},
the threshold $p_c$ for random node failures is given by
$$
p_c = 1 - \frac{\somme{k=0}{K} z^k/(k-1)!}
{\somme{k=0}{K}z^k/(k-2)!},
$$
where $K$ is the maximal degree or the network, related to $p_k$ and $N$
by Lemma~\ref{lem_K_er}.
\end{corollaire}
\begin{proof}{
Direct application  of Theorem~\ref{th_ps_general}
using Lemma~\ref{lem_moments_finis_er}.
}\end{proof}

\begin{corollaire}
\label{cor_ps_ER}
\cite{Cohen2000RandomBreakdown}
\forPoisson{},
the threshold $p_c$ for random node failures is
$$p_c = 1 - \frac{1}{z}.$$
\end{corollaire}
\begin{proof}{
Direct application  of Theorem~\ref{th_ps_general}
using Lemma~\ref{lem_moments_er}.
}\end{proof}

\begin{corollaire}\cite{Cohen2000RandomBreakdown}
\label{cor_ps_sf_finis}
\forSfFinis{},
the threshold $p_c$ for random node failures is
$$
p_c =
\left\{
\begin{array}{ll}
1 - \left[
\frac{2-\alpha}{3-\alpha}m -1 \right]^{-1}&\mbox{ if }\alpha > 3\\
1 - \left[\frac{2-\alpha}{\alpha-3} m N^{\frac{3-\alpha}{\alpha-1}}
     - 1\right]^{-1}
      & \mbox{ if }2<\alpha<3\\
1 - \left[ \frac{2-\alpha}{3-\alpha} mN^{\frac{1}{\alpha-1}} -1 \right]^{-1} & \mbox{ if }1<\alpha<2.
\end{array}
\right.
$$
\end{corollaire}
\begin{proof}{
We can rewrite Theorem~\ref{th_ps_general} into
$p_c=1 - 1/(\moy{k^2}/\moy{k}-1)$.
From the approximations of $\moy{k}$ and \moy{k^2} in Lemma~\ref{lem_moments_finis_sf},
we then obtain:
$$
p_c =
\left\{
\begin{array}{lr}
1 - \left[
\frac{2-\alpha}{3-\alpha}m -1 \right]^{-1}&\alpha > 3\\
1 - \left[\frac{2-\alpha}{\alpha-3} m^{\alpha-2}K^{3-\alpha} - 1\right]^{-1}
      & 2<\alpha<3\\
1 - \left[ \frac{2-\alpha}{3-\alpha} K -1 \right]^{-1} & 1<\alpha<2.
\end{array}
\right.
$$
Using the evaluation of $K$ in Lemma~\ref{lem_K_sf}, we obtain the result.
}\end{proof}

\begin{corollaire}\cite{Cohen2000RandomBreakdown}
\label{cor_ps_sf}
\forSf{},
 the threshold $p_c$ for random node failures is
$$
p_c =
\left\{
\begin{array}{ll}
1 - \left[
\frac{2-\alpha}{3-\alpha}m -1 \right]^{-1}&\mbox{ if }\alpha > 3\\
1 & \mbox{ if }1<\alpha<3.
\end{array}
\right.
$$
\end{corollaire}
\begin{proof}{
Direct application of Corollary~\ref{cor_ps_sf_finis} when the size
tends towards infinity.
}\end{proof}

\begin{corollaire}
\label{cor_ps_vsf_finis}
\forVsfFinis{},
the threshold $p_c$ for random node failures is given by
$$p_c = 1-\frac{\harmo{K}{\alpha-1}}
{\harmo{K}{\alpha-2} - \harmo{K}{\alpha-1}},$$
\defharmo{},
where $K$ is the maximal degree or the network, related to $p_k$ and $N$
by Lemma~\ref{lem_K_vsf}.
\end{corollaire}
\begin{proof}{
Direct application of Theorem~\ref{th_ps_general}
using Lemma~\ref{lem_moments_finis_vsf}.
}\end{proof}

\begin{corollaire} \cite{Newman2000Robustness}
\label{cor_ps_vsf}
\forVsf{}, the threshold $p_c$ for random node failures is
$$p_c = 1-\frac{\zeta(\alpha-1)}{\zeta(\alpha-2) - \zeta(\alpha-1)}.$$
\end{corollaire}
\begin{proof}{
Direct application of Corollary~\ref{cor_ps_vsf_finis} when the size
tends towards infinity.
}\end{proof}

\bigskip

We plot numerical evaluations of these results in
Figure~\ref{fig_ps_seuil}, together with experimental results.
We also give in Table~\ref{tab_ps} the thresholds for specific
values of the exponent and the average degree.

\begin{figure}[!h]
\begin{center}
\includetroisres{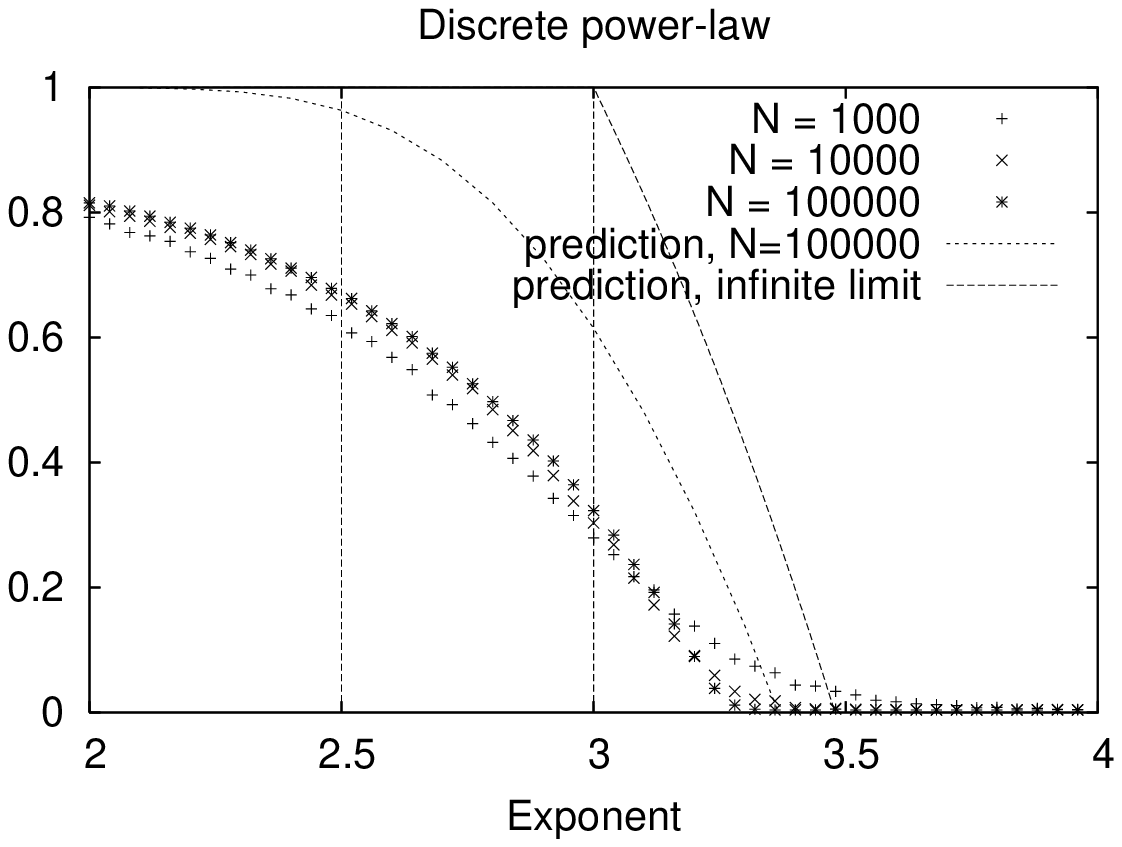}
\end{center}
\caption{
\label{fig_ps_seuil}
Thresholds for random node failures.
\lefttoright
\refplotsseuil
}
\end{figure}

\degc{ps}{random node failures}{1 & 0.74 & 0.62 & 0.59 & 1 & 0.67 & 0.47 & 0.45}{1 & 0.55 & 0.38 & 0.34 & 1 & 0.32 & 0.29 & 0.26}

The central point here is to notice that power-law and Poisson networks
display a qualitatively different behavior in case of node failures.
In theory, power-law networks have a threshold $p_c=1$ as long as the
exponent is lower than $3$ (most real-world cases), which means that all
nodes have to be removed to achieve a breakdown. On the contrary, for
Poisson networks only a finite (\ie\ strictly lower than $1$) fraction
of the nodes has to be removed.
This leads to the conclusion that power-law networks are significantly
more resilient to node  failures than Poisson networks, which confirms
the experimental observations discussed in introduction.

However, this result is moderated by the two following observations.
First, Poisson networks may have a quite large threshold when their
average degree grows (which appears from both analytic previsions and
experiments).
Second, and more importantly, power-law networks of finite size $N$
are much more sensitive to failures than what is predicted for the
infinite limit.
This is already true from the analytic previsions, and even more
pronounced for experiments.

This is particularly clear when one compares the behavior of networks
of various kinds but with the same average degree, see
Table~\ref{tab_ps}.
The experimental thresholds of Poisson networks are at most
$38\%$ smaller than those for power-law networks of the same average
degrees.


We may therefore conclude that power-law networks
are indeed more resilient to
random node failures than Poisson ones, but that the difference in
practice is not as striking as predicted by the infinite limit
approximations.

\subsection{Link point of view of random node failures.}
\label{sec_psa}

As discussed in the preliminaries, one may wonder what happens in networks
during random {\em node} failures in terms of {\em  the number of links
removed}.
The plots of the size of the largest component as a fraction of the number
of links removed during random node failures are given in Figure~\ref{fig_psa_cc}
(notice that
these plots are nothing but (nonlinear) rescalings of the plots in
Figure~\ref{fig_ps_cc}).
The question we address here therefore is: how many links have
been removed when we reach the threshold for random node failures?
This is not equivalent to random  removals
of links, which are studied in the next subsection.

\begin{figure}[!h]
\begin{center}
\includetroisres{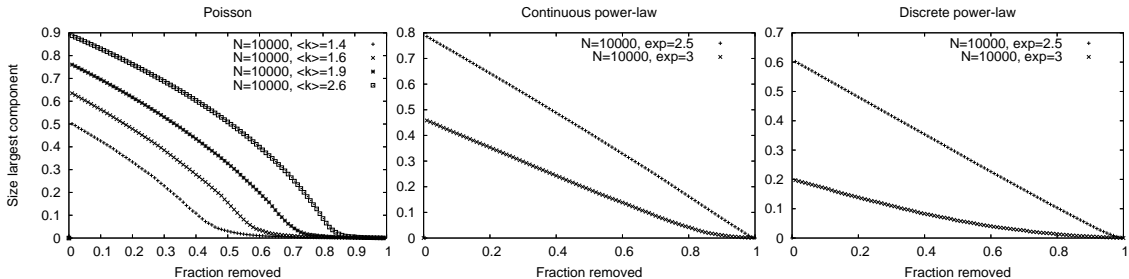}
\end{center}
\caption{
\label{fig_psa_cc}
Size of the largest connected component as a function of the
fraction of removed {\em links}, during random {\em node} failures.
\lefttoright
\refplotscc
}
\end{figure}

Like in the previous subsection, we will begin with general results
and then apply them to the cases  under concern.

\subsubsection{General results}

One can evaluate the number of links removed during random node
failures as follows.

\begin{proposition}
\label{prop_psa_mp}
In large random networks, after the removal of a fraction $p$ of the
nodes during random nodes failures, the fraction of removed links is
$m(p) = 2p-p^2$.
\end{proposition}
\begin{proof}{
Let us consider a network in which we randomly remove a fraction
$p$ of the nodes. Since the nodes are chosen randomly, we can assume
that the same fraction $p$ of the stubs in the network were attached
to the removed nodes. Each stub is kept with probability $(1-p)$, the
fraction of pairs of stubs linking non removed nodes is therefore
$(1-p)^2$. This last quantity is the fraction of non removed links
and $1-(1-p)^2 = 2p - p^2$ is finally the fraction of removed links.
%
%
%
}\end{proof}

We can now use this result to study the threshold for random node
failures in terms of the fraction of removed links.

\begin{corollaire}
\label{cor_psa_pc}
The fraction of links removed at the threshold $p_c$ for random node
failures in large random networks with degree distribution $p_k$ is
$$
m(p_c) = 2p_c - p_c^2 =
1-\left(\frac{\moy{k}}{\moy{k^2}-\moy{k}}\right)^2.
$$
\end{corollaire}
\begin{proof}{
Immediate from Theorem~\ref{th_ps_general} and Proposition~\ref{prop_psa_mp}.
}\end{proof}

\subsubsection{The cases of Poisson and power-law networks}

We now apply the general result above to the cases of interest, which
gives a corollary in each case.

\begin{corollaire}
\label{cor_psa_ER_finis}
\forPoissonFinis{},
the fraction of links removed at the threshold $p_c$ for random node
failures is
$$
m(p_c) = 1 - \left(\frac{\somme{k=0}{K} z^k/(k-1)!}
{\somme{k=0}{K}z^k/(k-2)!}\right)^2,
$$
where $K$ is the maximal degree or the network, related to $N$
by Lemma~\ref{lem_K_er}.
\end{corollaire}
\begin{proof}{
Direct application of Corollaries~\ref{cor_ps_ER_finis} and~\ref{cor_psa_pc}.
}\end{proof}

\begin{corollaire}
\label{cor_psa_ER}
\forPoisson{},
the fraction of links removed at the threshold $p_c$ for
random node failures is
$$
m(p_c) = 1 - \frac{1}{z^2}.
$$
\end{corollaire}
\begin{proof}{
Direct application of Corollaries~\ref{cor_ps_ER} and~\ref{cor_psa_pc}.
}\end{proof}

\begin{corollaire}
\label{cor_psa_sf_finis}
\forSfFinis{},
the fraction of links removed at the threshold $p_c$ for random node
failures is
$$
m(p_c) =
\left\{
\begin{array}{lr}
 1 - \left[
\frac{2-\alpha}{3-\alpha}m -1 \right]^{-2} &\alpha > 3\\
 1 -
\left[\frac{2-\alpha}{\alpha-3} m N^{\frac{3-\alpha}{\alpha-1}}
     - 1\right]^{-2}
      & 2<\alpha<3\\
 1 -
\left[ \frac{2-\alpha}{3-\alpha} mN^{\frac{1}{\alpha-1}} -1 \right]^{-2}
    & 1<\alpha<2.
\end{array}
\right.
$$
\end{corollaire}
\begin{proof}{
Direct application of Corollaries~\ref{cor_ps_sf_finis} and~\ref{cor_psa_pc}.
}\end{proof}

\begin{corollaire}
\label{cor_psa_sf}
\forSf{},
the fraction of links removed at the threshold $p_c$ for random node
failures is
$$
m(p_c) =
\left\{
\begin{array}{lr}
 1 - \left[
\frac{2-\alpha}{3-\alpha}m -1 \right]^{-2} &\alpha > 3\\
 1       & 1<\alpha<3.
\end{array}
\right.
$$
\end{corollaire}
\begin{proof}{
Direct application of Corollaries~\ref{cor_ps_sf}
and~\ref{cor_psa_pc}.
}\end{proof}

\begin{corollaire}
\label{cor_psa_vsf_finis}
\forVsfFinis{},
the fraction of links removed at the threshold $p_c$ for random node
failures is given by
$$m(p_c) = 1-
\left(\frac{\harmo{K}{\alpha-1}}
{\harmo{K}{\alpha-2} - \harmo{K}{\alpha-1}}
\right)^2,
$$
\defharmo{},
and $K$ is the maximal degree or the network, related to $N$
by Lemma~\ref{lem_K_vsf}.
\end{corollaire}
\begin{proof}{
Direct application of Corollaries~\ref{cor_ps_vsf_finis} and~\ref{cor_psa_pc}.
}\end{proof}

\begin{corollaire}
\label{cor_psa_vsf}
\forVsf{},
the fraction of links removed at
the threshold $p_c$ for random node failures is
$$
m(p_c) = 1-\left(\frac{\zeta(\alpha-1)}{\zeta(\alpha-2) - \zeta(\alpha-1)}\right)^2.
$$
\end{corollaire}
\begin{proof}{
Direct application of Corollaries~\ref{cor_ps_vsf} and~\ref{cor_psa_pc}.
}\end{proof}

We plot numerical evaluations of these results in
Figure~\ref{fig_psa_seuil}, together with experimental results.
We also give in Table~\ref{tab_psa} the thresholds for specific
values of the exponent and the average degree.

\begin{figure}[!h]
\begin{center}
\includetroisres{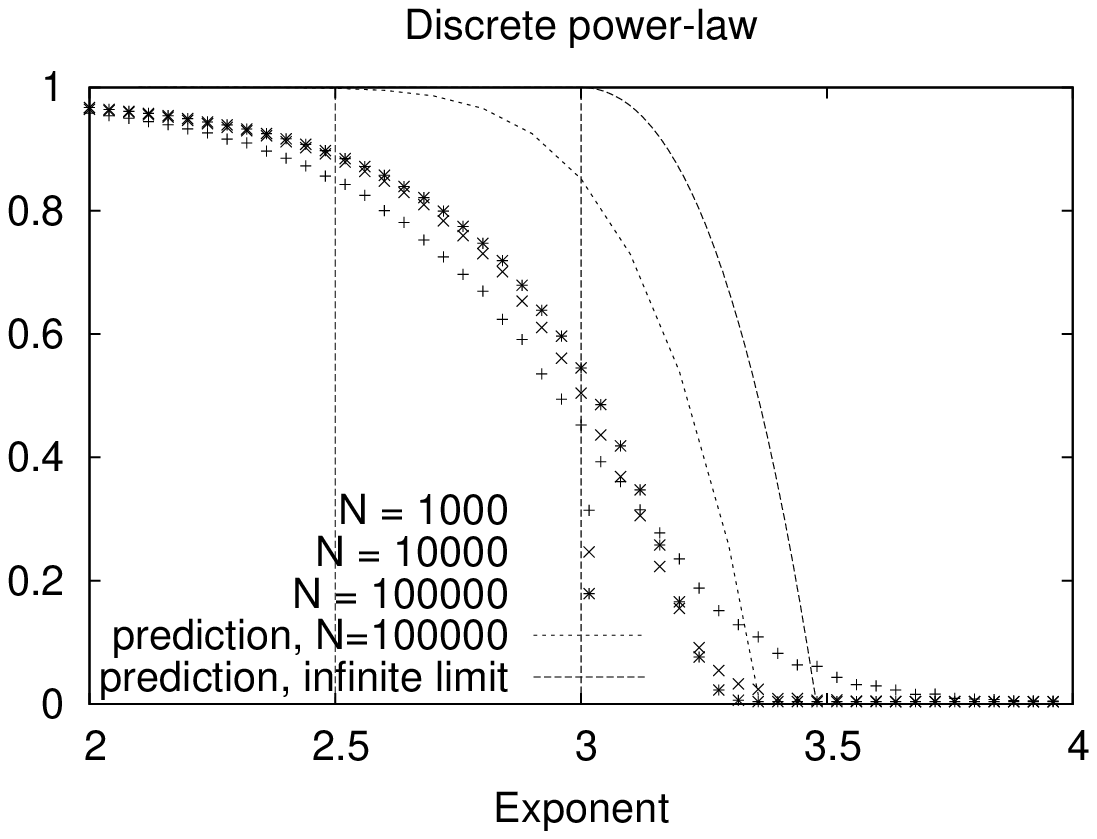}
\end{center}
\caption{
\label{fig_psa_seuil}
Thresholds for the link point of view of random node failures.
\lefttoright
\refplotsseuil
}
\end{figure}

\degc{psa}{the link point of view of random node failures}{1 & 0.93 & 0.85 & 0.83 & 1 & 0.89 & 0.72 & 0.69}{1 & 0.80 & 0.61 & 0.57 & 1 & 0.54 & 0.49 & 0.44}

As expected, these results are not qualitatively different from what
is observed from the node point of view. Again, power-law networks
are more resilient than Poisson ones, but the difference  in practice
is not as important as in the predictions.

Notice also that
the fraction of removed links is significantly larger at the threshold
than the fraction of removed nodes. This is a simple consequence of the
fact that removing a node leads to the removal of both its stubs and
some of its neighbors, as explained in the proof of
Proposition~\ref{prop_psa_mp}.

\subsection{Random link failures.}
\label{sec_pa}

Until now we observed the behavior of random network when {\em nodes}
are randomly removed, both from the nodes  and from the link points
of view. One may then wonder what happens when we remove {\em links}
rather than nodes, still at random. This may model link failures, just
like random removal of nodes models node failures.

Typical behaviors for each type of random networks under concern, when
one randomly removes links, are plotted in
Figure~\ref{fig_pa_cc}. Just like in the case of random node failures
(see Figure~\ref{fig_ps_cc}), there is a qualitative difference
between Poisson and power-law networks. Going further, the plots are
very similar to the ones for node failures. We will see that the
formal results for both cases are indeed identical.

\begin{figure}[!h]
\begin{center}
\includetroisres{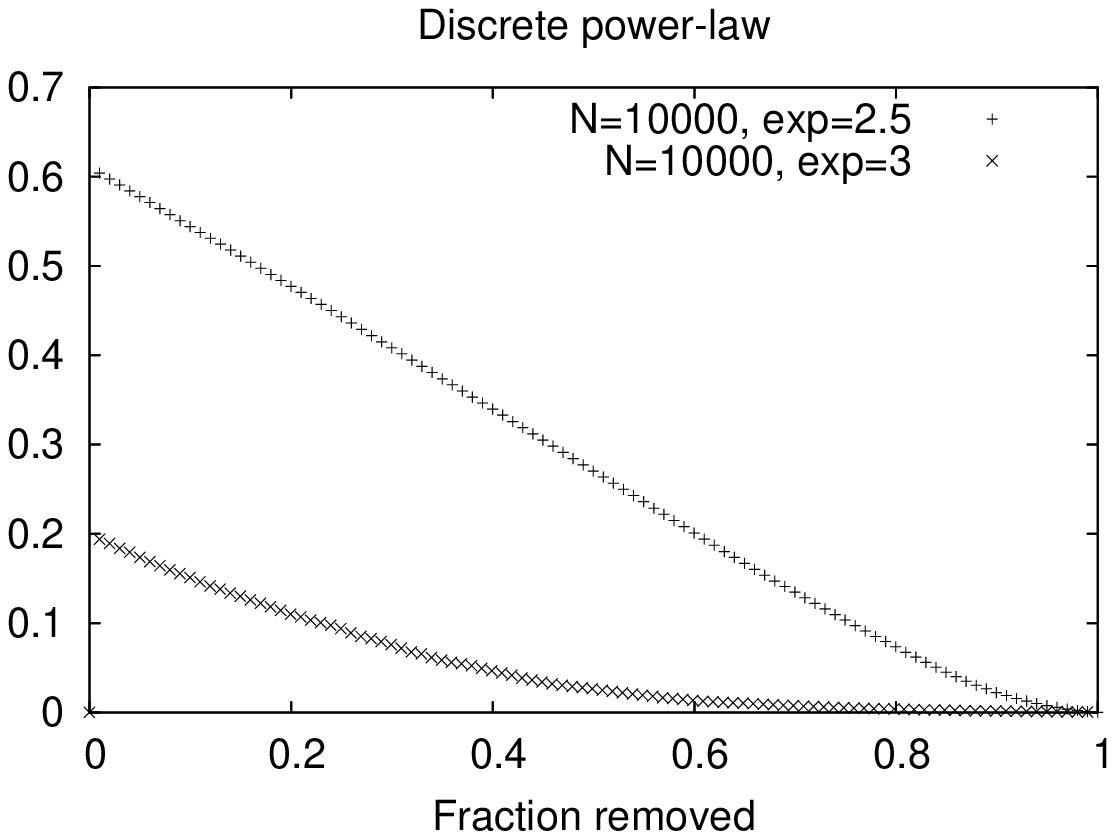}
\end{center}
\caption{
\label{fig_pa_cc}
Size of the largest connected component as a function of the
fraction of randomly removed links.
\lefttoright
\refplotscc
}
\end{figure}

Again, in this section we will first prove a  general result which
we apply to the three cases under concern. We then compare formal
results to experiments, and discuss them.

\subsubsection{General results}

The goal of this section is to prove that the the threshold $m_c$ for
random link failures actually is the same as the one for
random {\em node} failures (see Theorem~\ref{th_ps_general}):

\begin{theoreme} \cite{Newman2000Robustness,Cohen2001Attack}
\label{th_pa_general}
The threshold $m_c$ for random link failures in large random networks
with degree distribution $p_k$
is
$$
m_c = 1 - \frac{\moy{k}}{\moy{k^2} - \moy{k}}.
$$
\end{theoreme}

Just as we did for Theorem~\ref{th_ps_general}, we will present the
main ways to derive this result, as described in~\cite{Newman2000Robustness,Newman2003Handbook} and in~\cite{Cohen2001Attack,Cohen2003Handbook}.
They are very similar to the ones for
Theorem~\ref{th_ps_general} therefore we will present them in less  detail.

\bigskip

The first proof is based on the fact that, as explained in the preliminaries,
a random network in which one randomly removes links may still be viewed as
a random network, with another degree distribution.
Therefore, one can use the criterion given in Theorem~\ref{thseuil} with
this new degree distribution to decide if the network still has a giant
component. We therefore begin with the computation of the new degree
distribution.

\begin{lemme} \cite{Cohen2001Attack}
\label{lem_pa_pkm}
In a large random network with degree distribution $p_k$, after the
removal of a fraction $m$ of the links during random link failures the
degree distribution \pk{m} is given by
$$
\pk[k]{m} = \sum_{k_0=k}^{\infty} p_{k_0} {k_0 \choose k}(1-m)^{k} m^{k_0-k}.
$$
\end{lemme}
\begin{proof}{
Removing randomly a fraction $m$ of the links corresponds to removing
randomly a fraction $m$ of the stubs. If a given node has degree $k_0$
before the removals, then the probability that its degree becomes
$k'\le k_0$ is ${k_0\choose  k'} (1-m)^{k'} m^{k_0-k'}$.
Indeed, $k_0 - k'$ of its stubs have been removed, with probability
$m^{k_0-k'}$, and $k'$ of its stubs are still present, with probability
$(1-m)^{k'}$. The result follows.
}\end{proof}

\noindent
This leads to the first proof of Theorem~\ref{th_pa_general}:\\
\begin{proof}{
Notice that Lemma~\ref{lem_pa_pkm} actually is nothing but a direct
rewriting of Lemma~\ref{lem_ps_pk} on random node failures.
Therefore Theorem~\ref{th_pa_general} is derived from Lemma~\ref{lem_pa_pkm}
and Theorem~\ref{thseuil} in exactly the same way as Theorem~\ref{th_ps_general}
is derived from Lemma~\ref{lem_ps_pk} and Theorem~\ref{thseuil}.
}\end{proof}

\bigskip

The other method used to obtain this result~\cite{Newman2000Robustness}
 relies on generating functions.
As explained in the preliminaries, each link is marked as {\em present} with
probability $1-m$, and {\em absent} with probability $m$.

Recall that $F_1(x)$ is the generating function for the probability that,
when following a random link, this link is unmarked (\ie{} present) and leads to
a node with $k$ other (marked or unmarked) links emanating from it.
In our case, $F_1(x)$ therefore is
$$F_1(x) = \somme{k=0}{\infty} (1-m)q_kx^k = (1-m)G_1(x).$$

\noindent
This leads us to the second proof of Theorem~\ref{th_pa_general}:\\
\begin{proof}{
Again, the generating function $F_1(x)$ obtained here is exactly the
same as the one obtained in Section~\ref{F1_ps}, page~\pageref{F1_ps},
for random node failures.
The proof therefore is the same as the proof for Theorem~\ref{th_ps_general},
page~\pageref{preuve2_ps}.
}\end{proof}

\subsubsection{The cases of Poisson and power-law networks}

Since theoretical results for random link failures are the same as
those for random node failures, we focus in this section on
experimental results. See Figure~\ref{fig_pa_seuil} and
Table~\ref{tab_pa}.

\begin{figure}[!h]
\begin{center}
\includetroisres{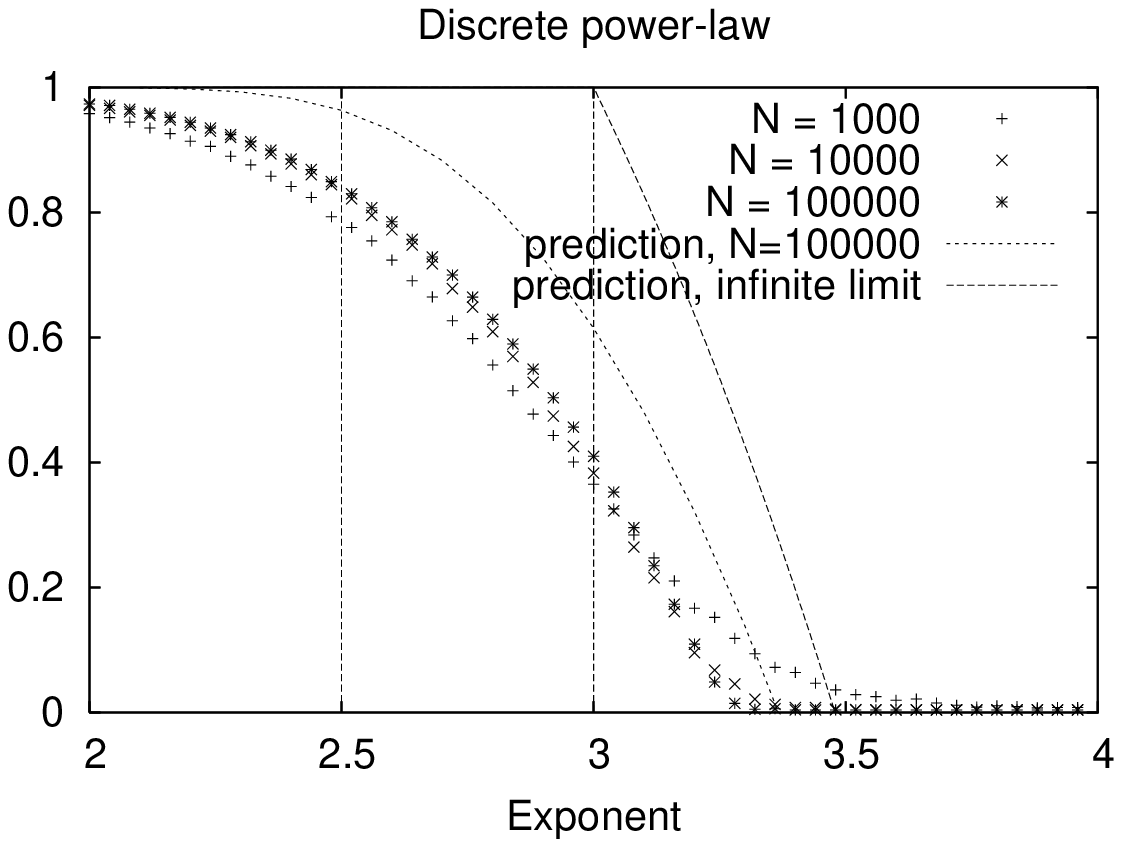}
\caption{
\label{fig_pa_seuil}
Thresholds for random link failures.
\lefttoright
\refplotsseuil
}
\end{center}
\end{figure}

\degc{pa}{random link failures}{1 & 0.90 & 0.62 & 0.60 & 1 & 0.84 & 0.47 & 0.46}{1 & 0.67 & 0.38 & 0.35 & 1 & 0.41 & 0.29 & 0.27}

In principle, these plots and values should be exactly the same as
the ones in Figure~\ref{fig_ps_seuil} and in Table~\ref{tab_ps}.
This is true for the analytic previsions, but experiments differ
significantly, which deserves more discussion.

When we consider the
size of the largest connected component as a function of the fraction
of removed nodes/links, see Figures~\ref{fig_ps_cc} and~\ref{fig_pa_cc},
then it appears clearly that, though the plot seems to reach $0$ at the
same fraction, they do not have the same shape.
Since we chose to define the threshold as the value for which the
largest connected component reaches $5\,\%$ of the total number of nodes,
the different shapes give
different experimental thresholds.

As explained in Section~\ref{sec_plots}, the $5\,\%$ value is somewhat
arbitrary, but we insist on the fact that, in the case of power-law
networks considered in these experiments,  the other main method for
computing the threshold cannot be applied: in  several cases,
the slope of the plots for
these networks is always decreasing (see Section~\ref{sec_plots}).

Finally, the same conclusions as the ones for random node failures
hold:
power-law networks are more resilient to random link failures than
Poisson ones, but the difference in practice is not as striking
as predicted by the results for the infinite limit.

%

\subsection{Conclusion on random failures.}

Two main formal conclusions have been reached  in this section
concerning the case where the size of the network tends towards
infinity.
First, as expected from the empirical results discussed in the
introduction, Poisson and power-law networks behave qualitatively
differently in case of (node or link) random failures: whereas
Poisson networks display a clear threshold, in power-law ones all the
nodes or links have to be removed to achieve a breakdown.
Second, what happens in case of random link failures is very
similar, if not identical, to what happens in case of random node
failures. On the other hand, link point of view does not change
the observations qualitatively but the fraction of removed links
at the threshold is significantly larger than the fraction of
removed nodes.
This also means that the thresholds for the link point of view of
random node failures are larger than the thresholds for random link
failures.

The qualitative difference between Poisson and power-law networks
leads to the conclusion that power-law networks are much more resilient
to random failures. This may be used in the design of large scale
networks, and it may also be seen as an explanation of the fact that
real-world networks like the internet or biological networks seem
very resilient to random errors. This has been widely argued in the
literature, see for instance~\cite{dorogovtsev2000NetworksBook,barabasi02linked}.

These results however concern only the limit case where the size of
the network tends towards infinity. When one considers networks of a
given size $N$, even for very large values of $N$, then the difference
between Poisson and power-law
networks often is much less striking than predicted. This is even clearer when one
considers the link point of view.


%% file: attack.tex
\section{Resilience to attacks.}
\label{sec_a}
The aim of this section is to study the resilience of random networks
to targeted attacks.
In our context, an attack on a network consists of a series of node or
link removals, like failures. The difference lies in the
fact that the removals are {\em not} random anymore; instead, the
nodes or links to remove are chosen according to a {\em strategy}.

Obviously, one may define many different strategies, and failures
themselves could be considered as attacks where the strategy consists
of choosing randomly.
More subtle strategies can however be
much more efficient to destroy a network.
We already presented such a  strategy,  defined in the initial
paper on the topic \cite{albert00error}, which received since then much attention. It consists
of the removal of nodes in decreasing
order of their degree;
we will call this strategy a {\em classical attack}.

We will first consider these classical attacks
(Section~\ref{sec_as}) on general random networks, and then
apply the obtained results to Poisson and power-law networks.
In order to deepen our understanding of classical attacks, we will
consider these attacks from the {\em link}
point of view (Section~\ref{sec_asa}): what fractions of the {\em links} are removed during
classical attacks?
We will also introduce new attack strategies (both on nodes and on
links) to give deeper understanding on classical attacks (Section~\ref{sec_na}).
Finally we will conclude this section with a detailed discussion on
the efficiency of classical attacks, as well as
other attack strategies
(Section~\ref{sec_a_conclu}).

\subsection{Classical attacks.}
\label{sec_as}

In this section, we first present a general result on classical attacks,
independent of the type of underlying network, as long as it is a {\em
random} network. We detail the two main proofs proposed for this result
in \cite{Newman2000Robustness,Newman2003Handbook} and~\cite{Cohen2001Attack,Cohen2003Handbook}.
We then apply this general result to the special cases under concern:
Poisson and power-law (both discrete and continuous versions) networks.
Figure~\ref{fig_as_cc} displays the behaviors observed for these three
types of networks.

\begin{figure}[!h]
\begin{center}
\includetroisres{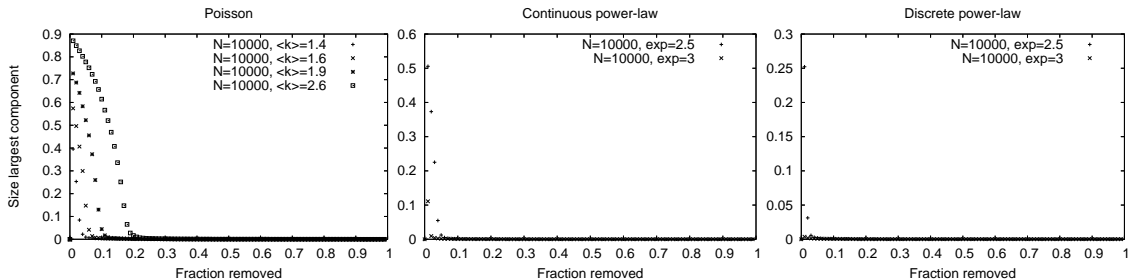}
\end{center}
\caption{
Size of the largest connected component as a function of the
fraction of nodes removed during classical attacks.
\lefttoright
\refplotscc
}
\label{fig_as_cc}
\end{figure}

As explained in the preliminaries, there is no fundamental difference
in the behaviors of Poisson and power-law networks in case
of classical attacks:
in both cases the largest connected component is quickly destroyed.
It is important however
to notice that Poisson networks are significantly more resilient than
power-law ones. The aim of this section is to formally confirm these
observations, and give both formal and intuitive explanations.

\subsubsection{General results}

Our aim here is to prove the following general result,
which gives the value of the threshold for classical attacks.

\begin{theoreme} \cite{Newman2000Robustness,Cohen2001Attack}
\label{th_as_general}
The threshold $p_c$ for classical attacks in random networks,
with size tending towards infinity and degree distribution $p_k$,
is given by
$$
\frac{\sum_{k=0}^{K(p_c)} k(k-1)p_k}{\moy{k}} = 1,
$$
where $K(p_c)$ is the maximal degree in the network after the attack,
related to $p_c$ by Lemma~\ref{lem_as_Kp}.
\end{theoreme}

As in the case of failures, there are two main ways to derive
this result, proposed in \cite{Newman2000Robustness,Newman2003Handbook} and~\cite{Cohen2001Attack,Cohen2003Handbook}.
They both rely on the following result which concerns the maximal
degree of random networks after removal of a fraction $p$ of the
nodes during a classical attack, which we denote by $K(p)$.

\begin{lemme} \cite{Cohen2001Attack}
\label{lem_as_Kp}
In a random network with size tending towards infinity
and degree distribution $p_k$,
after removal of a fraction $p$ of the nodes during
a classical attack the maximal degree $K(p)$ is given by
$$p = 1 - \somme{k=0}{K(p)} p_k.$$
\end{lemme}
\begin{proof}{
Before the removals, the network has a maximal degree $K$.
The new maximal degree $K(p)$ can then be evaluated using
$\somme{k=K(p)+1}{K} p_k = p$.
From Lemma~\ref{lem_K_general}, this is equivalent to
$\somme{k=K(p)}{\infty} p_k = p + \frac{1}{N}$
(neglecting the difference between $\somme{k=K}{\infty} p_k$
and $\somme{k=K+1}{\infty} p_k$).
Since $N$ tends to infinity, we can neglect $1/N$, which gives
$p = \somme{k=K(p)+1}{\infty} p_k$, hence the result.
}\end{proof}

In order to compute the threshold
for random networks with a given degree distribution
and size tending towards infinity,
one therefore has to proceed in two steps: first
compute the value of $K(p_c)$ using Theorem~\ref{th_as_general},
then obtain the value of $p_c$ using Lemma~\ref{lem_as_Kp}. Note
that we will mainly use Lemma~\ref{lem_as_Kp} to compute a fraction of removed nodes given a maximal degree and not the converse.

Before entering in the core of the proofs, notice that the above results
hold for random networks {\em with size tending towards infinity}.
It would be possible to write equivalent results for large
networks of finite size $N$,
by taking into account the (original) maximal degree
given by Lemma~\ref{lem_K_general}.
However, we do not consider this case  in this section
for two main reasons.
First, and most important, the results for large networks of finite
size $N$ would be very similar to the ones for the infinite limit.
Indeed, the
equivalent of Lemma~\ref{lem_as_Kp} for finite networks is:
$$
p = 1 - \somme{k=0}{K(p)} p_k - \frac{1}{N}
$$
(where we neglected the difference between
$\somme{K+1}{\infty}p_k$ and $\somme{K}{\infty} p_k$).
For large networks, $1/N$ is very small compared to $p$ when $p$ is
the threshold for classical attacks (which we will prove later
in this section). Therefore the maximal degree of a large network of
finite size $N$ after a classical attack is very close to the one
at the infinite limit.
Second, considering large networks of finite size $N$ would make the
following computations much more intricate.
One must however keep in
mind that the case of finite networks is tractable,
and results can be derived in a very similar way.

\bigskip
We now give the two main proofs available for
Theorem~\ref{th_as_general}. The first one, from \cite{Cohen2001Attack,Cohen2003Handbook},
uses the following preliminary result.

\begin{lemme} \cite{Cohen2001Attack}
\label{lem_as_sp}
In a random network with size tending towards infinity and
degree distribution $p_k$,
when a fraction $p$ of the nodes is removed  during a classical attack,
the fraction $s(p)$ of stubs attached
to removed  nodes  is given by
$$
s(p) = 1 - \frac{1}{\moy{k}} \somme{k=0}{K(p)} kp_k,
$$
where $K(p)$ is the maximal degree of the network after the attack,
related to $p$ by Lemma~\ref{lem_as_Kp}.
\end{lemme}
\begin{proof}{
Each node of degree $k$ has $k$ stubs attached to it. Therefore the
fraction of stubs attached to all nodes of degree $k$ is equal to $kp_k/\moy{k}$.
Therefore, the total number of stubs attached to removed nodes is
$s(p)= \frac{1}{\moy{k}}\sum_{K(p)+1}^{K} kp_k$.
At the infinite limit, this is equivalent to
$s(p) = \frac{1}{\moy{k}} \somme{K(p)+1}{\infty} kp_k$, hence the result.
}\end{proof}

\medskip
\noindent
We can now give the first proof of Theorem~\ref{th_as_general}:\\
\begin{proof}{
The central point here is to understand that the network obtained
after the removal of a fraction $p$ of the nodes during a
classical attack is equivalent to a random network on which random
link failures occurred.


Indeed, a classical attack has two kinds of effects: it reduces the
maximal degree in the network by removing the nodes with highest
degree, and it removes the links attached to these nodes. If we
consider  links as pairs of randomly chosen stubs, as explained in
the preliminaries, then a classical attack removes all the stubs
attached to the removed nodes and some {\em other} stubs, which were linked
to removed stubs. Since pairs of stubs are linked
randomly, this is equivalent to randomly removing the correct number of
stubs
from the  subnetwork composed of the nodes which are not removed.
This is again equivalent to randomly removing  half as many links in  this subnetwork.

If the classical attack removes a fraction $p$ of
the nodes, the fraction of stubs attached to removed nodes
is $s(p)$, given by Lemma~\ref{lem_as_sp}.
The probability for any given stub of a remaining node to be linked
to a stub of a removed node is therefore $s(p)$.
Finally, each stub
attached to a remaining node is removed with
probability $s(p)$, which is equivalent to the removal of
the same fraction of links.

Since links in this subnetwork are constructed by choosing
random pairs of stubs, it is also a random network. Moreover,
its degree distribution is nothing but the original one
with a cutoff (which is the maximal degree after the attack):
$\left(p'_k = \frac{p_k}{\sum_{k=0}^{K(p)} p_k} = \frac{p_k}{1-p_c}\right)_{0\le k\le K(p)}$,
where $K(p)$
is the maximal degree after the attack.

We finally obtain that a classical attack is equivalent to
random link failures on a random network with known degree
distribution. The value of the threshold can therefore be derived
from Theorem~\ref{th_pa_general} on random link failures.
We then have to relate this threshold, which is the number of
stubs removed among remaining nodes, to the
number of nodes removed during classical attacks.

We can now apply Theorem~\ref{th_pa_general}
to compute the fraction $s(p_c)$ of links to remove randomly
to destroy the random network described above.
This gives
$$1-s(p_c) =\frac{\moy{k(p_c)}}{\moy{k^2(p_c)}-\moy{k(p_c)}},$$
where $\moy{k(p_c)}$ and \moy{k^2(p_c)} are the first and second moment
of the degree distribution
$\left(p'_k =  \frac{p_k}{1-p_c}\right)_{0\le k \le K(p_c)}$,
which is the degree distribution described above.
 The first two moments of this distribution are
$\moy{k(p_c)} = \somme{k=0}{K(p_c)}kp_k/(1-p_c)$ and
$\moy{k^2(p_c)} = \somme{k=0}{K(p_c)}k^2p_k/(1-p_c)$.

\noindent
We can finally transform the above relation into the claim:
\begin{eqnarray*}
1-s(p_c)
&= &
\frac{\moy{k(p_c)}}{\moy{k^2(p_c)}-\moy{k(p_c)}}\\
\frac{1}{\moy{k}}\somme{k=0}{K(p_c)} kp_k
& = &
\frac{\somme{k=0}{K(p_c)}kp_k}{\somme{k=0}{K(p_c)}k^2p_k
 - \somme{k=0}{K(p_c)}k p_k}\\
\frac{1}{\moy{k}}\somme{k=0}{K(p_c)} kp_k
& = &
\frac{\somme{k=0}{K(p_c)}kp_k}{\somme{k=0}{K(p_c)} k(k-1) p_k}\\
\moy{k}
& = &
\somme{k=0}{K(p_c)} k(k-1) p_k.\\
\end{eqnarray*}
}\end{proof}

\bigskip

The other method used to obtain this result~\cite{Newman2000Robustness,Newman2003Handbook}
relies on generating functions. As explained in the
preliminaries, the fraction $p$ of nodes of highest degrees are
marked as {\em absent}, and the others are marked as {\em present}.

Recall that $F_1(x)$ is the generating function for the probability of finding
an unmarked (\ie{} present) node
with $k$ other (marked or unmarked) neighbors at the end of a randomly
chosen link.
In our case, $F_1(x)$ therefore is:
$$F_1(x)=\frac{1}{\moy{k}}\sum_{k=1}^{K(p)} kp_k x^{k-1}.$$

\noindent
We can then prove Theorem~\ref{th_as_general} as a direct consequence
of Theorem~\ref{th-seuil-newman}:\\
\begin{proof}{
From Theorem~\ref{th-seuil-newman}, the threshold $p_c$ is reached
when $F'_1(1) = 1$.
Differentiating $F_1$ gives the result:
$F'_1(x)= \frac{1}{\moy{k}}\somme{k=2}{K(p)} k(k-1) p_k x^{k-2}$ and
$F'_1(1) = \frac{1}{\moy{k}} \somme{k=2}{K(p)} k(k-1) p_k$.
}\end{proof}

Again, the two proofs of Theorem~\ref{th_as_general} have different
advantages and drawbacks. See the comments at the end of
Section~\ref{sec_ps_gl}.

\subsubsection{The cases of Poisson and power-law networks}
\label{sec_as_app}
Theorem~\ref{th_as_general} is valid for any random network, whatever
its degree distribution. To study the behavior of Poisson and
power-law networks in case of classical attacks, we therefore only
have to apply it to these cases. More precisely, we will consider
Poisson, discrete power-law and continuous power-law networks, with
size tending towards infinity.
Comparison with simulations will be provided at the end of the
section, see Figure~\ref{fig_as_seuil}.

\begin{corollaire}
\label{cor_as_ER}
\forPoisson{},
the threshold $p_c$ for classical attacks is given by
$$
z = e^{-z} \sum_{k=0}^{K(p_c)} \frac{z^k}{(k-2)!},
$$
where $K(p_c)$ is the maximal degree of the network after the attack,
related to $p_c$ by Lemma~\ref{lem_as_Kp}.
\end{corollaire}
\begin{proof}{
Direct application of Theorem~\ref{th_as_general},
with $p_k = e^{-z}z^k/k!$.
}\end{proof}

\begin{corollaire} \cite{Newman2000Robustness}
\label{cor_as_vsf}
\forVsf{},
the threshold $p_c$ for classical attacks is given by
$$
\harmo{K(p_c)}{\alpha-2} - \harmo{K(p_c)}{\alpha -1}
= \zeta(\alpha -1),
$$
\defharmo{},
and $K(p_c)$ is the maximal degree of the network after the attack,
related to $p_c$ by Lemma~\ref{lem_as_Kp}.
\end{corollaire}
\begin{proof}{
Direct application of Theorem~\ref{th_as_general},
with $p_k = k^{-\alpha}/\zeta(\alpha)$.
}\end{proof}

\begin{corollaire} \cite{Cohen2001Attack}
\label{cor_as_sf}
\forSf{},
the threshold $p_c$ for classical attacks is given by
$$
\left(\frac{K(p_c)}{m}\right)^{2-\alpha} -2
\ \  =\ \  \frac{2-\alpha}{3-\alpha}\ m
\left(\left(\frac{K(p_c)}{m}\right)^{3-\alpha}-1\right),
$$
where $K(p_c)$ is the maximal degree of the network after the attack,
related to $p_c$ by Lemma~\ref{lem_as_Kp}.
\end{corollaire}
\begin{proof}{
From Theorem~\ref{th_as_general}, we have:
$\left(\sum_{k=m}^{K(p_c)} k(k-1)p_k\right)/\moy{k} = 1$, thus
$\moy{k}=\sum_{k=m}^{K(p_c)}k^2p_k - \sum_{k=m}^{K(p_c)}Kr_k$.
We have $p_k = m^{\alpha-1}(k^{-\alpha+1}-(k+1)^{-\alpha+1})$,
and $\moy{k} = m(\alpha-1)/(\alpha-2)$, from Lemma~\ref{lem_moments_sf}.
From this, switching back to the continuous form, we obtain
$\somme{k=m}{K(p_c)} k p_k = (\alpha-1) m^{\alpha-1}\int_m^{K(p_c)} k k^{\alpha-1} \mathrm{d}k =
\frac{\alpha-1}{-\alpha + 2} m^{\alpha-1} \left( K(p_c)^{-\alpha+2} - m^{-\alpha + 2} \right)$
and similarly for the second moment,
$
\somme{k=m}{K(p_c)} k^2 p_k =
\frac{\alpha-1}{-\alpha + 3} m^{\alpha-1} \left(K(p_c)^{-\alpha+3} - m^{-\alpha + 3}\right)
$.

\noindent
The above equality then becomes:
$$
(\alpha-1)m^{\alpha-1}
\left[ \frac{K(p_c)^{-\alpha+3} - m^{-\alpha+3}}{-\alpha + 3}
 - \frac{K(p_c)^{-\alpha+2} - m^{-\alpha+2}}{-\alpha + 2} \right]
 = m\frac{\alpha-1}{\alpha-2}.
$$
We can finally transform this equation to obtain the result:
\begin{eqnarray*}
 \frac{m^{-\alpha+3}}{-\alpha+3}
  \left[ \left( \frac{K(p_c)}{m} \right)^{-\alpha+3} - 1 \right]
 - \frac{m^{-\alpha+2}}{-\alpha+2}
  \left[ \left( \frac{K(p_c)}{m} \right)^{-\alpha+2} - 1 \right]
& = &
\frac{m^{-\alpha+2}}{\alpha-2}\\
  \frac{m}{-\alpha+3}
  \left[ \left( \frac{K(p_c)}{m} \right)^{-\alpha+3} - 1 \right]
 - \frac{1}{-\alpha+2}
  \left[ \left( \frac{K(p_c)}{m} \right)^{-\alpha+2} - 1 \right]
& = &
\frac{1}{\alpha-2}\\
  \frac{m}{-\alpha+3}
  \left[ \left( \frac{K(p_c)}{m} \right)^{-\alpha+3} - 1 \right]
 - \frac{1}{-\alpha+2}
  \left[ \left( \frac{K(p_c)}{m} \right)^{-\alpha+2} - 2 \right]
& = &
0
\end{eqnarray*}
}\end{proof}

\bigskip

Numerical evaluations of these results can be done by computing the
maximal degree $K(p_c)$ with the appropriate corollary,
 and then injecting it into Lemma~\ref{lem_as_Kp} to evaluate $p_c$.
For Poisson and discrete power-law networks (Corollaries~\ref{cor_as_ER} and~\ref{cor_as_vsf}),
the equations can be solved in a
similar way as  computing the maximal degree of a random
network, see Section~\ref{secK}.
For continuous power-law networks (Corollary~\ref{cor_as_sf}),
the equation
 can be solved using a computer algebra
system \cite{mapleurl}. Notice that this equation is  not defined for
$\alpha=3$; one may obtain the threshold for this value as the limit
of its values when $\alpha$ tends to it.

Moreover, still in the case of continuous power-law networks,
one may use the following result,
 which is simpler and more
precise (since it allows non-integer values for the degree),
instead of Lemma~\ref{lem_as_Kp}.

\begin{lemme} \cite{Cohen2001Attack}
\label{lem_as_sf}
\forSf{},
the maximal degree $K(p)$ after the removal of a fraction $p$ of the nodes
during classical attacks is related to $p$ by
$$p = m^{\alpha-1}K(p)^{-\alpha+1}.$$
\end{lemme}
\begin{proof}{
From Lemma~\ref{lem_as_Kp}, we have
$
p = \sum_{K(p)+1}^\infty p_adj$,
which we can approximate to be equal to $\sum_{K(p)}^\infty p_k$.
Switching back to the continuous definition, we have that
$p
= (\alpha-1)m^{\alpha-1}\int_{K(p)}^{\infty} x^{-\alpha} \mathrm{d}x
= (\alpha-1)m^{\alpha-1}\left[ \frac{k^{-\alpha+1}}{-\alpha+1} \right]_K^\infty
= (\alpha-1)m^{\alpha-1}K^{-\alpha+1}/(\alpha-1) = m^{\alpha-1}K^{-\alpha+1}$,
hence the result.
}\end{proof}

\bigskip

We plot numerical evaluations of these results in
Figure~\ref{fig_as_seuil}, together with experimental results.
We also give in Table~\ref{tab_as} the thresholds for specific
values of the exponent and the average degree.

\begin{figure}[!h]
\begin{center}
\includetroisres{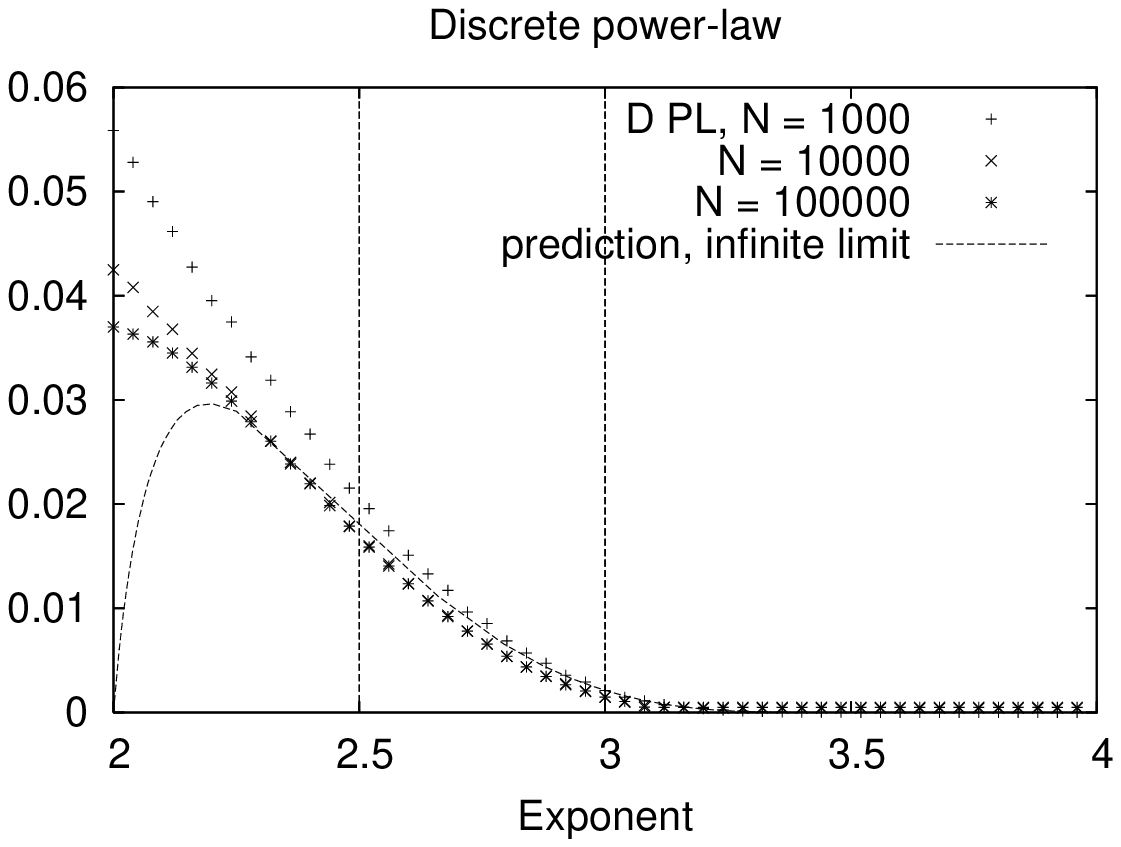}
\end{center}
\caption{
Thresholds for classical attacks.
\lefttoright
\refplotsseuil
}
\label{fig_as_seuil}
\end{figure}

\degc{as}{classical attacks}{0.056 & 0.038 & 0.18 & 0.19 & 0.018 & 0.017 & 0.08 & 0.09}{0.025 & 0.012 & 0.05 & 0.05 & 0.002 & 0.0015 & 0.03 & 0.035}

It appears clearly that
both types of networks are
very sensitive to classical attacks: only a few percents of the nodes
have to be removed to destroy them.
Moreover, the thresholds for power-law networks are
much lower than the ones for Poisson  networks: they are almost
one order of magnitude smaller than for
comparable Poisson networks.
These are certainly the main results on the topic and we will deepen them
in the rest of the section.

\subsection{Link point of view of classical attacks.}
\label{sec_asa}

The classical attack strategy removes highest degree nodes first. Since
in a power-law network there are some very high degree nodes,
this leads in this case to the removal of a huge number of links. One may
then wonder if its efficiency on power-law networks is due to
the fact that the number of removed links is much larger than in the case
of random failures. Likewise, one may
wonder if the fact that a classical attack removes  much
more links in a power-law network than in a Poisson one is the cause of its difference of efficiency on these two types of networks.
These explanations actually have been
proposed by some authors as an intuitive explanation of the
results presented above.

Figure~\ref{fig_asa_cc} displays the behaviors observed for these three
types of networks. One can  see there that the thresholds for Poisson
and power-law networks are much closer than from the node point of
view, see Figure~\ref{fig_as_cc}. One may also observe that, though
there are  significant differences, when one removes from power-law
networks as many links as what is needed to destroy a Poisson network
with the same average degree, then the size of the largest connected
component becomes very small.


The questions we address here therefore are: how many links have been
removed when we reach the threshold for classical attacks? Is this
number similar for power-law and Poisson networks? And is it similar
to the threshold for link failures? We will investigate these questions
more precisely in this section by first introducing a general result and
then apply it to the cases under concern.



\begin{figure}[!h]
\begin{center}
\includetroisres{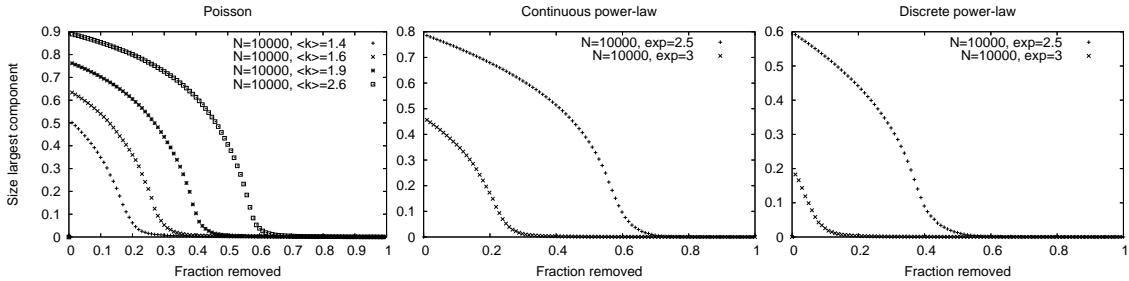}
\end{center}
\caption{
Size of the largest connected component as a function of the
fraction of {\em links} removed during classical attacks.
\lefttoright
\refplotscc
}
\label{fig_asa_cc}
\end{figure}


\subsubsection{General results}

Our aim here is to prove the following general result, which gives
the fraction $m(p)$ of links removed when a fraction $p$ of
the nodes are removed during classical attacks.

\begin{theoreme}
\label{th_asa_general}
In a large random network, the fraction $m(p)$ of links removed
when a fraction $p$ of the nodes have been removed during a classical
attack is
$$m(p) = 2s(p) - s(p)^2,$$
where $s(p)$ is the fraction of stubs attached to removed nodes,
and is related to $p$ by Lemma~\ref{lem_as_sp}.
\end{theoreme}
\begin{proof}{
Let us consider a network in which we remove a fraction
$p$ of the nodes during a classical attack.
Let $s(p)$ be the fraction  of the stubs in the network attached to the
removed nodes
($s(p)$ is linked to $p$ by Lemma~\ref{lem_as_sp}).
Each stub attached to a remaining node is therefore kept with probability $(1-s(p))$, the
fraction of pairs of stubs linking non removed nodes is therefore
$(1-s(p))^2$. This last quantity is the fraction of non removed
links,
and $m(p) = 1-(1-s(p))^2 = 2s(p) - s(p)^2$ is finally the fraction of
removed links, hence the result.
%
%
}\end{proof}

\subsubsection{The cases of Poisson and power-law networks}
Theorem~\ref{th_asa_general} is valid for any random network,
whatever its degree distribution.
It makes it possible to
compute the fraction of links removed at the threshold for classical attacks.
To study the behavior of Poisson and power-law networks,
we therefore only have to apply it to these cases.
More precisely, we will consider Poisson, continuous power-law and
discrete-power-law networks.

In each case, one first has to compute the (node) threshold $p_c$ for classical
attacks using the appropriate corollary in
Section~\ref{sec_as_app},
then use Lemma~\ref{lem_as_sp} to
obtain $s(p_c)$, before applying Theorem~\ref{th_asa_general}.

In the case of continuous power-law networks,
it is possible to obtain $s(p_c)$ from $p_c$ more easily as follows.

\begin{lemme} \cite{Cohen2001Attack}
\label{lem_asa_sf}
\forSf{},
when a fraction $p$ of the nodes is removed  during a classical attack,
the fraction $s(p)$ of stubs attached
to removed  nodes  is given by
$$
s(p) = p^{(2-\alpha)/(1-\alpha)}.
$$
\end{lemme}
\begin{proof}{
From Lemma~\ref{lem_as_sp}, we know that
$s(p)= (\somme{k=K(p)+1}{\infty} k p_k)/\moy{k}$, which we can
approximate to be equal to $(\somme{k=K(p)}{\infty} k p_k)/\moy{k}$.
Moreover, we have $\moy{k} = \frac{\alpha-1}{\alpha-2}m$ from
Lemma~\ref{lem_moments_sf}.

Switching back to the continuous case, we therefore have
$$s(p) = \frac{1}{\moy{k}} (\alpha-1)m^{\alpha-1}\int_{K(p)}^{\infty}k k^{-\alpha} \mathrm{d}k
      = \frac{1}{\moy{k}} (\alpha-1)m^{\alpha-1}
        \left[\frac{k^{-\alpha+2}}{-\alpha+2} \right]_{K(p)}^\infty
      = m^{\alpha-2}K(p)^{-\alpha+2}.$$
From Lemma~\ref{lem_as_sf},
we finally obtain $K(p)= m p^{1/(-\alpha+1)}$,  hence the result.
}\end{proof}

\bigskip

We plot numerical evaluations of these results in
Figure~\ref{fig_asa_seuil}, together with experimental results.
We also give in Table~\ref{tab_asa} the thresholds for specific
values of the exponent and the average degree.

\begin{figure}[!h]
\begin{center}
\includetroisres{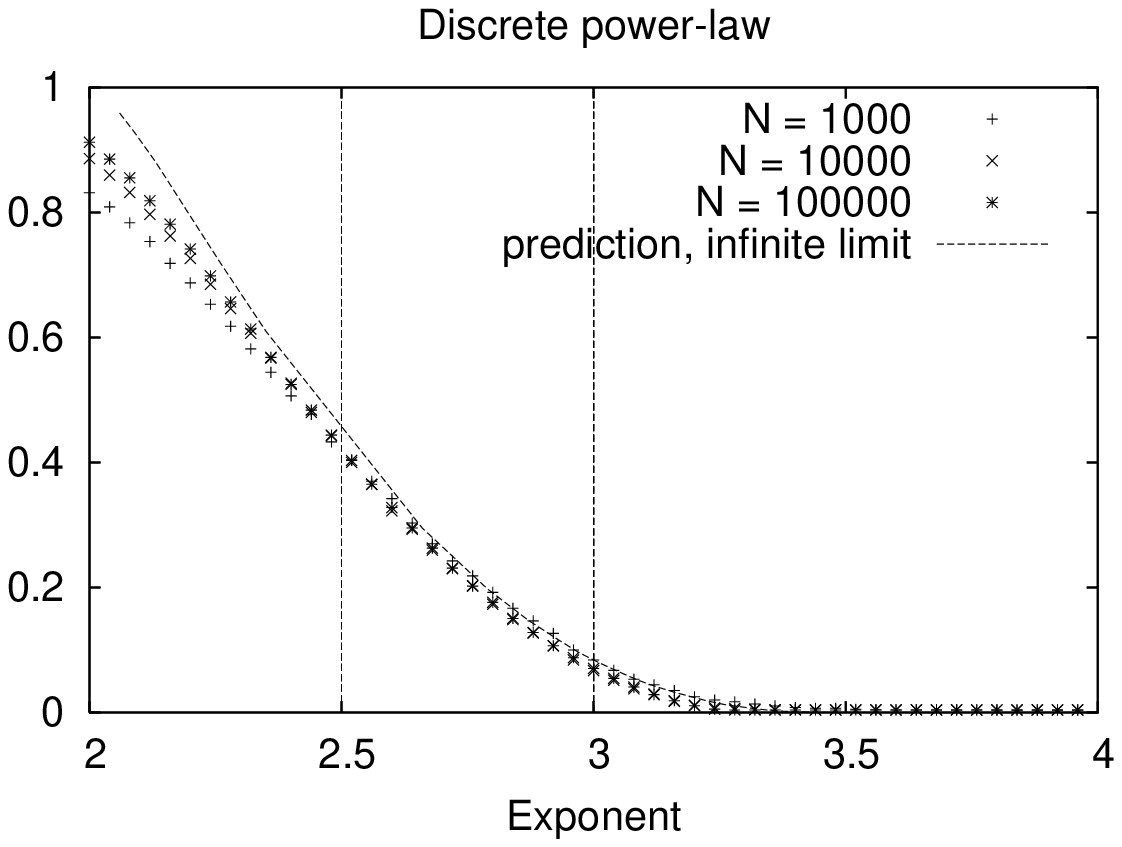}
\end{center}
\caption{
Thresholds for the link point of view of classical attacks.
\lefttoright
\refplotsseuil
}
\label{fig_asa_seuil}
\end{figure}

\degc{asa}{the link point of view of classical attacks}{0.62 & 0.6 & 0.5 & 0.6 & 0.45 & 0.42 & 0.28 & 0.4}{0.3 & 0.24 & 0.15 & 0.28 & 0.08 & 0.07 & 0.1 & 0.2}

The results are striking: the thresholds are much larger from the link
point of view than from the node point of view (see Table~\ref{tab_as}
for comparison).
More importantly, while the number of nodes to be removed is much lower
for power-law networks than for Poisson ones, the corresponding number
of links if similar for both kinds of networks:
 the links thresholds are similar.

The conclusion from these observations is that the fact that
power-law networks are rapidly destroyed during classical attacks may be
viewed as a consequence of the fact that many links are removed.
It is however important to notice that the obtained behavior for
power-law networks is not the same as the one obtained if we remove
the same amount of links at random
(see Figure~\ref{fig_pa_seuil} and Table~\ref{tab_pa} for comparison).
Therefore, although
the amount of removed links is huge and this plays a role in
the behavior of power-law networks, this is not sufficient to
explain the observed behavior.
This means that the links attached to highest degree nodes play a more
important role regarding the network connectivity than random links.

\subsection{New attack strategies.}
\label{sec_na}

In this section we introduce two very simple new attack strategies,
one targeting nodes (Section~\ref{sec_nas}) and the other
targeting links (Section~\ref{sec_naa}).
These strategies are close to random failures.
Our aim is not to provide efficient attack strategies,
but rather to deepen our understanding of previous results.

These two new attack strategies rely on the following observation.
We have seen (Theorem~\ref{thseuil}) that a random network
with size tending towards infinity has a giant component if
$ \moy{k^2} - 2\moy{k} > 0$. This is equivalent to the condition
$p_1 < \somme{k=3}{\infty} k(k-2)p_k$
(where $p_k$ is the fraction of nodes of degree $k$).
The fraction of nodes of degree $1$ in the network therefore plays
a key role. The two attack strategies are based on the idea that
increasing this fraction should quickly break the network.

Since our aim here is not to compute the exact value of the
threshold, but rather to understand a general behavior,
we will only consider in the sequel the case of networks
with size tending towards infinity.

\subsubsection{Almost-random node attacks}
\label{sec_nas}
The first attack strategy simply consists of randomly removing
nodes of degree at least 2. We call it the almost-random node attack
strategy.

We first present a general result for this strategy, then apply
it to the special cases under concern:
Poisson and power-law (both discrete and continuous) networks.
Figure~\ref{fig_nas_cc} displays the behaviors observed for these three
types of networks.

\begin{figure}[!h]
\begin{center}
\includetroisres{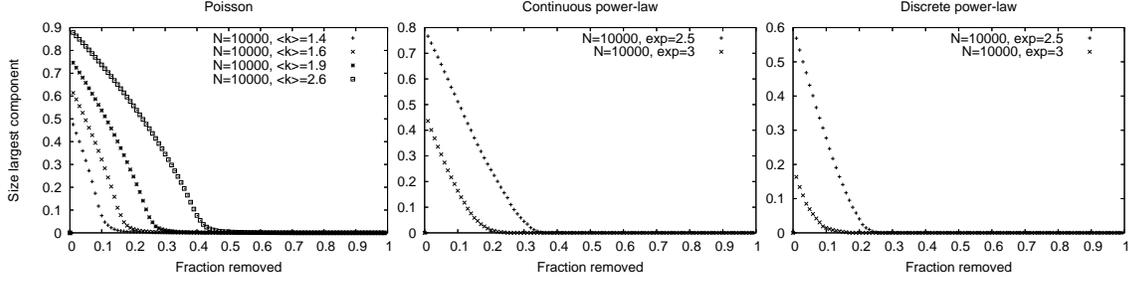}
\end{center}
\caption{
Size of the largest connected component as a function of the
fraction of nodes removed during almost-random node attacks.
\lefttoright
\refplotscc
}
\label{fig_nas_cc}
\end{figure}

Although this strategy is barely different from random node failures,
it is actually much more efficient than failures
(see Figure~\ref{fig_ps_cc} for comparison).
In particular, it has a finite threshold for all
the types of networks we consider.


\begin{theoreme}
\label{th_nas_general}
The threshold $p_c$ for almost-random node attacks
for large random networks with degree distribution $p_k$,
is bounded by
$$ p_c < 1-p_1-p_0.$$
\end{theoreme}
\begin{proof}{
When all nodes that had initially a degree higher that
one have been removed, then the network surely has no
giant component anymore since all remaining nodes have degree at
most $1$.
All nodes of degree higher than one represent a fraction
$1-p_1-p_0$ of all nodes;
this is therefore an upper bound for the threshold
for this attack strategy.
}\end{proof}

Theorem~\ref{th_nas_general} is valid for any random network, whatever
its degree distribution. To study the behavior of Poisson and
power-law networks in case of classical attacks, we therefore only
have to apply it to these cases. More precisely, we will consider
Poisson and power-law networks (both discrete and continuous), with
size tending towards infinity.
Comparison with simulations will be provided at the end of the
section, see Figure~\ref{fig_nas_seuil}.

\begin{corollaire}
\label{cor_nas_er}
\forPoisson{},
the threshold $p_c$ for almost-random node attacks is bounded by
$$p_c < 1- e^{-z}(z+1).$$
\end{corollaire}
\begin{proof}{
Direct application of Theorem~\ref{th_nas_general} with $p_k = e^{-z} z^k/k!$.
}\end{proof}

\begin{corollaire}
\label{cor_nas_sf}
\forSf{},
the threshold $p_c$ for almost-random node attacks is bounded by
$$p_c < 1 - m^{\alpha-1} (1 - 2^{-\alpha+1}).$$
\end{corollaire}
\begin{proof}{
Direct application of Theorem~\ref{th_nas_general},
with $p_k = m^{\alpha-1}(k^{-\alpha+1} - (k+1)^{-\alpha+1})$.
}\end{proof}

\begin{corollaire}
\label{cor_nas_vsf}
\forVsf{},
the threshold $p_c$ for almost-random node attacks is bounded by
$$p_c < 1- 1/\zeta(\alpha).$$
\end{corollaire}
\begin{proof}{
Direct application of Theorem~\ref{th_nas_general}
with $p_k = k^{-\alpha}/\zeta(\alpha)$.
}\end{proof}

We plot experimental results for the value of the threshold in
Figure~\ref{fig_nas_seuil}, as well as the upper bounds given above.
We also give in Table~\ref{tab_nas} the thresholds for specific
values of the exponent and the average degree.

\begin{figure}[!h]
\begin{center}
\includetroisres{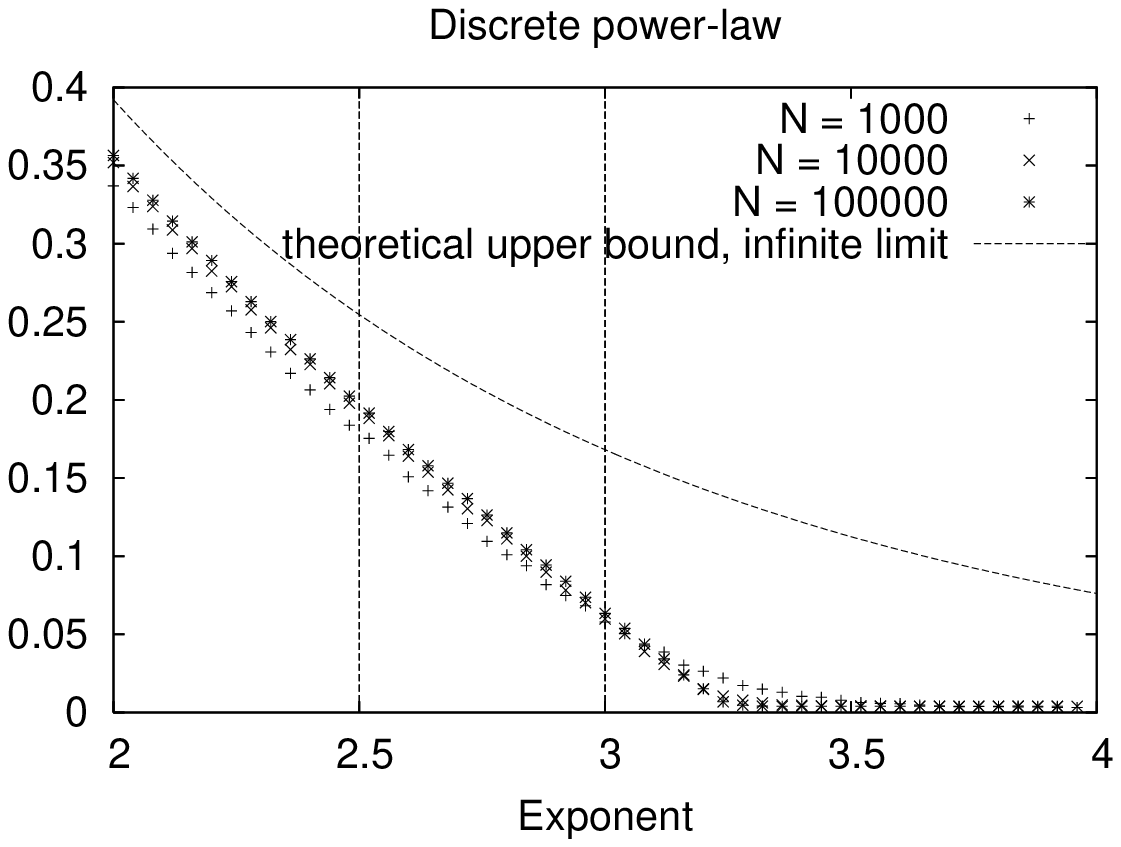}
\end{center}
\caption{
Thresholds and upper bounds for almost-random node attacks.
\lefttoright
\refplotsseuil
}
\label{fig_nas_seuil}
\end{figure}

\degcn{nas}{almost-random node attacks}{0.35 & 0.29 & 0.73 & 0.43 & 0.25 & 0.2 & 0.57 & 0.25}{0.25 & 0.16 & 0.48 & 0.16 & 0.17 & 0.06 & 0.41 & 0.10}

We recall that our aim here is not to obtain an efficient attack
strategy, but to study the ability of a strategy
very similar to random failures to have the same qualitative
behavior as classical attacks, namely to display a finite threshold
for power-law networks.

To this regard, the values of the thresholds displayed in
Table~\ref{tab_nas} are quite large (one has to remove a large fraction
of the nodes do destroy the networks), but remain significantly lower
than $1$ and much lower than the thresholds for node failures
(see Table~\ref{tab_ps}). This shows that the efficiency of classical
attacks relies in part on simple properties like
removing nodes of degree larger than $1$.

\subsubsection{Almost-random link attacks}
\label{sec_naa}

The other attack strategy consists of randomly removing
links between nodes of degree at least 2, \ie\ a node of degree $1$ will always stay connected during the attack. We call it the almost-random
link attack strategy.

We first present a general result for this strategy, then apply
it to the special cases under concern:
Poisson and power-law (both discrete and continuous) networks.
Figure~\ref{fig_naa_cc} displays the behaviors observed for these three
types of networks.

\begin{figure}[!h]
\begin{center}
\includetroisres{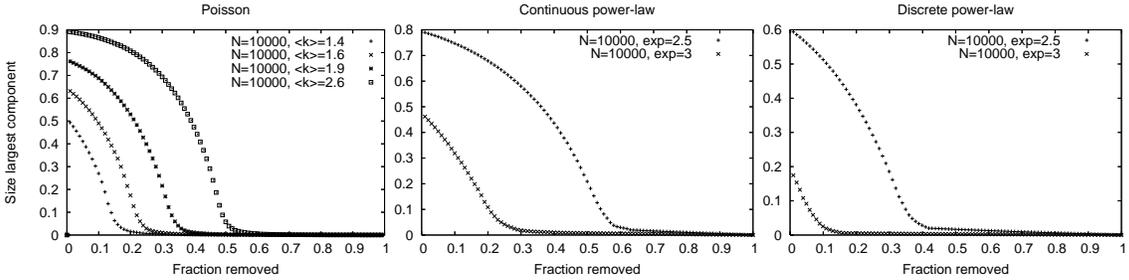}
\end{center}
\caption{
Size of the largest connected component as a function of the
fraction of links removed during almost-random link attacks.
\lefttoright
\refplotscc
}
\label{fig_naa_cc}
\end{figure}

Although this strategy is barely different from random link failures,
it is actually much more efficient than failures (see Figure~\ref{fig_pa_cc} for comparison).
In particular, it has a finite threshold for all
the types of networks we consider, which makes it much more efficient
on power-law networks.



\begin{theoreme}
\label{th_naa_general}
The threshold $m_c$ for the almost-random link attack strategy
for large random networks with maximal degree sublinear in the
number of nodes and degree distribution $p_k$
is bounded by
$$ m_c < 1 - \frac{2p_1}{\moy{k}} + \frac{p_1^2}{\moy{k}^2}.$$
\end{theoreme}
\begin{proof}{
The threshold is bounded by the fraction of links between two nodes of degree
at least $2$. Indeed, when all such links have been removed, the only links
left are between a node of degree $1$ and another node. Therefore
the network is nothing but a set of disjoint stars (each central node
being connected to nodes of degree~$1$). The size of the largest
component therefore is
less than $K+1$, where $K$ is the maximal degree in the original network.
Hence, if $K$ is sublinear with respect to the number
of nodes, so is the size of the largest component after the attack.

Let us now evaluate the number of links between nodes of degree at
least $2$. This quantity is $|E|$ minus the number of links incident
to at least one node of degree $1$. The number of such links is given
by the number of nodes of degree $1$, minus the number of links between
two nodes of degree $1$.

The number of links between two nodes of degree $1$ can be evaluated
as follows. There are $N p_1$ nodes of degree $1$, each of them having
a probability ${N p_1}/{2|E|}$ of being connected to another node of
degree $1$\footnote{This is an approximation of the real value $(N-1)p_1/(2|E|-1)$}.
Therefore the number of {\em nodes} of degree $1$
adjacent to another node of degree $1$ is
${Np_1\cdot Np_1}/{2|E|} = {Np^2_1}/{\moy{k}}$. Finally, the number
of links between two such nodes is ${N p^2_1}/{2\moy{k}}$.

From this we have that the number of links adjacent to at least one
node of degree $1$ is: $Np_1 - {N p^2_1}/{2\moy{k}}$,
and the number of links {\em not}
adjacent to any node of degree $1$ is:
$|E| - Np_1 + {N p^2_1}/{2\moy{k}}.$ The fraction of such links
therefore is
$ 1 - \frac{2p_1}{\moy{k}} + \frac{p_1^2}{\moy{k}^2}$,
hence the result.
}\end{proof}


\begin{corollaire}
\label{cor_naa_ER}
\forPoisson{},
the threshold $m_c$ for almost-random link attacks is bounded by
$$m_c < 1-2e^{-z}+e^{-2z}.$$
\end{corollaire}
\begin{proof}{
Direct application of Theorem~\ref{th_naa_general} with $p_k = e^{-z}z^k/k!$,
the maximal degree of the network being sublinear in the size of the
network (see Lemma~\ref{lem_K_er}).
}\end{proof}


\begin{corollaire}
\label{cor_naa_sf}
\forSf{},
the threshold $m_c$ for almost-random link attacks is bounded by
$$m_c <
1-
\frac{2(\alpha-2)m^{\alpha-2}(1-2^{-\alpha+1})}{(\alpha-1)}
+ \left(\frac{(\alpha-2)m^{\alpha-2}(1-2^{-\alpha+1})}{(\alpha-1)}\right)^2.
$$
\end{corollaire}
\begin{proof}{
Direct application of Theorem~\ref{th_naa_general}
with $p_k = m^{\alpha-1}(k^{-\alpha+1} - (k+1)^{-\alpha+1})$, $\moy{k} = m\ \frac{\alpha-1}{\alpha-2}$ (Lemma~\ref{lem_moments_sf}),
the maximal degree of the network being sublinear in the size of the
network (Lemma~\ref{lem_K_sf}).
}\end{proof}

\begin{corollaire}
\label{cor_naa_vsf}
\forVsf{},
the threshold $m_c$ for almost-random link attacks is bounded by
$$m_c < 1 - \frac{2\zeta(\alpha -1) -1}{\zeta^2(\alpha-1)}.$$
\end{corollaire}
\begin{proof}{
Direct application of Theorem~\ref{th_naa_general} with $p_k = k^{-\alpha}/\zeta(\alpha)$, $\moy{k} = \frac{\zeta(\alpha-1)}{\zeta(\alpha)}$ (Lemma~\ref{lem_moments_vsf}),
the maximal degree of the network being sublinear in the size of the
network (Lemma~\ref{lem_K_vsf}).
}\end{proof}

We plot experimental results for the value of the threshold in
Figure~\ref{fig_naa_seuil}, as well as the upper bounds given above.
We also give in Table~\ref{tab_naa} the thresholds for specific
values of the exponent and the average degree.

\begin{figure}[!h]
\begin{center}
\includetroisres{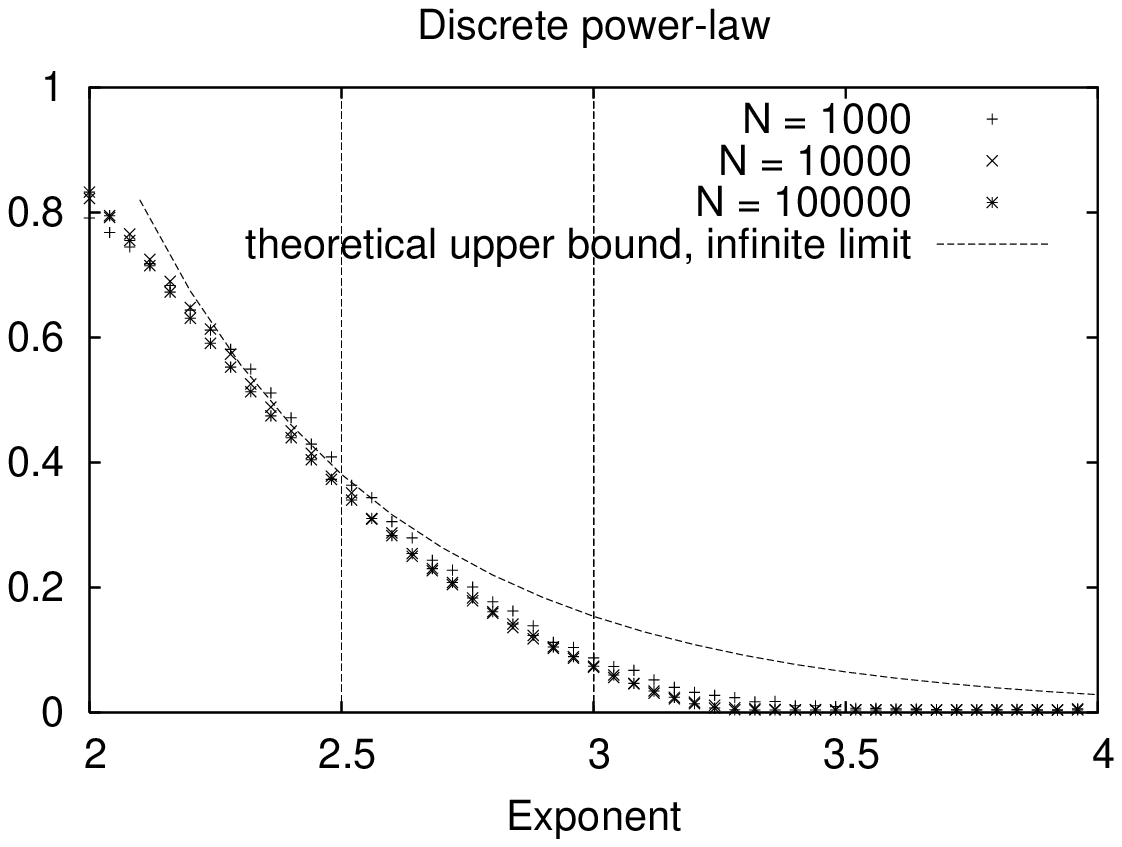}
\end{center}
\caption{
Thresholds and upper bounds for almost-random link attacks.
\lefttoright
\refplotsseuil
}
\label{fig_naa_seuil}
\end{figure}

\degcn{naa}{almost-random link attacks}{0.62 & 0.55 & 0.86 & 0.51 & 0.37 & 0.35 & 0.72 & 0.32}{0.39 & 0.22 & 0.64 & 0.23 & 0.15 & 0.07 & 0.57 & 0.15}

As in the case of almost-random node attacks, the values of
the thresholds are quite large but remain significantly lower than $1$.
Since our aim  is still to study the ability of a strategy
very similar to random failures to display a finite threshold, this
result is satisfactory.
This shows that the efficiency of classical attacks relies in part
on simple properties like removing
links between nodes of degree at least $2$.

Going further, one may notice that almost-random link attacks perform
better than classical attacks in terms of the number of removed links (see Table~\ref{tab_asa}).
This shows that classical attacks, although they focus on high degree
nodes, actually remove many links connected to nodes of degree one,
which play little role in the connectivity of the network.
The simple almost-random strategy, on the opposite,
focuses on those links which really disconnect the network.

\subsection{Conclusion on attacks.}
\label{sec_a_conclu}

There are two main formal conclusions for this section. First, as expected
from the empirical results discussed in introduction, power-law networks
are very sensitive to classical attacks,  much more than Poisson networks.
Second, the link point of view shows that many links are actually removed
when the thresholds for classical attacks are reached. Moreover, very simple
attack strategies close to random node or link failures also lead to finite
(and reasonably small) thresholds.

Altogether, these results make it possible to discuss
precisely the efficiency of classical
attacks. First, although the number of links removed during
such  attacks is huge, this is not sufficient to explain the
collapse of the network. Indeed the removal of the same number
of links at random does not collapse the network. Second, the
number of removed links during classical attacks in a Poisson
network and in a power-law network are very similar. This
moderates the conclusion that power-law networks are particularly
sensitive to classical attacks, since in terms of links both are equally robust.


Finally, the attack strategies we introduced, which are very close
to random failures, show that the efficiency of classical attacks
relies strongly on simple properties like  removing nodes of
degree larger than $1$ and links between nodes of degree at least $2$.



%% file: real_graphs.tex
\section{Resilience of real-world networks.}
\label{sec_real}

So far we have presented theoretical results as well as experiments on failures and attacks
for different types of random networks.
We have seen that, apart from the general shape
of their degree distributions,  precise properties of the
networks under concern, like for instance their fraction of nodes of
degree $1$, may play a crucial role in their behavior in case of
failures or attacks. Other properties not captured by the models, like
clustering for instance, may also play an important role. In order to
give some insight on the practical incidence of the
results above, we present in this section empirical results on real-world
complex networks.

We will  consider the following real-world cases, which are representative
of the ones considered
in studies of this field.
 The {\em actor} network is composed of movie actors which are
 linked together if they played in the same movie.
It is obtained from the {\em Internet Movie
 DataBase} \cite{imdburl2}.
 See~\cite{watts98collective,barabasi02linked} for results on this network.
 The {\em co-authoring} network is composed scientific authors, two
authors  being linked together if they signed a paper (present in the
archive)  together.
It is obtained from the {\em arXiv} site
 \cite{arxivurl2}
See~\cite{newman01scientific1,newman01scientific2} for results on this
kind of networks.
 The {\em cooccurrence} network is composed of words of the  {\em Bible}~\cite{bibleurl2},
 two words being linked  if they belong to the same sentence.
 See~\cite{ferrer01small} for results on this kind of networks.
 The {\em internet1} and the {\em internet-core} networks are
 two internet maps
 where the nodes are routers in the internet, two routers being linked
 if they are at one hop at the {\sc ip} level.
 These networks are described and studied in \cite{govindan00heuristics,Guillaume2005IPTopology}.
 The {\em protein} network is composed of proteins,
two proteins being linked if they interact together in a cell.
It is obtained from \cite{barabasilaburl}. See \cite{jeong00largescale}
for results on this kind of networks.
 The {\em www} network is composed of web pages, two pages being linked
 together if there is a hyperlink from one of them to the other. This
 sample is obtained from~\cite{barabasilaburl}.
 See~\cite{albert99diameter} for results on this kind of networks.
 Finally, the {\em p2p} network is a set of exchanges between peers, captured
 in a running peer-to-peer system, two peers being linked if they
 exchanged a file during the capture.
 See~\cite{Guillaume2004P2P,Guillaume2005P2P} for precise description and results on this network.
We give in Table~\ref{tab_avgdeg_real} the number of nodes and links of these
networks, together with their average degree.
More detailed description of these data can be found in the references.

\begin{table}
\begin{center}
\begin{tabular}{|l|l|l|l|}
\hline
network &  \# nodes & \# links & avg. deg.\\
\hline
actor &  392\,341 & 15\,038\,083& 76.658\\
co-authoring &  16\,402& 29\,552 & 3.603\\
cooccurrence & 9\,297 &  392\,066 &  84.342\\
internet1 & 228\,263 & 320\,149 & 2.805\\
internet-core & 75\,885 & 357\,317 & 9.417\\
protein & 2\,115 & 4\,480 & 4.236\\
www & 325\,729 & 1\,497\,135 & 9.192\\
p2p & 159\,541 & 17\,454\,369 & 218.807\\
\hline
\end{tabular}
\end{center}
\caption{Number of links, nodes and average degree of the real-world networks under study.}
\label{tab_avgdeg_real}
\end{table}

We will discuss these networks' resilience to failures
and attacks and will explain it using our knowledge of their structure
and the results presented in this paper.
Note that these networks are obtained through intricate measurement
procedures. It is known that the obtained views of the networks cannot
always be considered as representative of the original networks, see for instance~\cite{measurement,tarissan2009efficient}
in the general case and~\cite{Guillaume2005IPTopology,Vespignani2005TCS,chen02origin,chang01inferring,IPconn,Moore2005Traceroute,lakhina02sampling} in the case of the internet.
However, our goal is not to address such issues 
and we do not claim to give exact results on these networks.
Our goal is rather to illustrate our approach.

Figures~\ref{fig_actors} to~\ref{fig_p2p-tcp} show the obtained plots,
together with the actual degree distribution of the network.
We also present in Figures~\ref{fig_real_csf_2.5} and~\ref{fig_real_vsf_3}
the behavior of power-law random networks in case of failures and attacks,
in order to compare it to the behavior of our real-world networks.
We chose the two types of  power-law networks that yielded the most different
behaviors, namely continuous power-law networks with exponent $2.5$
and discrete power-law networks with exponent $3$,
in order to span all existing behaviors for these networks.


\newcommand{\scalecomp}{0.4}
\begin{figure}
\begin{center}
\includegraphics[scale=\scalecomp]{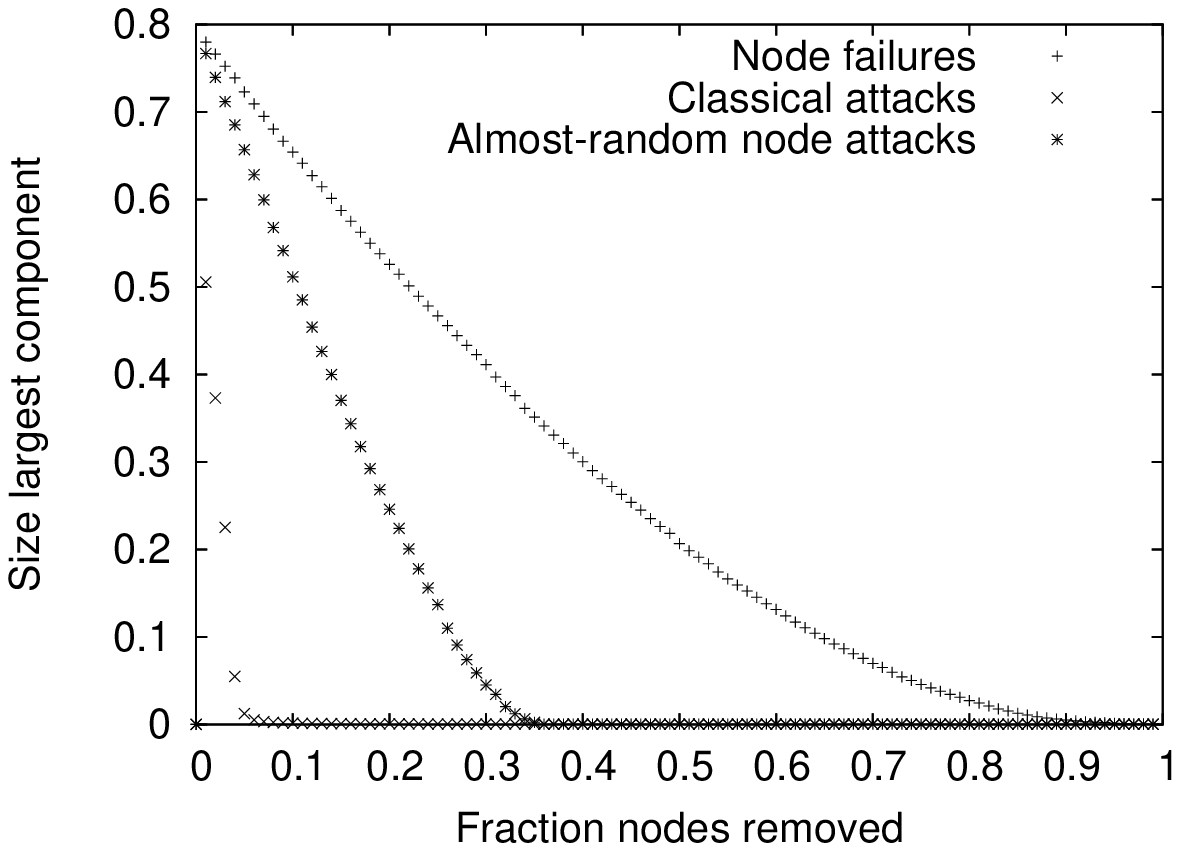}
\includegraphics[scale=\scalecomp]{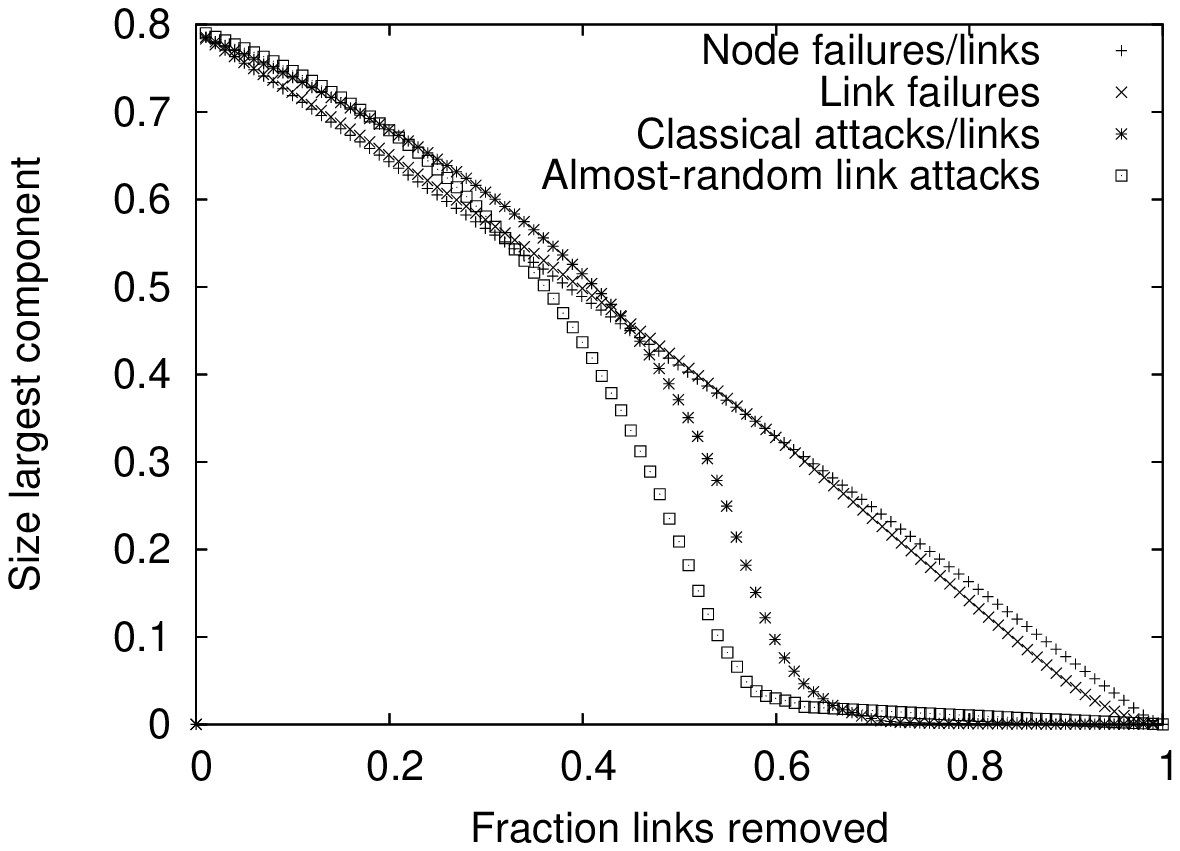}
\end{center}
\caption{Behavior of continuous power-law networks with exponent 2.5 in case of failures and attacks.}
\label{fig_real_csf_2.5}
\end{figure}

\begin{figure}
\begin{center}
\includegraphics[scale=\scalecomp]{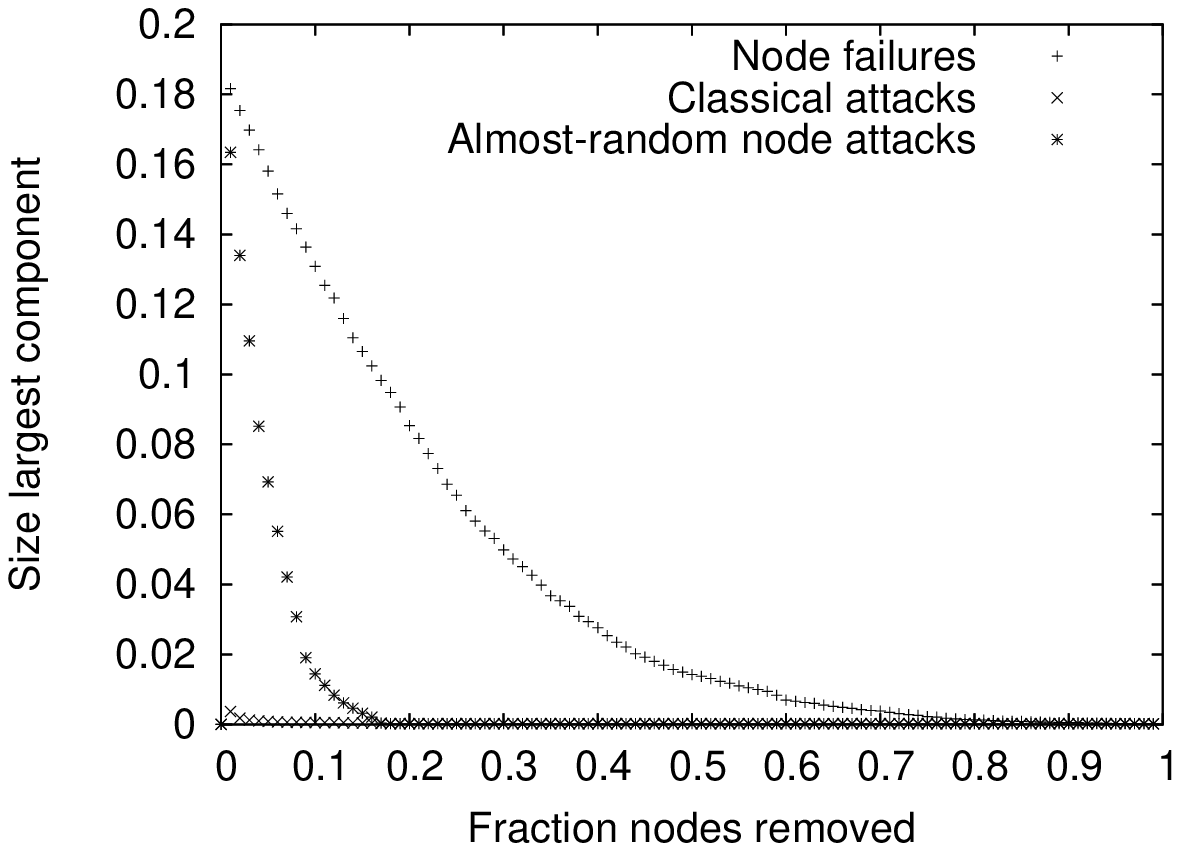}
\includegraphics[scale=\scalecomp]{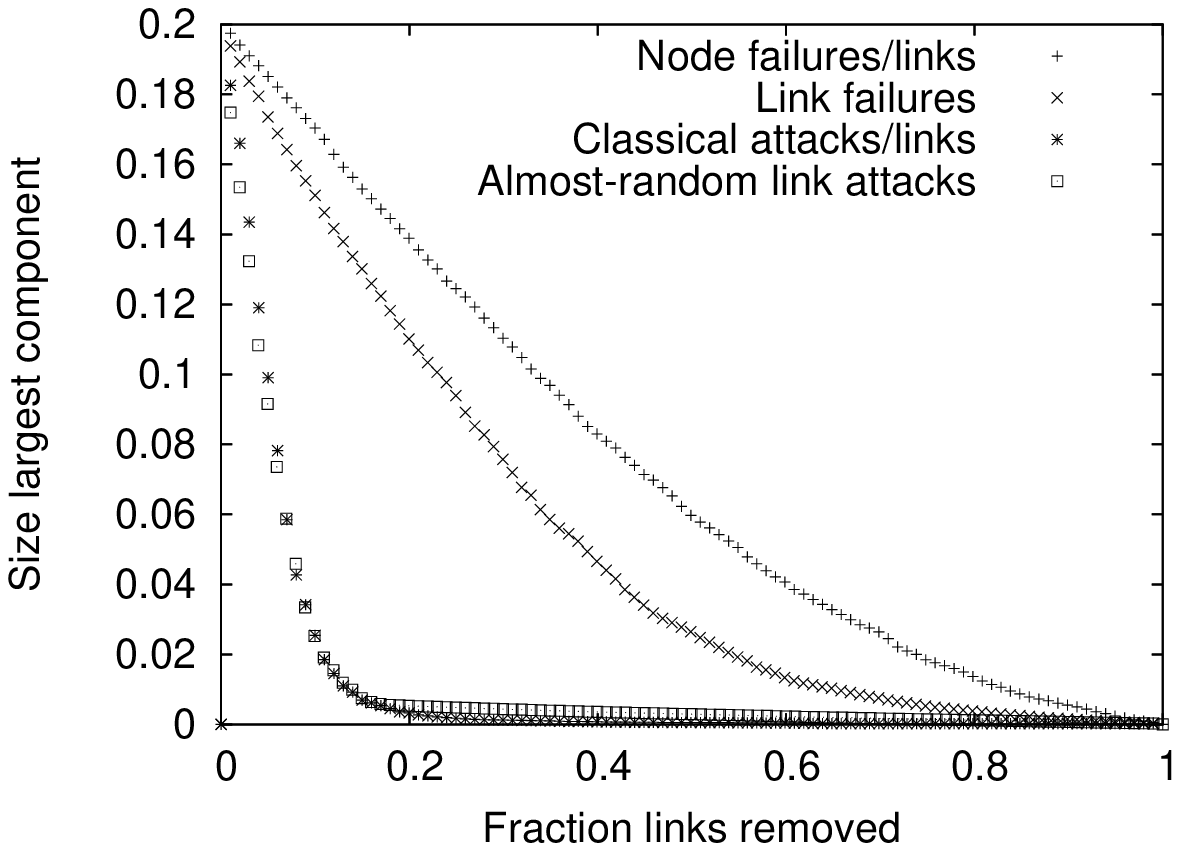}
\end{center}
\caption{Behavior of discrete power-law networks with exponent 3 in case of failures and attacks.}
\label{fig_real_vsf_3}
\end{figure}

\figreal{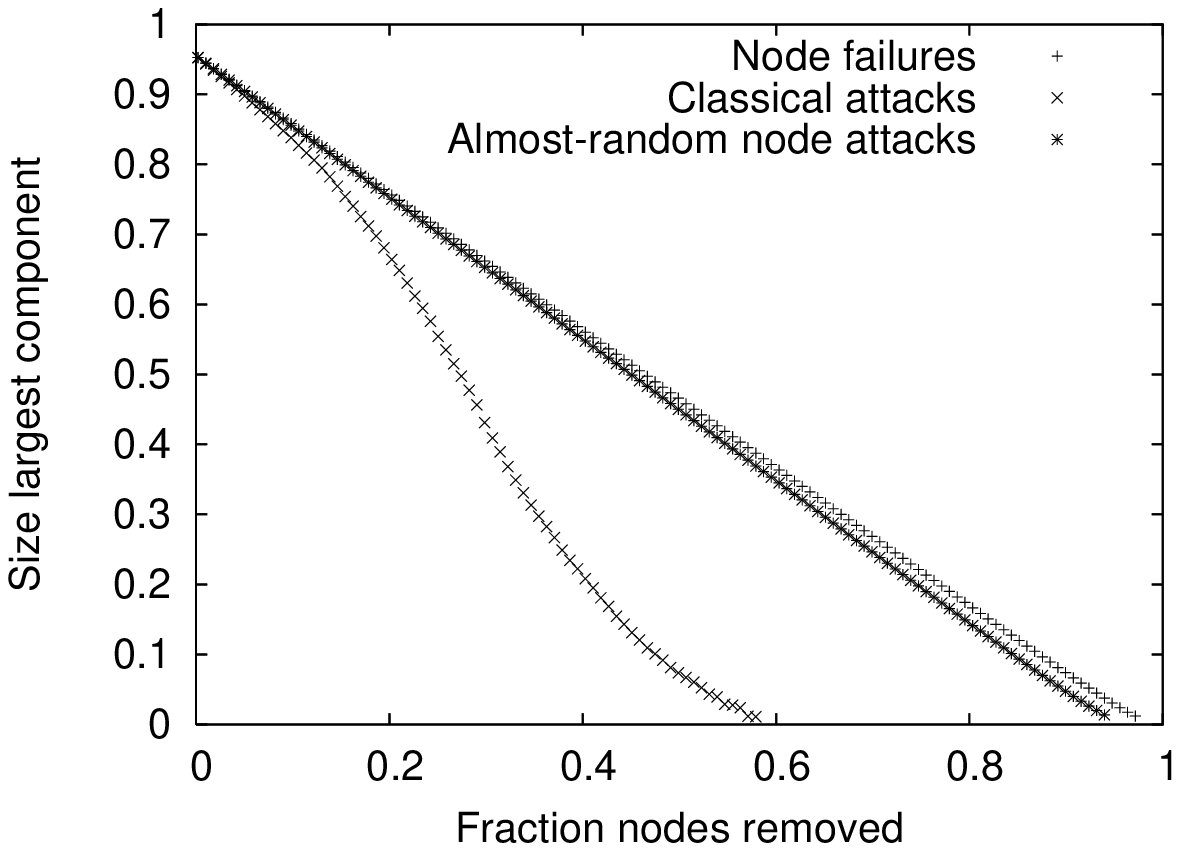}{actor}

\figreal{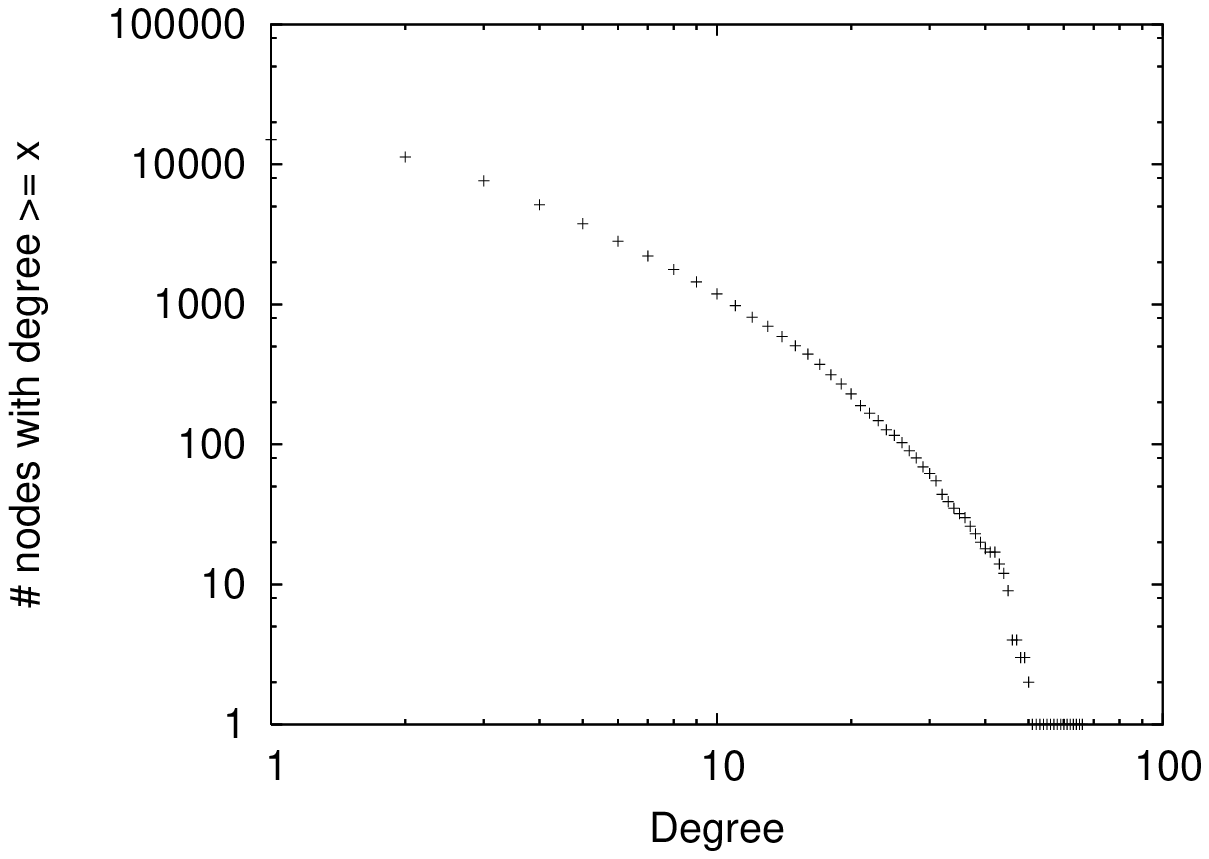}{co-authoring}

\figreal{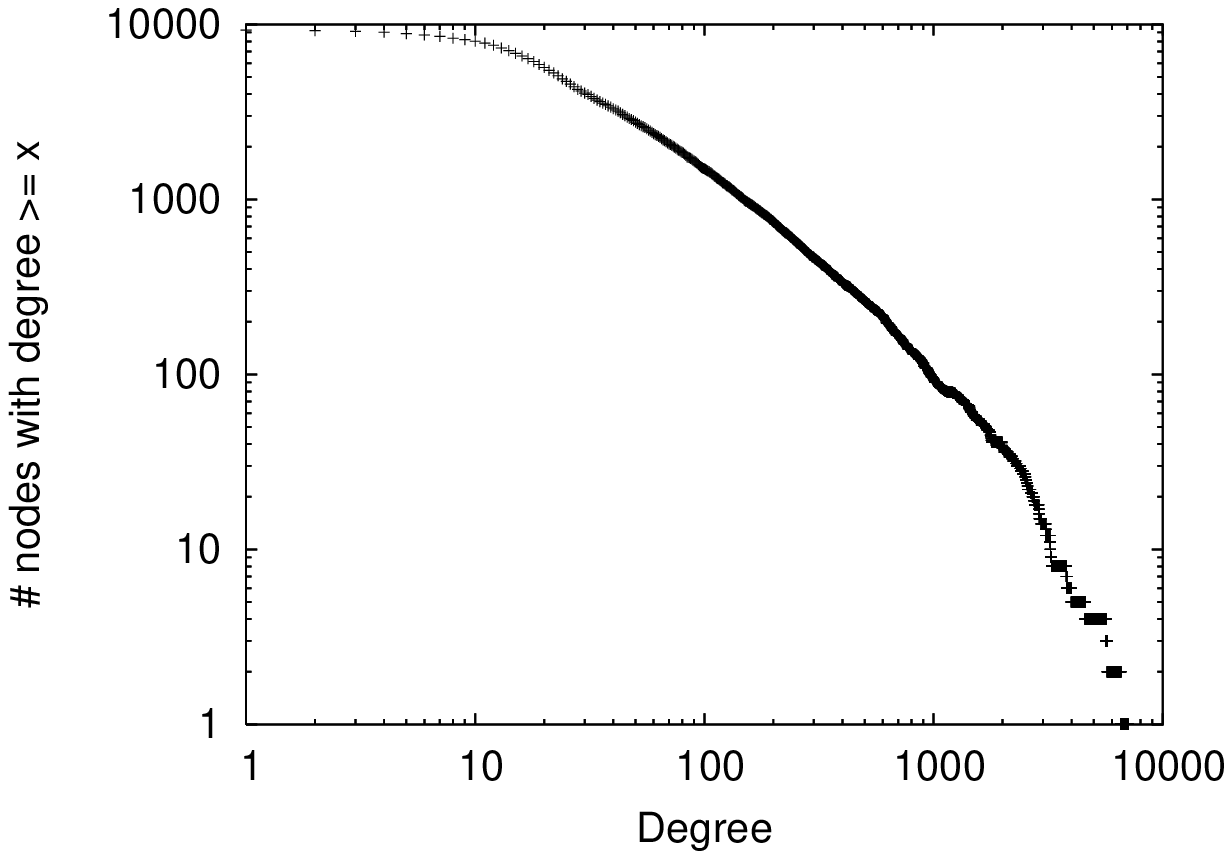}{coccurrence}

\figreal{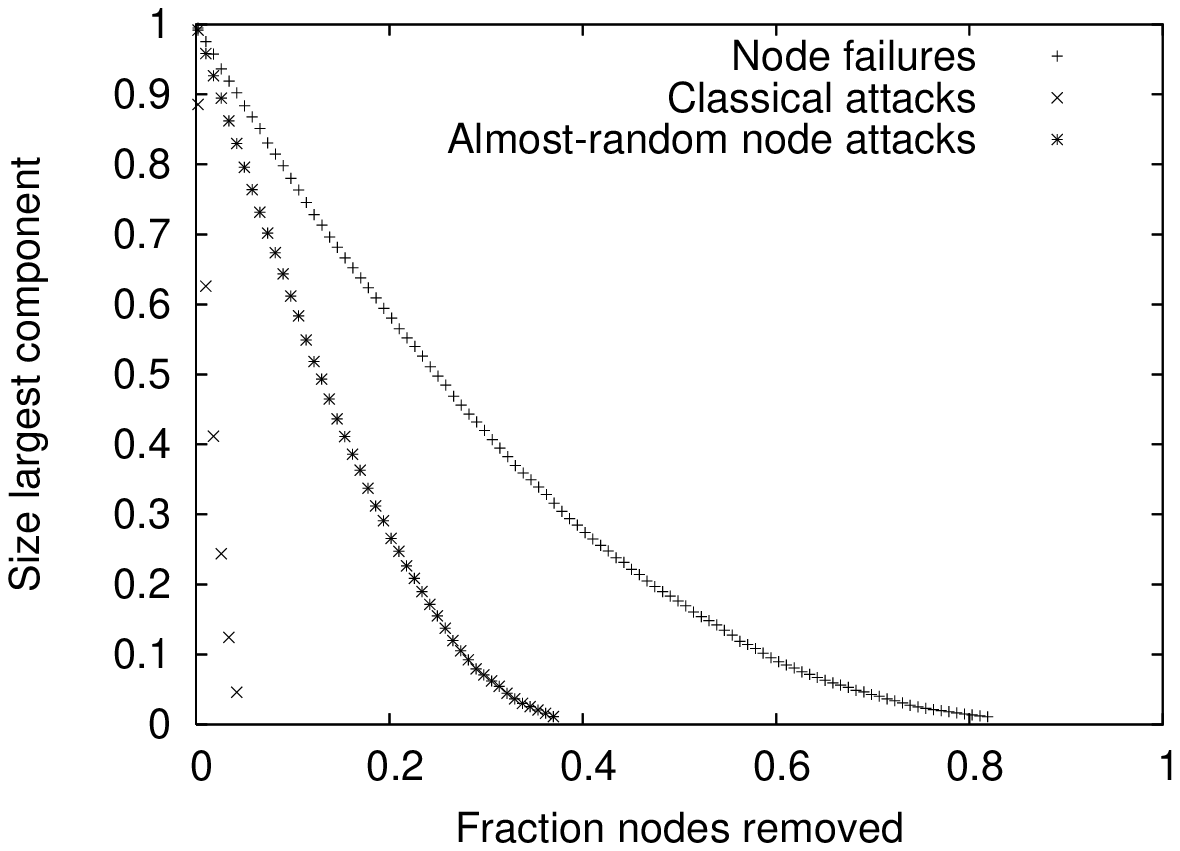}{internet1}

\figreal{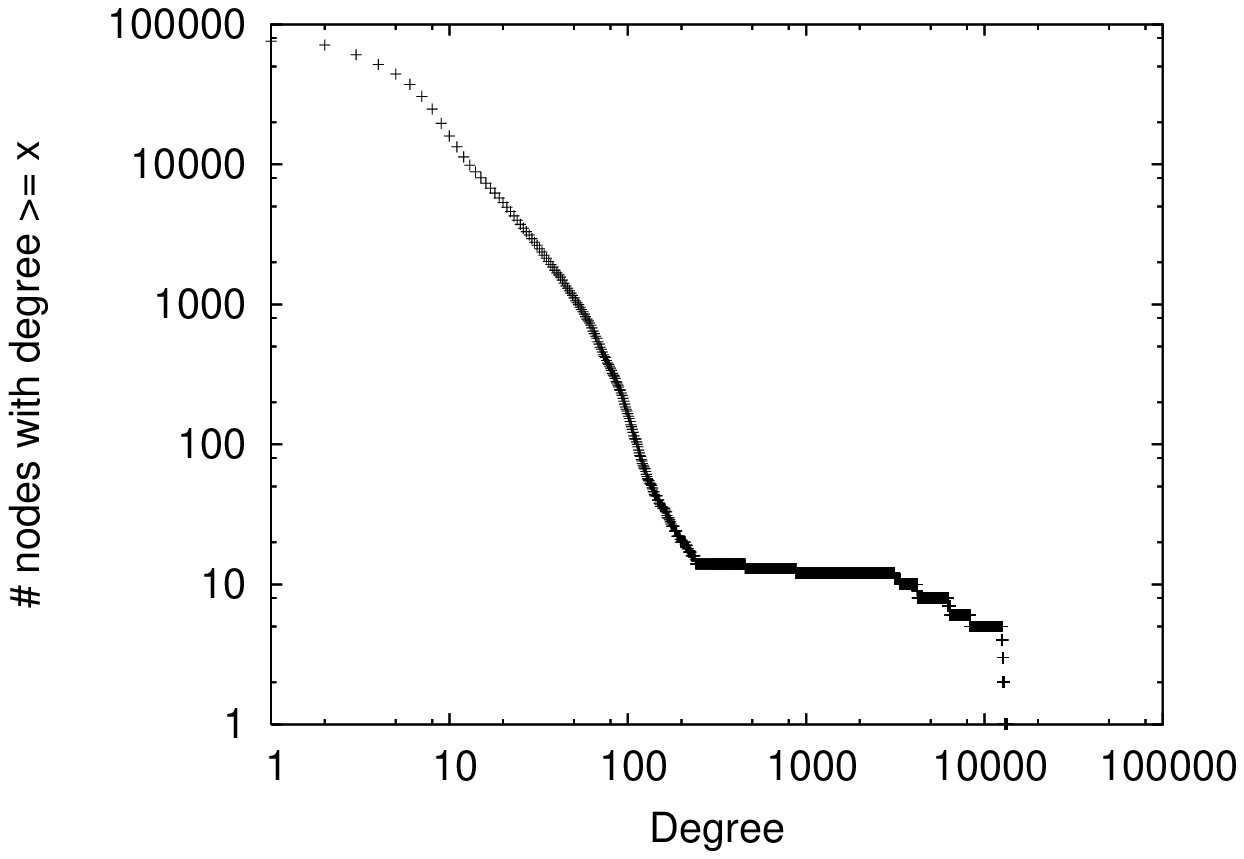}{internet-core}

\figreal{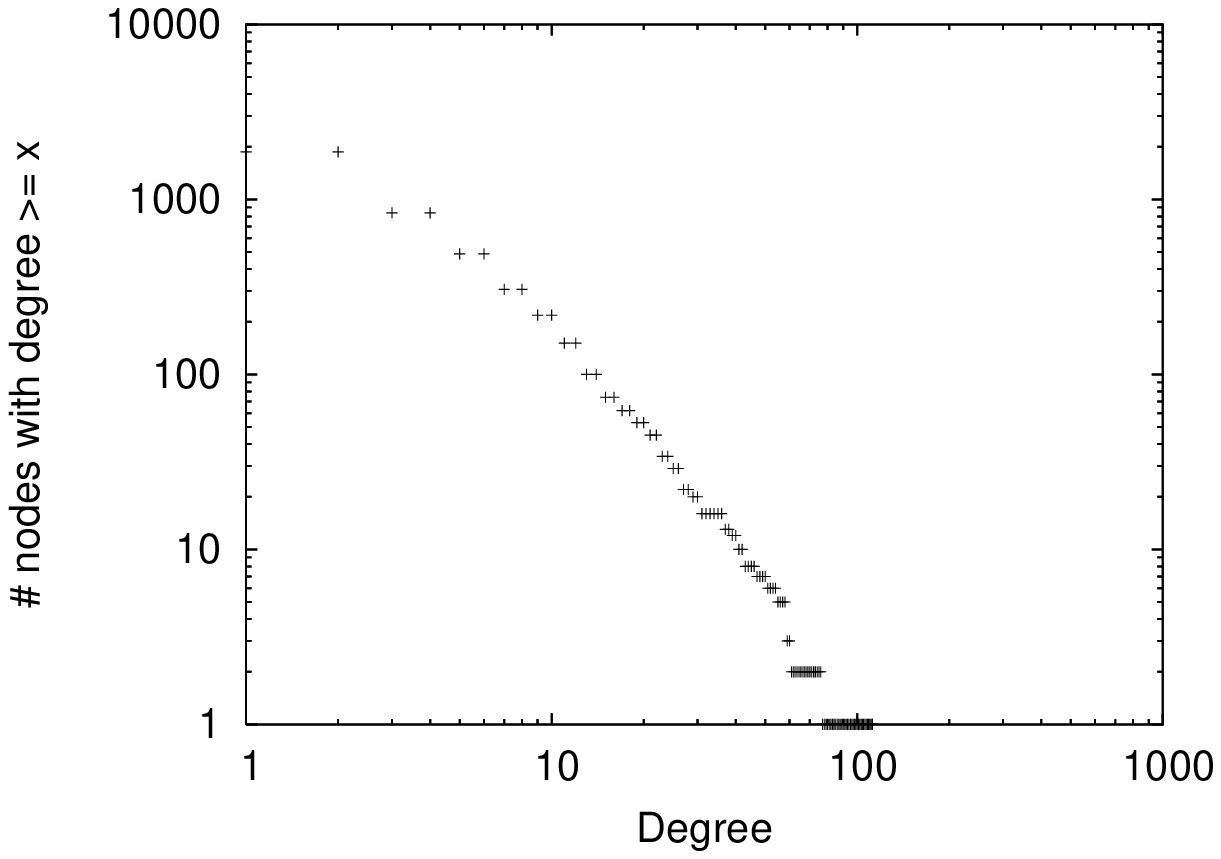}{protein}

\figreal{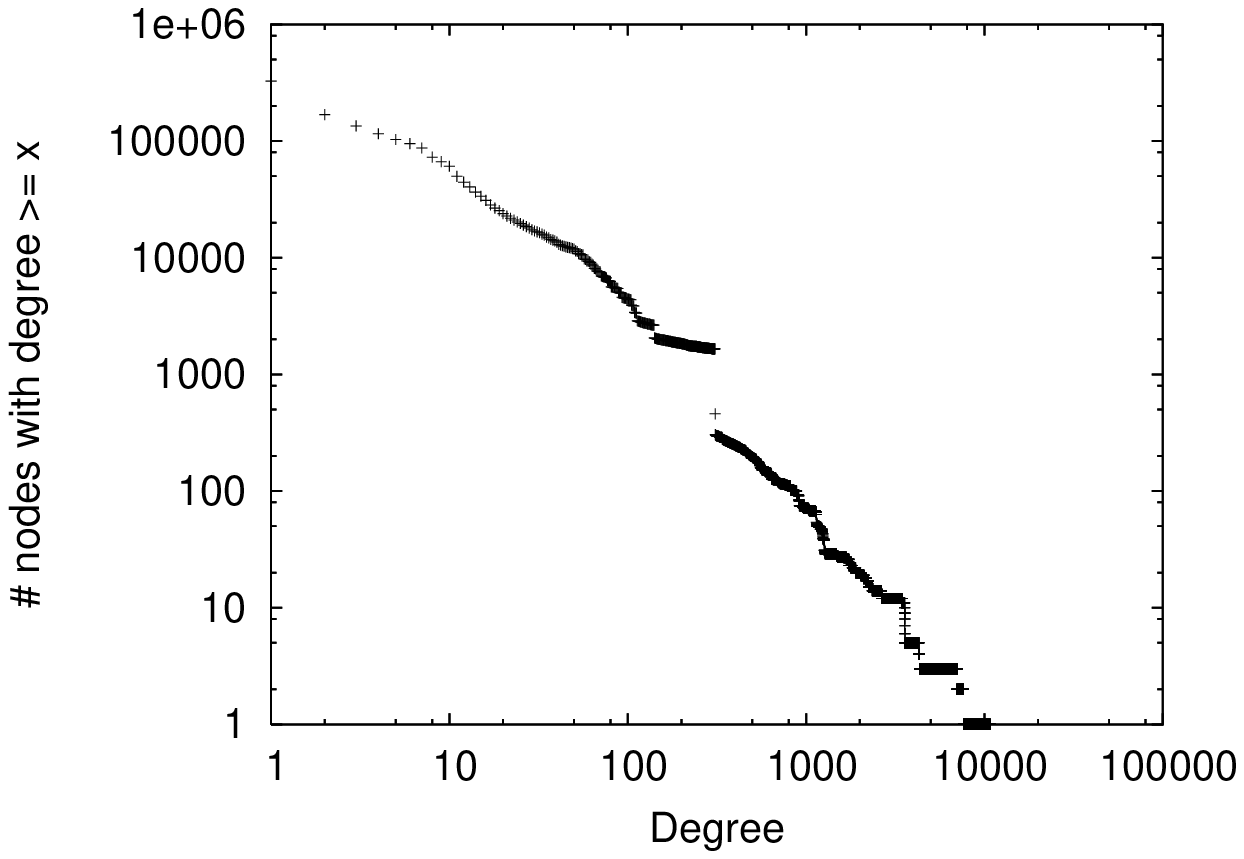}{www}

\figreal{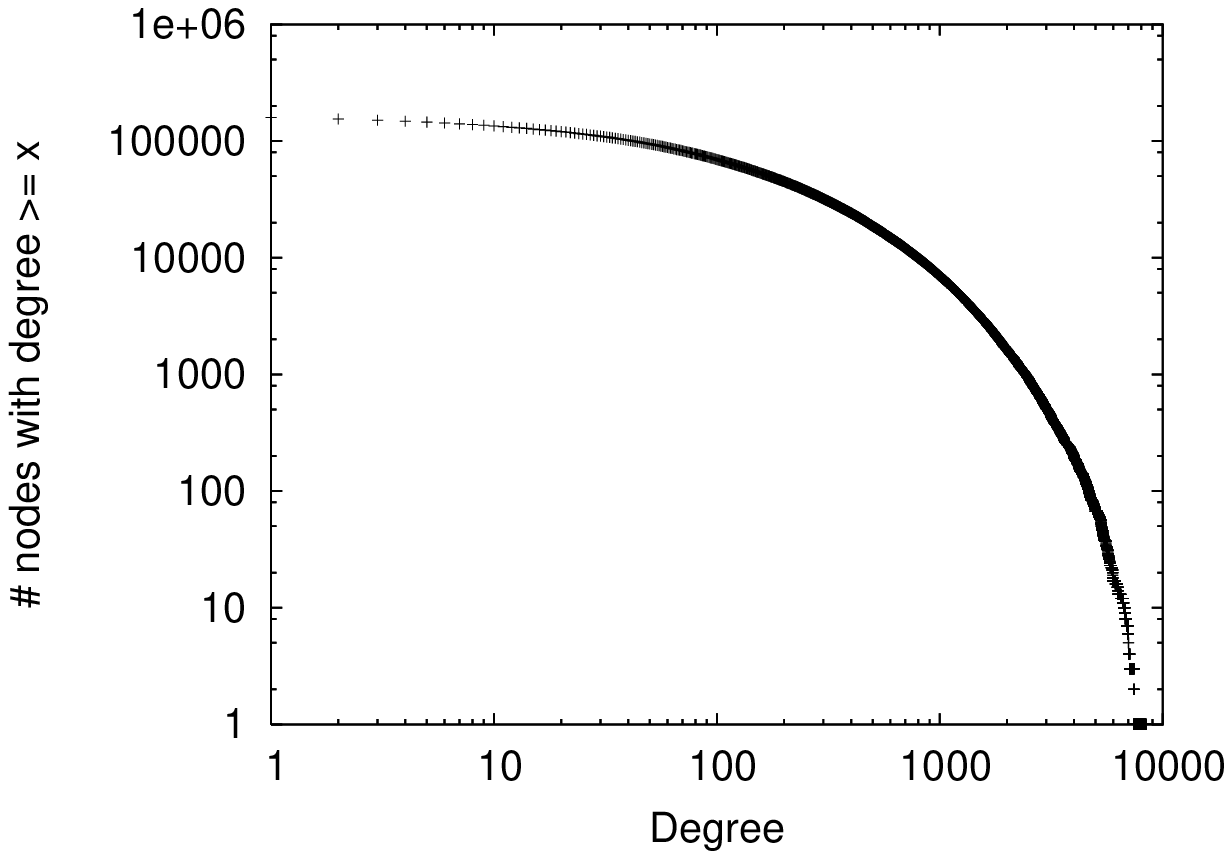}{p2p}

\medskip
First notice that none of these networks has a Poisson structure, though
their degree distribution sometimes is quite far from a power-law. In
all the cases the degree distribution is highly heterogeneous,  which
makes the discussed attack strategies relevant. 
Notice however that some networks (in particular,
the {\em actor}, {\em cooccurrence}, {\em internet-core} and {\em p2p}
networks)
have a significantly lower fraction of nodes of degree one than what can be
expected for power-law networks.

It appears that in all the cases the behavior for random node failures
and classical attacks matches the theoretical expectations: there is a significant
difference between random node failures and  classical attacks, even if
this difference is quite small in the cases of {\em actor} and {\em p2p}
networks.
The same is true for almost-random node attacks:
in most cases, they are an intermediate between random nodes failures and
classical attacks.
Notice however that in several cases, namely {\em actor}, {\em cooccurrence},
 {\em internet-core} and {\em p2p}, the behaviors in the case
 of random node failures and almost-random node attacks are
 very similar. This is due to the fact that, as pointed out above,
 there are few nodes of degree $1$ with respect to the total
 number of nodes.

Concerning link removals, we observe
two different types of behaviors.
The {\em co-authoring}, {\em internet1}, {\em protein}, and {\em www}
networks more or less conform to the expected behavior.
Notice however that the {\em www} network is very resilient to all link removal strategies.
We will come back to this point later.

The {\em  actor}, {\em cooccurrence}, {\em internet-core},
and {\em p2p} networks,
however, behave very differently from what was expected.
For these networks,  almost
 all links have to be removed to simply {\em reduce significantly} the
 size of the largest connected component,
for all strategies except node
 failures seen from the link point of view.
This can be explained as follows.
First, 
these networks have very strong average degrees:
they are the four networks with the strongest average degree.
Denser networks are naturally more resilient to link removals.
Also, notice that these networks have a relatively low fraction of
nodes of degree one, compared to the other networks.
This naturally induces that a very large fraction of the links are
attached to the highest degree nodes.
The removal of a few of these nodes does not have a strong impact
on the network, however this induces the removal of {\em most} of
the links, causing the plot for classical attacks seen from the
link point of view to be almost horizontal in the beginning.
The fact that there are few nodes of degree one also explains
why random link failures and almost-random link attacks
behave very similarly.
Finally, counter-intuitively, the most efficient link removal
strategy for these networks is node failures seen from the link
point of view (though this strategy cannot be called  efficient in itself,
since all links must be removed to break down the networks):
random node failures cause the size of the largest connected component
to decrease linearly in this case 
(meaning that the largest connected component contains all nodes except
the ones actually removed by failures).
We have seen in Proposition~\ref{prop_psa_mp} that, when a fraction $p$
of the nodes are removed by random failures,
the corresponding fraction of removed links is $2p - p^2$.
The size of the largest component as a function of the fraction of removed
links is therefore expected to evolve as $\sqrt{1-x}$,
which corresponds to the shape of the random node failures seen from the
link point of view plot.

It must be noticed finally that though the {\em  actor}, {\em cooccurrence}, {\em internet-core},
and {\em p2p} networks behave similarly and have high average degrees,
there is a large difference in
the average degree of the {\em internet-core} network and the one of the
three others.
The average degree of the {\em www} network is in fact almost the same as
the one of the {\em internet-core} network, but it behaves differently
(though it is indeed very resilient to link removal
strategies).
This means that the average degree does not uniquely determine the behavior
of a network, and that other properties are in cause.



\medskip

 The behaviors of the two maps of the internet are very different,
 which shows that one must be very careful  when deriving conclusions
 about such networks. Indeed, the measurement procedure only gives
 a partial and biased view, see~\cite{Guillaume2005IPTopology,Vespignani2005TCS,chen02origin,chang01inferring,IPconn,Moore2005Traceroute,lakhina02sampling}.
This moderates the often
 claimed assertion that the internet is very robust to failures and
 very sensitive to attacks, which has been derived from such
 experiments, typically conducted on maps of the kind of {\em
 internet1}.

 In this last case, one may notice that almost-random link attacks
 destroy the network more efficiently than classical attacks
 viewed from the link point of view. This indicates that the robustness
 of this network in case of failures is strongly due to the fact that
 the amount of nodes of degree $1$ is huge.



\medskip
In conclusion,
these experiments showed that, concerning random node failures and
classical attacks, real-world networks behave accordingly to theoretical
predictions.
When taking into account almost-random node attacks,
as well as link removal strategies, however,
more subtle behaviors occur.

This shows that all the aspects we have discussed in this
paper must be taken into account when dealing with practical cases.
This also  shows that much remains to be done
to fully understand the observed phenomena;
we will discuss this further in the conclusion.
On the other hand, one may see the study of the resilience of a given
network as a way to deepen the understanding of its structure and point out
some non-trivial features which should be explored.
This is the case,
for instance, of the remarks we made on the {\em www} network above.


%% file: conclu.tex
\section{Conclusion and discussion.}
\label{sec_conclu}

In this contribution, we focused on a set of previously known
results which received much attention in the last few years. These
results state that, although power-law networks are very resilient
to random (node or link) failures and Poisson ones are not, they
are very sensitive to a special type of attacks
(which we call {\em classical attacks}) consisting of removing the highest degree nodes first, while Poisson networks are not. This had led to the
conclusion that its power-law degree distribution may be seen as an
{\em Achille's heel of the internet}~\cite{nature00cover,Barabasi2003Handbook}.

These results were first obtained empirically~\cite{albert00error,broder00graph}, but
an important analytic effort has been made to prove them with
mean-field and asymptotic approximations~\cite{Cohen2000RandomBreakdown,Cohen2001Attack,Cohen2003Handbook,Newman2000Robustness,Newman2003Handbook}
using two different techniques.

Our first contribution is to give a unified and complete
presentation of these results (both empiric and analytic ones).
Since some of the involved techniques (in particular mean-field
approximations) are unusual in computer science, we emphasized on
the methods used and gave much more detailed proofs than in the original
papers. In particular, we pointed out the approximations where they
occur, discussed them in the light of the experiments, and
tried generally to give a didactic presentation.

\bigskip

Our second contribution is to introduce
some new results on cases which received less
attention,
maybe because these
results are less striking. They are however essential for
deepening one's understanding of this topic.
We focused in
particular on two aspects: studying the finite case,
and studying the link point of view of random node failures,
and classical attacks.
We also introduced new attacks,
very similar to random failures,
which allowed us to deepen our understanding of the phenomena at play.

Finally, we conducted extensive simulations,
 on random graphs of different types in order to confront them to theoretical results,
but also on real-world networks,
which put in evidence complex behaviors,
not attributable to the sole degree distribution.

All this showed that many of the classical conclusions of the field
should be discussed further.
We may now put all these results and their relations together to derive
global conclusions.

\medskip

Concerning random node and link failures (\ie{} random node and link removals),
the striking point is that, although analysis
predict completely different behaviors for Poisson and power-law
networks, in practice the differences, though  important, are not huge
(see Tables~\ref{tab_ps} to~\ref{tab_pa}).
This is even more pronounced for link failures.
This overestimation of the
difference was due to the study of the infinite limit
and to the approximations made. It may also be a consequence of our choice to
consider that
a network breakdown occurs when the size of the
the giant component reaches $5\%$ of all nodes,
but other conventions lead to similar conclusions.

Concerning classical attacks (\ie{} removal of nodes in decreasing order of their degree),
we have shown that, although the thresholds
for power-law networks indeed are very low, and much smaller than the ones for Poisson networks,
our other observations tend to moderate this conclusion.
Indeed, as one may have guessed, the number of links removed during a
classical attack is huge. When one considers the number of removed links,
power-law networks are not more fragile than Poisson ones.

The large number of removed links, though it clearly plays a role,
is however not sufficient to explain
the efficiency of classical attacks: if one removes the same fraction
of links randomly, then there is no breakdown.
This invalidates the often claimed explanation that
classical attacks are very efficient on power-law networks because they remove
many links.

Going further, if one removes the same, or even a smaller,
 fraction of links, but {\em almost} randomly
(\ie\ randomly among the ones which are linked to nodes of degree at
least $2$) then a breakdown occurs. In terms of the fraction
of removed links, classical attacks therefore lie between random
link failures and almost-random link attacks, which makes them not
so efficient.

Finally, the efficiency of classical attacks resides mainly in the
fact that it removes many links, and that these links are mostly
attached to nodes of degree larger than $1$.
Conversely, this explains the robustness of power-law networks to
random node failures:
 such failures often remove nodes
of degree $1$ and/or links attached to such nodes.

\medskip

Another conclusion of interest comes from the study of classical
attacks on Poisson networks (which was not done in depth until now).
Although these networks behave similarly in case of random
node failures and classical attacks, it must be noted that their
threshold is significantly lower in the second case. This goes
against the often claimed assumption that, because all nodes
have almost the same degree in a Poisson network,
there is little difference between
random node failures and classical attacks.
This is worth noticing, since it reduces the
difference, often emphasized, in the behavior of Poisson and power-law
networks.

\medskip

The observation of practical cases in Section~\ref{sec_real} also
provided interesting insights: in several cases, some observed behaviors
may be explained using the results in this paper  and our knowledge
of the properties of the underlying network. It  appears clearly
however that
other properties than degree distributions play an important role on network
resilience.
This points out interesting directions for further analysis.
Conversely, this shows that one may see the study of a particular network's resilience
as a way to obtain some insight on its structure.

\medskip

All these results led us to the conclusion that, although
random node failures and classical
attacks clearly behave differently and though the Poisson or
power-law nature of the network has a strong influence in this,
one should
be careful in deriving conclusions. This is confirmed by our
experiments on real-world networks. The sensitivity of
networks to attacks relies less on the presence of high-degree nodes
than on the fact that they have many low-degree nodes.
Conversely, their robustness to failures relies strongly on the fact that when we
choose a node at random, we choose such a node with high probability,
and not so much on the fact that high-degree nodes hold the network together.
Moreover, the fact that a classical attack on a power-law network
removes many links may be considered as partly, but not fully,
responsible for its rapid breakdown.

\bigskip

Although this paper is already quite long,
we had to make some choices in the presented results, and
there are of course many omissions.
For instance, we did not mention random networks with degree
correlations, on which interesting results exist~\cite{Boguna2003EpidemiCorrelations,Vazquez2003ResilienceCorrelations}, or other types of modeling, for instance the HOT framework~\cite{Doyle05robust,Li04first}. We also
ignored the various contributions considering other attack
strategies~\cite{lee05robustness,crucitti04error,crucitti04model,newth04evolving,Broido2002Resilience,park03static,Holme2002Vulnerability,Flajolet2002Robustness,Motter2002RangeAttacks,motter02cascade,Motter2004Cascade,Pertet2005Cascading,Zhao2004Cascading} and other definitions of the robustness than the size of the giant component~\cite{Latora2001Efficient,Holme2002Vulnerability,park03static,Crucitti2003EfficiencyScale-Free,Li04first,Doyle05robust}.
Likewise,
we could have compared real-world networks in Section~\ref{sec_real}
with random networks having exactly the same degree distribution,
which would certainly be enlightening. It would be interesting also
to compare more precisely the results obtained
with other practical definitions of a network breakdown.

Presenting and discussing all these aspects is impossible in a
reasonable space; instead, we chose to deepen the basic results of the
field, which we hope leads to a significantly improvement of our
understanding of them. This work could serve as a basis for further
deepening of some aspects we have voluntary omitted, such as the ones pointed out
above.

Going further, it must be clear that other properties than the
degree distribution, which should be investigated,
certainly play a role in real-world complex networks' robustness. This
appears clearly in Section~\ref{sec_real}, where several classes
of behaviors appear.

It seems obvious, for instance, that the
fact that nodes of real-world complex networks are organized in communities (groups of densely connected nodes), which is also captured
in part by the notion of clustering
coefficient~\cite{Blondel08fast,gibson98inferring,kumar99trawling,flake00efficient,girvan01community,flake02self-organization,Newman2004Communities,Latapy05Communities}, plays a key role. However, these properties are not captured by the random
networks  we considered here.

Studying robustness of networks with more subtle properties than
degree distributions would therefore be highly relevant, but most
remains to be done. In particular, while there is a consensus on
the modeling  of networks with a given degree distribution in the community
(through the use of the configuration model~\cite{Bender1978Configuration}, like here,
or the preferential attachment principle~\cite{barabasi99emergence}), there is no
consensus for more subtle properties like the clustering. Many
models have been proposed, but each has its own advantages and
drawbacks. Some of them seem however well suited for analysis,
as they are simple extensions of the configuration model~\cite{Guillaume2004BipartiIPL,Guillaume2004BipartiCAAN}
or of the preferential attachment one~\cite{dorogovtsev02evolution,dorogovtsev00structure}.

Finally, let us insist once more on the necessity of developing formal
results to enhance our understanding of empiric results.
It makes no doubt that experiments
(in our case  first obtained in~\cite{albert00error,broder00graph})
bring much understanding and intuition on
phenomena of interest.
The need for rigor and for a deeper understanding of what
happens during these experiments
is however strong.
It led to several approaches to analyze them. The main ones in
our context are developed in~\cite{Cohen2000RandomBreakdown,Cohen2001Attack,Cohen2003Handbook,Newman2000Robustness,Newman2003Handbook}.
They all rely on
mean-field approximations,
and we gave here the details of the underlying approximations and assumptions. Such an approach is
definitively rigorous, but is not {\em formal}. Obtaining exact results, or even approximate results,
with formal methods
would be another improvement. Some results begin
to appear in this direction~\cite{Bollobas2003Robustness}, but
much remains to be done and the task is challenging.



%% file: ack.tex
\medskip
\subsubsection*{Acknowledgments}

We thank the anonymous referees for taking the time to read this
paper in-depth and making very valuable comments for improving it.
We thank Fabien Viger for valuable comments on degree distributions.
This work was supported in part by the
European MAPAP SIP-2006-PP-221003 project,
the French ANR MAPE project,
the MetroSec ({\em Metrology of the internet for Security and quality of services}, \url{http://www.laas/fr/~METROSEC}) project
and by the GAP ({\em Graphs, Algorithms and Probabilities}) project.
